\documentclass{article}
\usepackage[utf8]{inputenc}
\usepackage[english]{babel}
\usepackage[dvipsnames]{xcolor}
\usepackage[margin=1in]{geometry}
\usepackage{mathrsfs}
\usepackage{amssymb}
\usepackage{amsmath}
\usepackage{amsthm}
\usepackage{bbm}
\usepackage{bm}
\usepackage{multirow}
\usepackage{graphicx}
\usepackage[ruled,vlined]{algorithm2e}
\usepackage{dirtytalk}
\usepackage{caption}
\usepackage{sidecap}
\usepackage{sectsty}
\sectionfont{\fontsize{12}{15}\selectfont}
\subsectionfont{\fontsize{10}{13}\selectfont}
\usepackage{natbib}
\usepackage{booktabs}
\usepackage{float}
\usepackage{adjustbox}
\usepackage{siunitx}
\usepackage{subcaption}
\usepackage{tikz}
\usetikzlibrary{arrows.meta}

\newtheorem{theorem}{Theorem}
\newtheorem{lemma}{Lemma}

\theoremstyle{definition}
\newtheorem{defn}{Definition}
\newtheorem{remark}{Remark}
\newtheorem{assum}{Assumption}

\newcommand{\RR}{\mathbb{R}}
\newcommand{\PP}{\mathbb{P}}

\newcommand{\EE}{\mathbb{E}}

\newcommand{\ZZ}{\mathbb{Z}}

\newcommand{\Var}{\mathrm{Var}}

\usepackage[utf8]{inputenc}

\usepackage[colorlinks=true,
            linkcolor=blue,
            urlcolor=blue,
            citecolor=blue]{hyperref}

\title{G-HIVE: Parameter Estimation and Approximate Inference for Multivariate Response Generalized Linear Models \\
with Hidden Variables}

\author{Inbeom Lee\thanks{ \fontsize{9pt}{0.4cm}\selectfont Booth School of Business, University of Chicago, Chicago, IL. E-mail: \texttt{inbeom.lee@chicagobooth.edu}.}~~~~~Yang Ning\thanks{\fontsize{9pt}{0.4cm}\selectfont Department of Statistics and Data Science, Cornell University, Ithaca, NY. E-mail: \texttt{yn265@cornell.edu}. }}

\date{\today}

\begin{document}

\maketitle

\begin{abstract}
In practice, there often exist unobserved variables, also termed hidden variables, associated with both the response and covariates. Existing works in the literature mostly focus on linear regression with hidden variables. However, when the regression model is non-linear, the presence of hidden variables leads to new challenges in parameter identification, estimation, and statistical inference. This paper studies multivariate response generalized linear models (GLMs) with hidden variables. We propose a unified framework for parameter estimation and statistical inference called \textsc{G-hive}, short for \underline{G}eneralized - \underline{HI}dden \underline{V}ariable adjusted \underline{E}stimation. Specifically, based on factor model assumptions, we propose a modified quasi-likelihood approach to estimate an intermediate parameter,  defined through a set of reweighted estimating equations. The key of our approach is to construct the proper weight, so that the first-order asymptotic bias of the estimator can be removed by orthogonal projection. Moreover, we propose an approximate inference framework for uncertainty quantification. 
Theoretically, we establish the first-order and second-order asymptotic bias and the convergence rate of our estimator. In addition, we characterize the accuracy of the Gaussian approximation of our estimator via the Berry–Esseen bound, which justifies the validity of the proposed approximate inference approach. Extensive simulations and real data analysis results show that \textsc{G-hive} is feasibly implementable and can outperform the baseline method that ignores hidden variables.

\end{abstract}

\section{Introduction}
In many regression problems, due to measurement limitations or ethical considerations, there often exist unobserved variables, also referred to as hidden variables. For example, in the analysis of high-throughput genomic data, researchers have long been aware that the measurements can be affected by many unobserved factors such as laboratory conditions, preparation time, and reagent lots \citep{irizarry2005multiple,luo2019batch}. These factors are called batch effects, which can be modeled as hidden variables \citep{leek2007capturing}. Similarly, in biomedical studies, the onset of a disease is likely associated with several unmeasured variables, such as environmental factors or habitual patterns \citep{katsaouni2021machine}. Ignoring hidden variables in the statistical analysis may introduce estimation bias and potentially lead to misleading scientific conclusions. Therefore, there is a pressing need to develop statistical methods that deal with hidden variables in a general regression framework, and that in particular, are applicable to binary or categorical data. 

This paper studies the multivariate response generalized linear model with hidden variables. Specifically, we assume that the $M$-dimensional response variable $Y=(Y_1,...,Y_M)^T$ given the observed covariates $X \in \RR^{p}$ and hidden variables $Z \in \RR^{K}$ follows the generalized linear model (GLM) with the canonical link
\begin{align}
f(Y_m|X,Z)=\exp\big[\{Y_m\cdot(\Theta_m X +B_m Z) - b(\Theta_m X + B_m Z)\}/\phi ~+ ~c(Y_m, \phi)\big] \label{glm}
\end{align}
where $b(\cdot)$ and $c(\cdot)$ are known functions and $\phi$ is the dispersion parameter. The parameters $\Theta_m$ and $B_m$ are the $m$-th row of coefficient matrices $\Theta \in \RR^{M \times p}$ and $B \in \RR^{M \times K}$, respectively. Given $n$ i.i.d copies of $(Y, X)$, we are interested in the estimation and inference of the coefficient matrix $\Theta$, the association between $X$ and $Y$ in the presence of hidden variables $Z$. In this work, we consider the regime where $p, M, K$ are all allowed to grow with the sample size $n$, where $p \leq n$ and $K \leq M$ hold.

In this work, we propose a unified framework for parameter estimation and statistical inference called \textsc{G-hive}, short for \underline{G}eneralized - \underline{HI}dden \underline{V}ariable adjusted \underline{E}stimation. Since the hidden variable $Z$ is random and unobserved, the coefficient matrix $\Theta$ is generally not identifiable. To make $\Theta$ (asymptotically) identifiable and estimable, we impose a factor model in (\ref{factor_model}) that relates $X$ and $Z$ \citep{bai2003inferential,fan2013large,fan2008high}. However, under these model assumptions, the distribution of $Y$ given $X$ does not follow a GLM and is indeed intractable, since we do not impose any parametric assumption on the distribution of $Z$.  To overcome this challenge, we carefully construct a modified quasi-likelihood for a new estimand $F^*$, which is defined through a set of reweighted estimating equations. The rationale behind the reweighted estimating equations is that the resulting first-order approximation of the bias $F^*-\Theta$ is shown to belong to the column space of $B$, which can be removed by estimating the projection matrix $P_B=B(B^TB)^{-1}B^T$. The intuition of our approach is explained in Section \ref{sec_ident}. Under certain identifiability conditions, using $F^*$ as a bridge, we introduce the estimator $\hat\Theta=\hat{P}_B^{\perp}\hat{F}$, where $\hat{P}_B$ is obtained by applying PCA to a carefully constructed weighted covariance matrix and $\hat{F}$ is the maximum modified quasi-likelihood estimator. Since our approach yields a tight pipeline, it is straightforward to implement in practice. 

In addition, we propose an approximate inference framework for uncertainty quantification. Unlike classical inference results for models without hidden variables, a new, unpleasant phenomenon of the estimator $\hat\Theta$ is that the asymptotic bias may dominate the stochastic error. Consequently, the limiting distribution of the estimator $\hat\Theta$ is no longer centered at $\Theta$ with the $\sqrt{n}-$rate. To address this issue,  we shift the target parameter from $\Theta$ to $P_B^{\perp}F^*$ (or $F^*$), which corresponds to the second-order (or first-order) approximation of $\Theta$. Intuitively, $P_B^{\perp}F^*$ can be viewed as the correct  limiting value of $\hat\Theta$, and therefore a confidence interval based on the limiting distribution of $\hat\Theta$ yields the desired coverage probability for $P_B^{\perp}F^*$. For this reason, we refer to this approach as second-order approximate inference.

Theoretically, our first key result shows that the approximation bias satisfies $\|F^*-\Theta\|_F /\sqrt{M}=O(1/\sqrt{p})$ and $\|P_B^{\perp}F^*-\Theta\|_F/\sqrt{M} = O(1/p)$. An interesting implication of this result is that collecting more observed covariates can mitigate the approximation bias. Moreover, it also explains why $P_B^{\perp}F^*$ (or $F^*$) is called the second-order (or first-order) approximation of $\Theta$ in our inference framework. Next, we establish the convergence rate of the estimator $\hat\Theta$. In particular, we show that, under mild conditions (e.g., $M$ is large enough), the convergence rate of $\hat\Theta$ is faster than that of $\hat F$, which corresponds to the baseline naive estimator that ignores the hidden variables. Finally, we characterize the accuracy of the Gaussian approximation of $\hat\Theta$ via the Berry–Esseen bound, which justifies the validity of the proposed approximate inference approach.

\subsection{Related Literature}\label{lit_review}
This work is most related to surrogate variable analysis (SVA) proposed by \cite{leek2007capturing}, and more recently developed by \cite{lee2017improved,wang2017confounder,mckennan2019accounting,bing2022adaptive,bing2023inference}, the last of which proposed a novel factor model based bias correction approach for multivariate response linear regression with hidden variables, which is a special case of GLMs. However, their approach is only applicable to linear regression. The challenge of extending their approach to GLMs is detailed in Section \ref{sec_backgroud}. Compared to these works, our main methodological novelty is that we propose to calibrate the residual by an approximate inverse variance weighting scheme. Such a calibration step is essential under the GLM for parameter identification, estimation consistency, and asymptotic normality. Theoretically, we discover a unique result in that our estimator under the GLM inherently has an asymptotic bias which decreases with $p$ but may still dominate the stochastic error. As a result, there is an interesting and much more delicate interplay between $p$ and $M$ in both the estimation error and the Berry–Esseen bound for Gaussian approximation. 

Along this line, a recent work by \cite{du2025simultaneous} studied simultaneous inference with unmeasured confounders when $p\gg n$. While they focused on the same GLM as in our (\ref{glm}), their imposed model for the unmeasured confounders is different from our factor model in (\ref{factor_model}), and consequently the corresponding assumptions on their model are different from our Assumption \ref{assum_glm}. In particular, our theory is established under a more challenging setting, where the covariance matrix of $X$ has spiked  eigenvalues. Their proposed method is a joint maximum likelihood approach, which requires the estimation of all coefficient matrices as well as the latent factors $Z$ for each sample. In contrast, our method is computationally more convenient and avoids estimating the unknown factors.

Another recent direction of relevance, grouped together under the term spectral deconfounding, includes work by \cite{cevid2020spectral,guo2022doubly,fan2024latent,wang2025latent,sun2024decorrelating} among others, and considers estimation and inference in high-dimensional regression models with unmeasured confounders. For example, \cite{ouyang2023high} focused on inference in high-dimensional GLMs with unmeasured confounders by generalizing the decorrelated score approach from \cite{ning2017general} to account for the effects induced by the unmeasured confounders. 
This work is similar to ours, but fundamentally different in that their response is assumed to be univariate. In contrast, we show that with multiple response variables, we can estimate the parameters in a collaborative way, improving the convergence rate compared to the univariate case.

Alternatively, one may view hidden variables as random effects or latent factors. The usage of random effects or latent variables in GLMs have many different forms in the literature \citep{bartholomew2011latent,mcculloch2001generalized}. For example, \cite{huber2004estimation}  introduced generalized linear latent variable models without any observed covariates, and proposed a Laplacian approximation to estimate the coefficient of the latent variable. A similar approach was also considered for generalized linear mixed effect models  \citep{breslow1993approximate}. All these works differ substantially from our approach.

%\subsection{Main Contributions}
%The main contributions of this paper are the following - first, we introduce a model based on a reweighting scheme. The true residuals are reweighted and the parameters of the misspecified model, $F_m^*$, are defined utilizing reweighted residuals. This reweighting is crucial as it facilitates the theoretical analysis for the non-linear model that would otherwise be intractable. Second, we propose a new two-step estimation procedure, G-HIVE, that incorporates an estimating equations step (the EE step) so that for each $m$ in $1 \leq m \leq M$ we obtain $\hat{F}_m$, the best estimate of $\Theta_m$ under a misspecified model that doesn't account for the hidden variables. We also obtain the residuals $\hat{\epsilon}_m$ using the best linear predictor $\hat{F}_mX$. Then, a spectral decomposition step (the PCA step) is used on the second sample moment of the residuals, $\hat{\epsilon}\hat{\epsilon}^T/n$, to extract information on $B$. More specifically, $\hat{P}^{\perp}_{B}$, an estimate of the projection matrix corresponding to the space orthogonal to the column space of $B$ is constructed. Finally, $\hat{\Delta} = \hat{P}_B^{\perp}\hat{F}$ is constructed from the estimates from the above steps. We also provide a theoretical analysis of each step in the estimation procedure which leads to a final deviation bound, $||\hat{\Delta}-\Delta||/\sqrt{M}$, which is stochastically bounded under reasonable assumptions and the above regime for $n,p,M$ and $K$. 

\subsection{Notation}
For any vector $v\in \RR^d$ and some real number $q \geq 0$, we define its $L_q$ norm as $||v||_q=(\sum_{j=1}^d |v_j|^q)^{1/q}$. For any matrix $H \in \RR^{d_1 \times d_2}$, we denote by $||H||_{\text{op}}$ and $||H||_F$ the operator norm and the Frobenius norm, respectively. $||H||_{\infty}=\max_i \sum_j |h_{ij}|$ denotes the maximum absolute row sum. Following the notation in \cite{vershynin2018high}, for any sub-Gaussian random variable (or vector) $h_2$, let $||h_2||_{\psi_2}$ denote its sub-Gaussian norm, and for any sub-exponential random variable (or vector) $h_1$, let $||h_1||_{\psi_1}$ denote its sub-exponential norm. For any symmetric matrix $H$, we write $\lambda_{k}(H)$ to denote its $k$-th largest eigenvalue, and $\lambda_{\min}(H)$ and $\lambda_{\max}(H)$ for the smallest and largest eigenvalues, respectively. For any two sequences $a_n$ and $b_n$, we write $a_n \lesssim b_n$ if there exists some fixed positive constant $C$ such that $a_n \leq Cb_n$. We also use the following notation to refer to the maximum and minimum: $a \vee b = \max(a,b)$, $a \wedge b = \min(a,b)$. %We denote $\Sigma_X=\Var(X)=\EE(XX^T)$, $\Sigma_Z=\Var(Z)=\EE(ZZ^T)$, $\Sigma_W=\Var(W)=\EE(WW^T)$ since we assume without loss of generality that $X,Z,W$ are centered. 
%Let $\XX$ denote the observed $p \times n$ covariate matrix, $\YY$ denote the observed $M \times n$ response matrix and $\ZZ$ denote the unobserved $K \times n$ matrix.

%We also expect the final estimation rate to be strictly faster than the rate obtained by just using the first MLE step.

\section{Informal Analysis of Parameter Identifiability}\label{sec_infor}
In this section, we first introduce our model setup and highlight the challenges of model identifiability under the GLM compared to the linear case, and afterwards we offer the intuition of our proposed approach.  

\subsection{Model Setup and Background}\label{sec_backgroud}
Recall that given the observed covariates $X \in \RR^{p}$ and hidden variables $Z \in \RR^{K}$, the response variable $Y=(Y_1,...,Y_M)^T$ follows the GLM (\ref{glm}). For simplicity, we assume that $\phi$ is known (set $\phi=1$) and that $X$ and $Z$ have zero mean. The parameters $\Theta_m$ and $B_m$ are the $m$-th row of coefficient matrices $\Theta \in \RR^{M \times p}$ and $B \in \RR^{M \times K}$, respectively. Without loss of generality, we assume rank$(B)=K<M$ since if $B$ is not of full column rank we can always reduce the dimensions of $Z$ such that the full column rank condition is met. Finally, we assume $Y_m$ and $Y_{m'}$ are independent given $X$ and $Z$ for $m'\neq m$. %In this model, we are interested in the estimation and inference on the coefficient matrix $\Theta$, the association between $X$ and $Y$ in the presence of hidden variable $Z$.   

To characterize the effect of hidden variables, we assume the following factor model \citep{bai2003inferential,fan2013large,fan2008high} that relates $X$, the observed variables, to $Z$, the hidden variables:
\begin{align}
    X = AZ + W, \label{factor_model}
\end{align}
where the noise term $W\in\RR^p$ has zero mean and is independent of $Z$, and $A \in \RR^{p \times K}$ is a matrix of unknown parameters. In this paper, we focus on the independent and homogeneous noise setting where $\Sigma_W=\EE(WW^T) = \tau I_p$, and without loss of generality, we set $\tau=1$. The proposed method can be easily extended to the dependent noise setting, provided the smallest and largest eigenvalues of $\Sigma_W$ are bounded from below and above by some constants.

To understand the challenge in the identifiability of $\Theta$, we first consider a special case of (\ref{glm}) in which $Y_m$ follows the linear regression model $Y_m=\Theta_m X +B_m Z+E_m$, where $E_m$ is the random noise. As shown in \cite{bing2022adaptive}, the model can be rewritten as $Y_m=(\Theta_m+B_mL) X +\epsilon_m$, where $L=\EE(ZX^T)\{\EE(XX^T)\}^{-1}$ is obtained by $L_2(P)$ projecting $Z$ onto the linear space generated by $X$ and $\epsilon_m=B_m (Z-LX)+E_m$. As a result, ignoring the hidden variable $Z$ and regressing $Y_m$ on $X$ leads to a biased estimator of $\Theta_m$. In fact, it is easily seen that the coefficient matrix $\Theta+BL$ can be identified via the first two moments of $(Y,X)$. To establish the identifiability of $\Theta$, a very natural idea is to separate $\Theta$ and $BL$ in the coefficient matrix $\Theta+BL$. In the literature, a commonly used identifiability assumption for $\Theta$ is  $P_{B}\Theta=0$ \citep{bing2022adaptive,lee2017improved,wang2017confounder}, where $P_B=B(B^TB)^{-1}B^T\in\RR^{M\times M}$ is the projection matrix onto the column space of $B$. Under this assumption, the two matrices  
$\Theta$ and $BL$ belong to two orthogonal spaces, and therefore we can identify $\Theta$ via $\Theta=P_B^\perp(\Theta+BL)$ where $P_B^\perp=I_M-P_B$, provided $P_B$ is identifiable. In this case, $\Theta$ can be naturally interpreted as the association between $X$ and $Y$ that cannot be explained via the hidden variables.  

Nevertheless, the above analysis suffers from the following two challenges when extended to the GLM setting in (\ref{glm}). First, unlike linear regression, for the GLM, the parameter obtained by regressing $Y_m$ on $X$ does not have a simple closed form, and therefore the relationship between this parameter and the parameter of interest $\Theta$ is unclear. Following the classical literature on misspecified models \citep{white1982maximum}, a routine approach is to define the pseudo-true parameter as
    \begin{align}\label{eq_score}
        F_m^{\textsc{MLE}} ~=~ \underset{F_m \in \RR^{p}}{\text{arg max}}~ \EE \Big\{Y_m\cdot(F_mX) - b(F_m X)\Big\}.
    \end{align}
Since in general $b(\cdot)$ is not a quadratic function,  $F_m^{\textsc{MLE}}$ does not have a simple closed form solution, which complicates the analysis of the identifiability of $\Theta$. To overcome this challenge, we first focus on quantifying the approximation bias of $F_m^{\textsc{MLE}}$ locally around the target parameter $\Theta_m$. Following the logic similar to the proof of Theorem \ref{thm_expectation} in Section \ref{theory}, we can establish that
    \begin{align}\label{eq_taylor_MLE}
        F_m^{\textsc{MLE}}-\Theta_m=\EE\Big\{b''(\Theta_m X +B_m Z)B_mZZ^T\Big\} A^T\Big\{\EE(b''(\Theta_m X +B_m Z)XX^T)\Big\}^{-1}+Rem_m,
    \end{align}
where the first term on the right hand side corresponds to the first-order bias of $F_m^{\textsc{MLE}}$ and $Rem_m$ represents the approximation error which is of a smaller order. We note that the factor model (\ref{factor_model}) plays a pivotal role in deriving the leading bias term and quantifying the magnitude of $Rem_m$. In contrast, under linear regression, we have $F_m^{\textsc{MLE}}-\Theta_m=B_m\EE(ZX^T)\{\EE(XX^T)\}^{-1}$, which does not require the factor model (\ref{factor_model}) to hold as this result is derived purely from the $L_2(P)$ projection of $Z$ onto the linear space of $X$. 

To identify $\Theta$, our next step is to separate $\Theta_m$ and the first-order bias in the decomposition of $F_m^{\text{MLE}}$ in (\ref{eq_taylor_MLE}). This brings us to the second major challenge in extending to the GLM setting, which is that the first order bias term of $F^{\text{MLE}}$ no longer lives in the column space of $B$. To see this more clearly, following the analysis used in linear regression, we stack the first-order biases in (\ref{eq_taylor_MLE}) over $1 \leq m \leq M$ into a matrix $\EE(DBZZ^T)A^T\EE(DXX^T)$, where $D\in\RR^{M\times M}$ is a diagonal matrix with the $m$th entry being $b''(\Theta_m X +B_m Z)$. Ignoring the $Rem_m$ term, we can write $F^{\textsc{MLE}}\approx\Theta+\EE(DBZZ^T)A^T\EE(DXX^T)$, where the matrix $F^{\textsc{MLE}}$  is identifiable row-by-row through (\ref{eq_score}). Unfortunately, we are not able to separate $\Theta$ and $\EE(DBZZ^T)A^T\EE(DXX^T)$ as in the linear case since $\EE(DBZZ^T)A^T\EE(DXX^T)$ is no longer in the column space of $B$ due to the presence of the matrix $D$. Consequently, under the same assumption $P_{B}\Theta=0$, the non-orthogonality of $\Theta$ and $\EE(DBZZ^T)A^T\EE(DXX^T)$ implies 
 $$
 P_B^\perp F^{\textsc{MLE}}\approx P_B^\perp(\Theta+\EE(DBZZ^T)A^T\EE(DXX^T))\neq \Theta,
 $$ 
which then results in $\Theta$ not being identifiable via $P_B^\perp F^{\textsc{MLE}}$.

\subsection{Our proposed approach}\label{sec_ident}
To address the identifiability problem, our main idea is to construct a properly weighted score function of the misspecified GLM to restore the orthogonality between $\Theta$ and the corresponding first-order bias. More precisely, for each $1 \leq m \leq M$, we define the  parameter $F_m^* \in \RR^{1 \times p}$ as the solution of the following estimating equation:
\begin{align}
    &\EE\bigg[\bigg\{\frac{Y_m - b'(F_m^* X)}{b''(F_m^* X)}\bigg\}X^T \bigg]=0. \label{score_reweighted}
\end{align} 
Compared to the score function from (\ref{eq_score}), the estimating equation (\ref{score_reweighted}) contains a denominator $b''(F_m^* X)$, which can be viewed as an approximation of the variance of $Y_m$ under the GLM, i.e., $\Var(Y_m|X,Z)=b''(\Theta_m X +B_m Z)$. Thus, we can also interpret (\ref{score_reweighted}) as an inverse variance weighted score function. Since the GLM in (\ref{eq_score}) is misspecified, in general we have $F_m^*\neq F_m^{\textsc{MLE}}$. The rationale behind the weighted score approach is that the first-order bias of $F_m^*$ will have a more desirable form, which will facilitate the analysis of the identifiability of $\Theta$. Indeed, Theorem \ref{thm_expectation} shows that
    \begin{align}\label{eq_taylor_weighted_score}
        F_m^*-\Theta_m=B_m\EE(ZZ^T) A^T\Big\{\EE(XX^T)\Big\}^{-1}+Rem_m',
    \end{align}
where $Rem_m'$ presents the remainder in the expansion which is of a smaller order. It is easily seen that the first-order bias on the right hand side of (\ref{eq_taylor_weighted_score}), when stacked satisfies $P_B^\perp B~\EE(ZZ^T) A^T\{\EE(XX^T)\}^{-1}=0$. Under the assumption  $P_{B}\Theta=0$,  we can show that, ignoring the $Rem_m'$ term,
$$
P_B^\perp F^*\approx P_B^\perp(\Theta+B_m\EE(ZZ^T) A^T\{\EE(XX^T)\}^{-1})= \Theta.
$$ 
As a result, we can asymptotically identify $\Theta$, provided $P_B$ is identifiable and the $Rem_m'$ term is asymptotically negligible.  
    
To justify the identifiability of $P_B$, we similarly define the inverse variance weighted residual as
\begin{align}
    \bar{\epsilon}_m :&= \frac{Y_m - b'(F_m^* X)}{b''(F_m^* X)},~~~~ \bar{\epsilon} := \big[\bar{\epsilon}_1,...,\bar{\epsilon}_M\big]^T.\label{eq_pseudo_error}
\end{align}
As shown in the proof of Theorem \ref{thm_projection}, we have
\begin{align}
\EE(\bar{\epsilon}\bar{\epsilon}^T)=\EE({\epsilon}{\epsilon}^T)+B\EE\Big[(Z-\Gamma X)(Z-\Gamma X)^T\Big]B^T+Rem'',\label{eq_var_pca}
\end{align}
where $\Gamma=\EE(ZZ^T) A^T\{\EE(XX^T)\}^{-1}$, 
\begin{align}
\epsilon_m := \frac{Y_m-b'(\Theta_m X + B_mZ)}{b''(\Theta_m X + B_mZ)},~~~~ \epsilon := \big[{\epsilon}_1,...,{\epsilon}_M\big]^T \label{true_error}
\end{align}
and $Rem''$ denotes the remainder term induced by the second-order term $Rem'$ in (\ref{eq_taylor_weighted_score}). Recall that the first term on the right hand side of (\ref{eq_var_pca}), $\EE({\epsilon}{\epsilon}^T)$, is a diagonal matrix. Consider the singular value decomposition of $B=V\Lambda U^T$ where $V \in \RR^{M \times K}$ and $U \in \RR^{K \times K}$ consist of the left and right singular vectors of $B$, respectively, and $\Lambda$ is the diagonal matrix of non-increasing singular values. From (\ref{eq_var_pca}), under the pervasiveness assumption in the factor model literature \citep{bai2003inferential,fan2013large,fan2008high}, we can (asymptotically) recover $V$ by applying spectral decomposition on $\EE(\bar{\epsilon}\bar{\epsilon}^T)$ and obtaining the first $K$ eigenvectors. Since we can verify $P_B = VV^T$, the projection matrix $P_B$ is identifiable.

Compared to the analysis of the identifiability of $P_B$ in linear regression \citep{bing2022adaptive}, our argument differs in the following two ways. First, the decomposition in (\ref{eq_var_pca}) is applied to the covariance of the inverse variance weighted residual, rather than the residual itself. Again, this reweighting approach  is crucial to ensure that the column space of $B$ can be identified by the spectral decomposition of a proper covariance matrix. Second,  the non-linear property of $b'(\cdot)$ requires a more careful analysis of the matrix perturbation errors in (\ref{eq_var_pca}) to apply the Davis-Kahan Theorem. In particular, the sample version of $Rem''$ in (\ref{eq_var_pca}) and the plug-in estimators of $F^*_m$ are correlated, leading to a slower rate of convergence. To address this technical challenge, we rely on cross-fitting and data splitting to facilitate the theory.

%cancel out the $b''(\cdot)$ term from the Taylor expansion when deriving (\ref{eq_var_pca}), which makes $P_B$ identifiable via PCA. 

\section{Parameter Estimation and Approximate Inference: G-HIVE}\label{sec_est_inf}
Recall that given $n$ i.i.d. observations $(Y^{(i)},X^{(i)})$ of $(Y,X)$, $i=1,...,n$, our goal is to estimate and do inference on $\Theta$. In this section, we present our estimation and inference procedure called \textsc{G-hive}, short for \underline{G}eneralized - \underline{HI}dden \underline{V}ariable adjusted \underline{E}stimation. The algorithm,  inspired by the identifiability of $\Theta$ detailed in Section \ref{sec_ident}, is summarized in Algorithm \ref{GHive}.

\subsection{Parameter Estimation}\label{sec_est}
In this algorithm, we first randomly split the data into two folds $D_1$ and $D_2$. We  use the data in $D_1$ to estimate the pseudo-true parameter $F_m^*$ in (\ref{score_reweighted}). A straightforward approach is to solve the estimating equation (\ref{score_reweighted}) with the expectation replaced by the sample average. However, in general, solving estimating equations may lead to multiple solutions (or the solution may not even exist), which complicates practical implementation. As an alternative, we propose to estimate  $F_m^*$ by maximizing the following modified quasi-likelihood function:  
\begin{align}   
\hat F^{(D_1)}_m=\arg\max Q^{(D_1)}_m(F_m), ~~\textrm{where}~~Q^{(D_1)}_m(F_m)=\frac{1}{|D_1|}\sum_{i\in D_1} \int^{F_mX^{(i)}}_0 \frac{Y_m^{(i)} - b'(\eta)}{b''(\eta)} d\eta. \label{MLE}
\end{align}
$Q^{(D_1)}_m(F_m)$ is a valid likelihood-type function since on the population level $\EE(\nabla Q^{(D_1)}_m(F_m^*))=0$ by (\ref{score_reweighted}), and $\EE(\nabla^2 Q^{(D_1)}_m(F_m^*))$ is negative definite. This implies that, on the sample level, $Q^{(D_1)}_m(F_m)$ is locally strictly concave around $F_m^*$. However, the function $Q^{(D_1)}_m(F_m)$ may not be strictly concave for all $F_m$, which implies the possibility of local maximizers. Thus, following the convention in the statistics literature, it is advisable to maximize $Q^{(D_1)}_m(F_m)$ from multiple initial values in practice. Finally, we note that our function $Q^{(D_1)}_m(F_m)$ is related but different from the standard quasi-likelihood \citep{wedderburn1974quasi} defined as $\frac{1}{n}\sum_{i=1}^n\int_0^{b'(F_mX^{(i)})}\frac{Y_m^{(i)} - \mu}{V(\mu)} d\mu$, where $V(\cdot)$ is the variance function. Since the pseudo-true parameter $F_m^*$ is defined under a misspecified GLM, maximizing the standard quasi-likelihood does not yield a consistent estimator of $F_m^*$.

After obtaining $\hat{F}^{(D_1)}_m$ for all $1 \leq m \leq M$, we can construct the estimated residuals using the data in $D_2$. That is, for $i\in D_2$, we can construct
\begin{align}
    \hat{\epsilon}_m^{(i)} = \frac{Y_m^{(i)} - b'(\hat{F}^{(D_1)}_mX^{(i)})}{b''(\hat{F}^{(D_1)}_mX^{(i)})}. \label{resid}
\end{align}
We can then estimate $\EE(\bar{\epsilon}\bar{\epsilon}^T)$ in (\ref{eq_var_pca}) by
\begin{align}
\hat{\Sigma}^{(D_2)} = \frac{1}{|D_2|}\sum_{i\in D_2} \hat{\epsilon}^{(i)}(\hat{\epsilon}^{(i)})^T, \label{noise_matrix}
\end{align}
where $\hat{\epsilon}^{(i)}=(\hat{\epsilon}_1^{(i)},...,\hat{\epsilon}_M^{(i)})^T\in\RR^M$. The sample splitting procedure guarantees the desired independence between $\hat{F}^{(D_1)}_m$ and the data $(Y_m^{(i)}, X^{(i)})$ in $D_2$, which simplifies the technical analysis of the sample covariance matrix of $\hat \epsilon^{(i)}$. To fully utilize the data, we can switch the role of $D_1$ and $D_2$ to construct the estimators $\hat{F}^{(D_2)}_m$ and $\hat{\Sigma}^{(D_1)}$, and eventually define
\begin{equation}\label{eq_sample_split}
\hat{F}_m=(\hat{F}^{(D_1)}_m+\hat{F}^{(D_2)}_m)/2, ~~\textrm{and}~~\hat{\Sigma}=(\hat{\Sigma}^{(D_1)}+\hat{\Sigma}^{(D_2)})/2. 
\end{equation}
Inspired by (\ref{eq_var_pca}), we apply spectral decomposition on $\hat{\Sigma}$ to get the first $K$ eigenvectors which are arranged as columns in $\hat{V}\in\RR^{M\times K}$ which is then used to construct $\hat{P}_B^{\perp}=I-\hat{V}\hat{V}^T$. We refer to this as the PCA step. In view of (\ref{eq_taylor_weighted_score}), we propose to remove the first-order bias of $F^*$ and estimate $\Theta$ with $\hat{\Theta}=\hat{P}_{B}^{\perp}\hat{F}$, where $F^*:=[F_1^{*T} , ... , F_M^{*T}]^T\in\RR^{M\times K}$ and $\hat{F}:=[\hat{F}_1^T , ... , \hat{F}_M^T]^T\in\RR^{M\times K}$. 

\begin{remark}\label{rem_1}
Since $K$, the number of eigenvectors of $\hat{\Sigma}$ to extract, is unknown in practice, the user needs to specify its value to implement the PCA step. Similar to \cite{ahn2013eigenvalue, lam2012factor}, we consider the following eigenvalue ratio approach. In particular, we estimate $K$ by
\begin{align}
    \hat{K} ~=~ \underset{j \in \{1,2,...,\bar{K}\}}{\arg\max}~\frac{\hat{\lambda}_j}{\hat{\lambda}_{j+1}} \label{selection_K}
\end{align}
where $\hat{\lambda}_{1} \geq \hat{\lambda}_{2}\geq ...$ are the eigenvalues of $\hat{\Sigma}$ and $\bar{K}$ is a pre-specified value. Similar to \cite{lam2012factor}, we set $\bar{K}=\lfloor (n \wedge M)/2 \rfloor$, as the rank of $\hat{\Sigma}$ is no greater than $n \wedge M$, and it is reasonable to look at the first half of the non-zero eigenvalues. The intuition of this approach is that by (\ref{eq_var_pca}) the sample covariance matrix $\hat{\Sigma}$ should have $K$ spiked eigenvalues and therefore the eigenvalue ratio ${\hat{\lambda}_j}/{\hat{\lambda}_{j+1}}$ is expected to reach the maximum at $j=K$. One desired property of this approach is that it does not require any knowledge of unknown population level quantities or additional tuning parameters. The theoretical justification of (\ref{selection_K}) follows the same argument as in \cite{bing2022adaptive}. We defer further technical results to  Section \ref{sec_rem_1} of the Appendix. In simulations, we implement this data-driven choice of $\hat{K}$ in \textsc{data-driven G-hive}, and it is shown to yield reasonable results. 
\end{remark}

\subsection{Approximate Inference}\label{sec_inf}
In this subsection, we consider how to construct an inference procedure for $\Theta$. Recall 
that following the analysis in Section \ref{sec_ident}, we can only argue $\Theta$ is asymptotically identifiable. To study the inferential property for $\Theta$, we have to characterize the asymptotic bias in the identification of $\Theta$. In view of (\ref{eq_taylor_weighted_score}), Theorem \ref{thm_population} below shows that the asymptotic bias satisfies
\begin{equation}\label{eq_inf_bias}
\big|\big|F^*_m-\Theta_m\big|\big|_2 =O \bigg(\frac{1}{\sqrt{p}}\bigg)~~~~~\text{and}~~~~~\big|\big|(P_B^{\perp}F^*)_m-\Theta_m\big|\big|_2 = O\bigg(\frac{1}{p}\bigg),
\end{equation}
where $(P_B^{\perp}F^*)_m$ is the $m$th row of $P_B^{\perp}F^*$. An important consequence of the above result is that the asymptotic bias may dominate the estimation error, making inference on $\Theta$ difficult or even infeasible. 

To explain the details, we focus on the estimator $\hat F$. The same argument applies to $\hat\Theta$ as well. Theorem \ref{thm_hatF} below shows that $\hat F$ is asymptotically linear. That is, under some conditions, for any $1\leq m\leq M$ and $1\leq j\leq p$,
\begin{equation}\label{eq_inf_F}
\sqrt{n}(\hat F_{mj}-F^*_{mj})=\frac{1}{\sqrt{n}}\sum_{i=1}^n \frac{Y^{(i)}_m - b'(F_m^* X^{(i)})}{b''(F_m^* X^{(i)})} e_j^TG_m^{-1}X^{(i)}+o_p(1),
\end{equation}
where $e_j$ is a unit basis vector with the $j$th entry being 1 and 0 otherwise, and $G_m$ is defined in (\ref{eq_G_m}) below. Applying the central limit theorem to the first term on the right hand side and using (\ref{eq_inf_bias}), we have
\begin{equation}\label{eq_inf_F_invalid}
\sqrt{n}(\hat F_{mj}-\Theta_{mj})=\sqrt{n}(\hat F_{mj}-F^*_{mj})+\sqrt{n}(F^*_{mj}-\Theta_{mj})\rightarrow_d N(0,\sigma^2_{mj})+O_p\Big(\sqrt{\frac{n}{p}}\Big),
\end{equation}
for some $\sigma^2_{mj}>0$. When $p=o(n)$, which is the regime considered in this work, the asymptotic bias dominates the stochastic error of the estimator $\hat F_{mj}$ so that the confidence interval based on the limiting distribution of the estimator $\hat F_{mj}$ does not yield the desired coverage probability for $\Theta_{mj}$. Unlike linear regression with hidden variables, the presence of asymptotic bias makes inference in GLMs substantially more challenging.  

To overcome this difficulty, we propose the following approximate inference framework. Specifically, we shift the parameter of interest from $\Theta$ to $F^*$ or $P_B^{\perp}F^*$, where, by (\ref{eq_inf_bias}), $F^*$ can be viewed as the first-order approximation of $\Theta$ and $P_B^{\perp}F^*$ the second-order approximation. As a result, we refer to inference on $F^*$ and $P_B^{\perp}F^*$ as {\it first-order approximate inference} and {\it second-order approximate inference}, respectively. Indeed, first-order approximate inference (i.e., inference on $F^*$) is immediately available from the result in (\ref{eq_inf_F}) or more generally, from Theorem \ref{thm_hatF}. However, it is reasonable to expect that inference on $P_B^{\perp}F^*$  would serve as a more accurate surrogate and provide more information on $\Theta$ compared to the first-order inference method.  Thus, in this work we focus on the following second-order approximate inference method for uncertainty quantification. 

Assume that the parameter of interest is defined as $u^T(P_B^\perp F^*)v$, where $u\in\RR^M$ and $v\in\RR^p$ are known vectors satisfying $||u||_2 = ||v||_2 = 1$. We thus allow for the inference on arbitrary linear combinations of parameters in our approach. Define a weighted covariance matrix as
\begin{equation}\label{eq_G_m}
G_m=\EE\Big(1+\zeta^{(i)}_m(F^*_m)\Big) X^{(i)} X^{(i)T},
\end{equation}
which corresponds to the expected Hessian of the modified quasi-likelihood function, where 
\begin{equation}\label{eq_zeta}
\zeta^{(i)}_m(F_m)=\frac{(Y^{(i)}_m - b'(F_m X^{(i)}))b'''(F_m X^{(i)})}{\{b''(F_m X^{(i)})\}^2}. 
\end{equation}
By Theorem \ref{thm_inference}, under certain conditions, we can show that
$$
\sqrt{n}u^T(\hat\Theta-P_B^\perp F^*)v/(s_n/\sqrt{n}) \rightarrow_d N(0, 1),
$$
where $s_n^2=\sum_{i=1}^n \EE(u^TP_B^\perp h^{(i)})^2$ with $h^{(i)}=(h^{(i)}_1,...,h^{(i)}_M)^T$ and $h^{(i)}_m=\bar \epsilon^{(i)}_mv^T G^{-1}_mX^{(i)}$. In addition, define $\hat h^{(i)}=(\hat h^{(i)}_1,...,\hat h^{(i)}_M)^T$, where $\hat h^{(i)}_m=\hat \epsilon^{(i)}_m v^T \hat G^{-1}_mX^{(i)}$ and
\begin{equation}\label{eq_hat_G_m}
\hat G_m=\frac{1}{n}\sum_{i=1}^n\Big(1+\zeta^{(i)}_m(\hat F_m)\Big) X^{(i)} X^{(i)T}
\end{equation}
is an estimate of $G_m$.  Theorem \ref{thm_inference} further shows that the asymptotic variance $s_n^2/n$ can be consistently estimated by $\hat s_n^2/n$, where $\hat s_n^2=\sum_{i=1}^n (u^T\hat P_B^\perp\hat h^{(i)})^2$. As a result, the $(1-\alpha)\%$ confidence interval for  $u^T(P_B^\perp F^*)v$ is given by $(u^T\hat\Theta v-q_{1-\alpha/2} \hat s_n/\sqrt{n}, u^T\hat\Theta v+q_{1-\alpha/2} \hat s_n/\sqrt{n})$, where $q_{1-\alpha/2}$ is the $(1-\alpha/2)$-quantile of a standard normal distribution. Finally, we note that while the proposed confidence interval yields the desired coverage probability for $u^T(P_B^\perp F^*)v$ rather than $u^T\Theta v$, we expect the proposed approximate inference framework to offer a valuable toolbox to quantify the uncertainty of estimating $\Theta$, and to provide useful information for inferring the magnitude of $\Theta$ in practice. This is confirmed in our simulation studies.

\begin{algorithm}
\caption{\textsc{G-hive}: Parameter Estimation and Approximate Inference}
    \textsc{Input}: i.i.d. observations $(Y^{(i)},X^{(i)}), ~i=1,...,n$, and rank $K$.
    \begin{itemize}
    \item[(1)] Randomly split the data into two folds $D_1$ and $D_2$.
    \item[(2)] Using the data in $D_1$, compute $\hat{F}_m^{(D_1)}$ by solving (\ref{MLE}).
    \item[(3)] Using the data in $D_2$, compute the sample covariance matrix $\hat{\Sigma}^{(D_2)}$ in (\ref{noise_matrix}).
    \item[(4)] Similarly, compute $\hat{F}_m^{(D_2)}$ and $\hat{\Sigma}^{(D_1)}$, and the averaged estimators $\hat{F}_m$ and $\hat{\Sigma}$ in (\ref{eq_sample_split}). 
    \item[(5)] Compute $\hat{P}_{B}^{\perp}=I_M - \hat{V}\hat{V}^T$, where $\hat{V}\in\RR^{M \times K}$ consists of columns corresponding to the first $K$ eigenvectors of $\hat{\Sigma}$.
    \item[(6)] Construct the point estimator $\hat{\Theta}=\hat{P}_{B}^{\perp}\hat{F}$, where $\hat{F}=\big[\hat{F}_1^T,...,\hat{F}_M^T\big]^T $.
    \item[(7)] Construct the $(1-\alpha)\%$ second-order approximate confidence interval, 
    $$(u^T\hat\Theta v-q_{1-\alpha/2} \hat s_n/\sqrt{n}, u^T\hat\Theta v+q_{1-\alpha/2} \hat s_n/\sqrt{n}).$$ 
    \end{itemize}
\label{GHive}
\end{algorithm}

\section{Statistical Guarantees}\label{theory}

We use $\Sigma_X=\EE(XX^T)$ and $\Sigma_Z=\EE(ZZ^T)$ to denote the covariance matrices of $X$ and $Z$, respectively, and throughout the paper we consider the asymptotic setting of $p, M, K \rightarrow \infty$ as $n\rightarrow\infty$.

\begin{assum}{(Identifiability Assumption).}\label{assum_ident}
Assume that $P_{B}\Theta=0$.  
\end{assum}

\begin{assum}{(Tail Assumption).}\label{assum_bound}
Assume that $W$ and $Z$ are sub-Gaussian vectors with bounded sub-Gaussian norm $\sigma_W^2$ and $\sigma_Z^2$, respectively. Given $X$ and $Z$, the error $Y_m-b'(\Theta_m X + B_mZ)$ is sub-exponential with bounded sub-Exponential norm $\sigma^2_{\epsilon, \max}$ for $1\leq m\leq M$, and $\underset{1 \leq i \leq n}{\max}~\underset{1 \leq j \leq p}{\max}~ \big|X^{(i)}_j\big| ~\leq~ C_0$ for some constant $C_0 > 0$. 
%    \item[~~~~(a)] $\underset{1 \leq i \leq n}{\max}~\underset{1 \leq j \leq p}{\max}~ |\XX_{ji}| ~\leq~ C_0$ for some fixed constant $C_0 > 0$
%    \item[~~~~(b)] $\big|\big|\Sigma_X^{-1/2}X\big|\big|_{\psi_2}~\leq~ c_x$ for some fixed constant $c_x > 0$. 
%    \item[~~~~(c)] $\underset{1 \leq i \leq n}{\max}~\underset{1 \leq m \leq M}{\max}~ |\YY_{mi}| ~\leq~ C_0^{**}$ for some fixed constant $C_0^{**}>0$.  
\end{assum}
\begin{assum}{(GLM Assumption).} \label{assum_glm}
There exist some constants $C_1,C_2, C_3>0$, such that $C_1 \leq b''(t) \leq C_2$, and $|b'(t)|$, $|b'''(t)|$, $|b''''(t)|$ are all upper bounded by $C_3$.
\end{assum}
\newpage
\begin{assum}{(Factor Model Assumption).}\label{assum_pervasiveness}
    \item[~~~~(a)] $\kappa_{A,1} \cdot p ~\leq~ \lambda_k(A^TA) ~\leq~ \kappa_{A,2}\cdot p $ for some fixed constants $\kappa_{A,1}, \kappa_{A,2} >0$ and all $1\leq k \leq K$.
    \item[~~~~(b)] $\kappa_{B,1} \cdot M ~\leq~ \lambda_k(B^TB) ~\leq~ \kappa_{B,2}\cdot M $ for some fixed constants $\kappa_{B,1}, \kappa_{B,2} >0$ and all $1\leq k \leq K$. 
    \item[~~~~(c)] $\kappa_{Z,1} ~\leq~ \lambda_k(\Sigma_Z) ~\leq~ \kappa_{Z,2} $ for some fixed constants $\kappa_{Z,1}, \kappa_{Z,2} > 0$ and all $1\leq k \leq K$.
\item[~~~~(d)]$||B_m||_2 \leq~ C_4$ for all $1 \leq m \leq M$ and for some fixed constant $C_4 > 0$.    
\end{assum}

%\begin{assum}{($L_2$ Projection Assumption).} $||B_m(Z-\tilde{Z})||_{\psi_2} \leq c/\sqrt{p}$ for some fixed constant $c>0$, where $\tilde{Z}:=\Sigma_ZA^T(A\Sigma_Z A^T + I_p)^{-1}X$. \label{assump_L2}
%\end{assum}

Assumption \ref{assum_ident} ensures the identifiability of $\Theta$ as explained in Section \ref{sec_ident}. In the related literature, there exist alternative identifiability conditions. We defer the detailed discussions to \cite{lee2017improved}, \cite{bing2022adaptive}, \cite{wang2017confounder} and \cite{bai2003inferential}. Assumption \ref{assum_bound} characterizes the tail behavior of the random vectors $W$ and $Z$ and the response variable $Y$. The sub-exponential condition for $Y_m-b'(\Theta_m X + B_mZ)$ holds for most GLMs such as linear, logistic and Poisson regression.  To simplify the proof, we also assume the elements in $X$ are upper bounded by a fixed constant, which can be relaxed by allowing $C_0$ to scale with $n$ and $p$ (e.g., $C_0\asymp \sqrt{\log(np)}$ for sub-Gaussian $X^{(i)}_j$). Assumption \ref{assum_glm} on the higher order derivatives of $b(t)$ is standard for analyzing GLMs.  Finally,  Assumptions \ref{assum_pervasiveness}(a) and \ref{assum_pervasiveness}(b) are known as the pervasiveness assumption in the factor model literature \citep{fan2013large,fan2008high,chang2015high}. A concrete example of when it is satisfied is discussed in Section \ref{sec_per} of the Appendix. Assumptions \ref{assum_pervasiveness}(c) and \ref{assum_pervasiveness}(d) are also mild conditions for factor models.

We first present a theorem that characterizes the approximation bias of $F_m^*$ defined in (\ref{score_reweighted}). 
\begin{theorem}\label{thm_population}
    Under Assumptions \ref{assum_ident}- \ref{assum_pervasiveness}, for $1 \leq m \leq M$, there exists $F_m^*$ defined in (\ref{score_reweighted}) such that 
    $$
F_m^*-\Theta_m=B_m \Sigma_Z A^T\Sigma_X^{-1}+Rem_m',
    $$
where $\max_{1\leq m\leq M}\|Rem_m'\|_2=O(1/p)$. In addition, the following hold:    
    \begin{align}
    \max_{1\leq m\leq M}\EE\Big[\big(\Theta_mX + B_m Z - F_m^* X\big)^4\Big] =O\bigg( \frac{1}{p^2}\bigg)\label{eq_thm_fourthmom}
    \end{align}
    and
    \begin{align}\label{eq_thm_first_bias}
       \max_{1\leq m\leq M} ||F_m^*-\Theta_m||_2 =O \bigg(\frac{1}{\sqrt{p}}\bigg).
    \end{align}
    This further implies that
    \begin{align}
    \frac{1}{\sqrt{M}}\big|\big|F^*-\Theta\big|\big|_F =O \bigg(\frac{1}{\sqrt{p}}\bigg)~~~~~\text{\textit{and}}~~~~~\frac{1}{\sqrt{M}}\big|\big|P_B^{\perp}F^*-\Theta\big|\big|_F = O\bigg(\frac{1}{p}\bigg). \label{eq.difference}\end{align}\label{thm_expectation}
\end{theorem}
This theorem provides a rigorous justification of equation (\ref{eq_taylor_weighted_score}) in Section \ref{sec_ident}, where the remainder term $Rem_m'$ in $L_2$ norm is of order $O(1/p)$ uniformly over $m$, and the first-order bias is of order $O(1/\sqrt{p})$ by (\ref{eq_thm_first_bias}). More importantly, it shows the theoretical advantage of the PCA step in our estimation procedure as it reduces the inherent bias that occurs from model misspecification due to the hidden variables. As seen in (\ref{eq.difference}), using $F^*$ as a proxy of $\Theta$ inevitably incurs the approximation bias (in terms of the $L_2$ error per response) with rate $O (1/\sqrt{p})$. However, by projecting $F^*$ to the orthogonal space of $B$ via the PCA step, we can reduce the approximation bias to have a faster rate of $O(1/p)$.

The next theorem provides the bound for the stochastic error $\big|\big|\hat{F}_m - F_m^*\big|\big|_2$ and the asymptotic linear approximation of $\hat{F}_m-F_m^*$ uniformly over $1\leq m\leq M$. 

\begin{theorem}\label{thm_hatF}
Under Assumptions \ref{assum_ident}- \ref{assum_pervasiveness}, $p\sqrt{\frac{\log (p\vee M)}{n}} \log(M\vee n)=o(1)$ and $\{\log (M\vee p)\}^3=O(n)$, there exists a local maximizer $\hat{F}^{(D_j)}_m$ of $Q^{(D_j)}_m(F_m)$ for $1\leq m\leq M$ and $j\in\{1,2\}$, such that
    \begin{align}\label{eq_thm_hatF_rate}
    \max_{1\leq m\leq M}\big|\big|\hat{F}_m-F_m^*\big|\big|_2 =O_p\Bigg(\sqrt{\frac{p\log (M\vee p)}{n}}\Bigg).
    \end{align}
In addition,  for any $v\in\RR^p$ with $\|v\|_2=1$, we have
    \begin{align}\label{eq_thm_hatF_linear}
    (\hat{F}_m-F_m^*)v =\frac{1}{n}\sum_{i=1}^n \frac{Y^{(i)}_m - b'(F_m^* X^{(i)})}{b''(F_m^* X^{(i)})} v^TG_m^{-1}X^{(i)}+Rem''_m,
    \end{align}
where $G_m$ is defined in (\ref{eq_G_m}) and $\max_{1\leq m\leq M}|Rem''_m|=O_p\Big(p^{3/2}\cdot\frac{\log (p\vee M)\log (n\vee M)}{n}\Big)$.    
\end{theorem}

Recall that the modified quasi-likelihood $Q^{(D_j)}_m(F_m)$ is non-concave and may have multiple local solutions. This theorem only applies to some local maximizer of $Q^{(D_j)}_m(F_m)$. While  the rate of convergence obtained in (\ref{eq_thm_hatF_rate}) agrees with the existing literature on M-estimation with increasing dimension \citep{portnoy1984asymptotic}, a unique challenge that had to be overcome is the fact that the covariance matrix $\Sigma_X$ has spiked eigenvalues, which is implied by the hidden variable model (\ref{factor_model}) and the pervasiveness assumption in Assumption \ref{assum_pervasiveness}. This required a more delicate analysis to control the perturbation of the Hessian matrix $\nabla^2 Q^{(D_j)}_m(F_m)$ around $F_m^*$.  

%which is technically more complicated than analyzing the MLE under GLMs due to the non-concavity of $Q_m(F_m)$. 

Combining the approximation error in Theorem \ref{thm_expectation} and the stochastic error in Theorem \ref{thm_hatF}, we obtain
    \begin{align}\label{eq_rate_first_order}
\frac{1}{\sqrt{M}}\big|\big|\hat F-\Theta\big|\big|_F =O \bigg(\frac{1}{\sqrt{p}}+\sqrt{\frac{p\log (M\vee p)}{n}}\bigg)=O \bigg(\frac{1}{\sqrt{p}}\bigg),
    \end{align}
where in the last step we notice that the error bound is dominated by the approximation error under the condition $p\sqrt{\frac{\log (p\vee M)}{n}} \log(M\vee n)=o(1)$ in Theorem \ref{thm_hatF}. 

Finally, the asymptotic linear expansion of $(\hat{F}_m-F_m^*)v$ in (\ref{eq_thm_hatF_linear}) shows that (\ref{eq_inf_F}) holds under the condition $p^{3/2}\cdot\frac{\log (p\vee M)\log (n\vee M)}{\sqrt{n}}=o(1)$, which validates first-order approximate inference. Since inference on $F^*$ is not the main focus of this work, we do not pursue further results along this line.

The next theorem provides the rate for the estimation error of $\hat{P}_{B}^{\perp}$ in the PCA step.
\begin{theorem}\label{thm_projection}
    Under the assumptions in Theorem \ref{thm_hatF}, if we further assume $\{\log (M\vee p)\}^5=O(n)$ and $K<n$, then we have
    \begin{align}\label{eq_thm_projection_1}
\big|\big|\hat{P}_{B}^{\perp} - P_B^{\perp}\big|\big|_F ~=~ O_p\Bigg(\frac{p}{\sqrt{M}}+\frac{1}{\sqrt{p}}+p\sqrt{\frac{\log (p\vee M)}{n}}+\sqrt{\frac{K}{n}}\Bigg).
    \end{align}
\end{theorem}

The estimation error of the projection matrix consists of four terms. The first two terms are the asymptotic bias corresponding to the remainder term in the expansion (\ref{eq_var_pca}), and the last two terms stem from the stochastic error of $\hat{P}_{B}^{\perp}$. Under the condition $p\sqrt{\frac{\log (p\vee M)}{n}} \log(M\vee n)=o(1)$ in Theorem \ref{thm_hatF}, the stochastic error is not necessarily dominated by the asymptotic bias in (\ref{eq_thm_projection_1}). So we need to keep all four terms in (\ref{eq_thm_projection_1}). Under additional conditions $p=o(\sqrt{M})$ and $K=o(n)$, the estimator $\hat{P}_{B}^{\perp}$ is consistent in Frobenius norm.

The proof of Theorem \ref{thm_projection} relies on the Davis-Kahan Theorem (Lemma \ref{DavisKahan}) and the bound for $\|\hat\Sigma-\Sigma\|_F$, where $\Sigma=B(\Sigma_Z^{-1}+A^TA)^{-1}B^T$. Based on a more refined expansion compared to (\ref{eq_var_pca}), we can decompose the error $\hat\Sigma_{mm'}-\Sigma_{mm'}$ into pairwise interactions of 6 error terms (21 terms in total), where we further need to distinguish the analysis for the diagonal term $m=m'$ and the off-diagonal term $m\neq m'$. The resulting proof is much more technical than that for linear regression. In particular, we apply sample splitting to decorrelate the error terms in the expansion of $\hat{P}_{B}^{\perp}$, leading to a faster rate of convergence for some of the error terms.

The next theorem provides the rate for the estimation error of our final estimator $\hat{\Theta}$.
\begin{theorem}\label{thm_final}
Under the same assumptions as in Theorem \ref{thm_projection}, we have
    \begin{align*}
        \frac{1}{\sqrt{M}}\big|\big|\hat{\Theta} - \Theta\big|\big|_F&= O_p\Big(Err_1+Err_2+Err_3\Big),
    \end{align*}
where 
\begin{align*}
Err_1&=\Bigg(\frac{p}{\sqrt{M}}+\frac{1}{\sqrt{p}}+p\sqrt{\frac{\log (p\vee M)}{n}}+\sqrt{\frac{K}{n}}\Bigg) \frac{||F^*||_{\mathrm{op}}}{\sqrt{M}},\\
Err_2&= \bigg(1+\frac{p}{\sqrt{M}}\bigg)\sqrt{\frac{p\log (p\vee M)}{n}},~~\textrm{and}~~Err_3=\frac{1}{p}.
\end{align*}
\end{theorem}
This theorem shows that the per-response $L_2$ estimation error of $\hat\Theta$ is bounded by three terms $Err_1, Err_2$ and $Err_3$, where $Err_1$ is inherent from the estimation error of $\hat P_B$ in Theorem \ref{thm_projection}, $Err_2$ comes from the estimation error of $\hat F$  in Theorem \ref{thm_hatF}, and $Err_3$ corresponds to the approximation error of $P_B^{\perp}F^*$ in Theorem \ref{thm_population}. When $p=o(\sqrt{M}),~K=o(n)$, $p\sqrt{\frac{\log (p\vee M)}{n}} \log(M\vee n)=o(1)$, and $||F^*||_{\text{op}}/\sqrt{M}=O(1)$, as $p,M,n, K\rightarrow \infty$, the estimator $\hat\Theta$ is consistent in the sense that $\|\hat{\Theta} - \Theta\|_F/\sqrt{M}=o(1)$. 

\begin{remark}\label{rem_rate}
To have a more refined comparison with the estimation error of $\hat F$ in (\ref{eq_rate_first_order}), we further assume $||\Theta||_{\text{op}}=O(1)$. Under this assumption, since by Theorem \ref{thm_population} we have $||F^*||_{\text{op}}\leq ||\Theta||_{\text{op}}+\|F^*-\Theta\|_F= O\big(1+\sqrt{M/p}\big)$, the rate of the estimator $\hat{\Theta}$ in Theorem \ref{thm_final} reduces to
\begin{align}\label{eq_final_rate_1}
\frac{1}{\sqrt{M}}\big|\big|\hat{\Theta} - \Theta\big|\big|_F= O_p\bigg(\frac{1}{p}+\sqrt{\frac{p}{M}}+\sqrt{\frac{p\log (p\vee M)}{n}}+\sqrt{\frac{K}{n p }}+\sqrt{\frac{p^3\log (p\vee M)}{nM}}\bigg).
\end{align}
Assuming $M\asymp p^{\alpha}$ and $K\asymp n^{\beta}$ for some positive constants $\alpha\geq 1$ and $\beta\leq 1$, the rate can be simplified to
\begin{align}\label{eq_final_rate_2}
\frac{1}{\sqrt{M}}\big|\big|\hat{\Theta} - \Theta\big|\big|_F= O_p\Bigg(\frac{1}{p^{1\wedge  \frac{\alpha-1}{2}}}+\frac{1}{p^{1/2}n^{\frac{1-\beta}{2}}}+\sqrt{\frac{p^{1\vee (3-\alpha)}\log p}{n}}\Bigg).
\end{align}
Provided $\alpha>2$ and $\beta<1$, the rate of our estimator $\hat\Theta$ in  (\ref{eq_final_rate_2}) is faster than the rate of $\hat F$ in (\ref{eq_rate_first_order}), which justifies the theoretical benefit of our proposed method over the naive MLE approach that ignores hidden variables.  
\end{remark}

Recall the notation in Section \ref{sec_inf}: $G_m$ and $\hat G_m$ are defined in (\ref{eq_G_m}) and (\ref{eq_hat_G_m}), $s_n^2=\sum_{i=1}^n \EE(u^TP_B^\perp h^{(i)})^2$ with $h^{(i)}=(h^{(i)}_1,...,h^{(i)}_M)^T$ and $h^{(i)}_m=\bar \epsilon^{(i)}_mv^T G^{-1}_mX^{(i)}$, and $\hat s_n^2=\sum_{i=1}^n (u^T\hat P_B^\perp\hat h^{(i)})^2$ with $\hat h^{(i)}=(\hat h^{(i)}_1,...,\hat h^{(i)}_M)^T$ and $\hat h^{(i)}_m=\hat \epsilon^{(i)}_m v^T \hat G^{-1}_mX^{(i)}$. Finally, we establish the limiting distribution of the estimator $u^T\hat\Theta v$ in the following theorem. 

\begin{theorem}\label{thm_inference}
Under the same assumptions as in Theorem \ref{thm_projection}, if we further assume $K=o(n)$, $p=o(\sqrt{M})$ and $\EE(u^TP_B^\perp h^{(i)})^2\geq C$ for some constant $C>0$, then for any $u\in\RR^M$ and $v\in\RR^p$ with $\|u\|_2=\|v\|_2=1$,  
\begin{align}\label{thm_inference_1}
\sup_t\Bigg|\PP\Big(\frac{u^T(\hat\Theta-P_B^\perp F^*)v}{s_n/\sqrt{n}}\leq t\Big)-\Phi(t)\Bigg|&~\leq~ C'(\delta_1+\delta_2+\delta_3) 
\end{align}
for some constant $C'>0$, where $\Phi(\cdot)$ is the c.d.f of a standard normal distribution, and the three error terms in (\ref{thm_inference_1}) are given by 
$$
\delta_1=\frac{R_n^3\big\{(\log M)^{9/2}\vee K^{3/2}\big\}}{\sqrt{n}}+R_n\sqrt{K} p^{3/2}\frac{\log (p\vee M)\log (n\vee M)}{\sqrt{n}},
$$
$$
\delta_2=R_n \|F^* v\|_2 \bigg(\sqrt{\frac{n}{pM}}+p\sqrt{\frac{\log (p\vee M)}{M}}+\frac{p\sqrt{n}}{M}\bigg),
$$
and
$$
\delta_3=R_n \bigg(\frac{1}{\sqrt{p}}+\frac{p}{\sqrt{M}}+\frac{p^{5/2}\log (p\vee M)\log (n\vee M)}{\sqrt{Mn}}\bigg),
$$
with $\|u\|_1\leq R_n$ for some $R_n>0$. Finally, the asymptotic variance $s_n^2/n$ can be consistently estimated by $\hat s_n^2/n$, i.e.,
\begin{align}\label{thm_inference_2}
\Bigg|\frac{\hat{s}_n^2}{n}-\frac{s_n^2}{n}\Bigg|=O_p\Bigg(R_n^2\Big\{\log (M\vee n)\Big\}^2\sqrt{K} \bigg\{\frac{p}{\sqrt{M}}+\frac{1}{\sqrt{p}}+p\sqrt{\frac{K\log (p\vee M)}{n}}\log \big(M\vee n\big)\bigg\}\Bigg).
\end{align}

\end{theorem}

We characterize the accuracy of the Gaussian approximation of $u^T\hat\Theta v$ via the Berry–Esseen bound (\ref{thm_inference_1}), where $\delta_1$ corresponds to the Gaussian approximation error of $\hat F$ which comes from the asymptotic linear expansion (\ref{eq_thm_hatF_linear}) in Theorem \ref{thm_hatF}, $\delta_2$ corresponds to the estimation error of $\hat P_B$, and $\delta_3$ corresponds to the product error of $\hat F$ and $\hat P_B$. To simplify the Berry–Esseen bound (\ref{thm_inference_1}), we further assume that $R_n=O(1)$, $K=O(1)$ and $\|F^* v\|_2=O(1)$. Then  the Berry–Esseen bound (\ref{thm_inference_1}) reduces to
\begin{align}\label{thm_inference_3}
\frac{(\log M)^{9/2}}{\sqrt{n}}+p^{3/2}\frac{\log (p\vee M)\log (n\vee M)}{\sqrt{n}}+\sqrt{\frac{n}{pM}}+p\sqrt{\frac{\log (p\vee M)}{M}}+\frac{p\sqrt{n}}{M}+\frac{1}{\sqrt{p}}.
\end{align}
Assuming $M\asymp n^{r_M}$ and $p\asymp n^{r_p}$ for some positive constants $r_M$ and $r_p$, the above bound goes to 0 when $r_p< 1/3$ and $r_M> (1-r_p)\vee (\frac{1}{2}+r_p)$ hold.  The condition $r_p<1/3$ on the number of covariates $p$ is comparable to the requirement on the dimensionality for M-estimation or the intrinsic dimensionality for sparse GLMs. The condition $r_M> (1-r_p)\vee (\frac{1}{2}+r_p)$ is unique in our hidden variable model, which implies that the number of responses $M$ needs to be large enough such that we can borrow information from different responses to better estimate $P_B$. It can also be viewed as a kind of blessing of dimensionality, as the error (\ref{thm_inference_3}) generally decreases with larger $M$.

Finally, as long as the asymptotic variance $s_n^2/n$ can be consistently estimated by $\hat s_n^2/n$ shown by (\ref{thm_inference_2}), the second-order approximate confidence interval proposed in Section \ref{sec_inf} yields the desired coverage probability for $u^T(P_B^\perp F^*)v$.

\section{Simulation Results}\label{sec_sim}
Here we present our simulation results, which can be divided into three categories: the approximation bias of $F^*$ and $P_B^{\perp}F^*$, the estimation error of $\hat{\Theta}$, and statistical inference of $\hat{\Theta}$. We first present the data generating mechanism, and then discuss each result in the above categories. 
%The first setting alters the magnitude of the effect of the hidden variables and showcases the robustness of our method compared to the baseline naive maximum likelihood estimation method. The second setting fixes $p=5$ and shows consistency of our method as $n$ increases. The third setting fixes $n=100$ and shows consistency of our method as we change the relevant theoretical rate that depends on $p$. The fourth setting shows consistency even when $M$ is large.
\subsection{The Data Generating Mechanism}
To satisfy Assumption \ref{assum_pervasiveness}, we first generate $pK$ i.i.d. $N(0,1)$ random variables to construct the matrix $A\in\RR^{p\times K}$ and normalize each of the $p$ rows to have an $L_2$ norm of 1. The usage of Gaussian random variables to satisfy Assumption \ref{assum_pervasiveness} is formally justified in Section \ref{sec_per} in the Appendix. We then generate $n$ i.i.d  random vectors $Z\sim N(\mathbf{0}_K,\Sigma_Z)$, where $\Sigma_Z$ is a circulant matrix with $1$'s on the diagonal and a decay rate of $-0.5$. Each component of $W$ is  drawn independently from $N(0,1)$ and we compute $X = AZ + W$ to generate the observed covariate matrix $X \in \RR^{p \times n}$. Similarly, each component of the hidden coefficient matrix $B\in\RR^{M\times K}$ is i.i.d. $N(0,1)$ and we normalize each of the $M$ rows to have an $L_2$ norm of 1. Then we multiply $B$ by a positive scalar $\eta \in \RR$, where $\eta$ is a parameter we can tweak to alter the influence of the hidden variables. A larger $\eta$ value corresponds to larger confounding from the hidden variables $Z$. Each element of the coefficient matrix $\Theta\in\RR^{M \times p}$ is i.i.d. $ N(0,1)$ and we normalize each row to have an $L_2$ norm of 1. We project $\Theta$ onto the orthogonal column space of $B$ to satisfy Assumption \ref{assum_ident} and get the final value for $\Theta$. Lastly, we generate  $n$ i.i.d. copies of $Y_m\in\{0,1\}$ from a Bernoulli distribution with $\PP(Y_m=1|X,Z)=\exp(\Theta_mX + B_mZ)/[1+\exp(\Theta_mX + B_mZ)]$ for each $1 \leq m \leq M$. Throughout the simulations, we set $K=3$. 

\subsection{The Methods}
We consider three variants of our \textsc{G-hive} algorithm. \textsc{data driven G-hive} is the proposed method that uses a data driven estimate of $K$ mentioned in Remark \ref{rem_1}. The other two are oracle type estimators, \textsc{Oracle(K) G-hive}, corresponding to the algorithm with the true value of $K$ given and \textsc{Oracle(P) G-hive}, the algorithm with the true projection matrix $P_B$ given. These two oracle type estimators illustrate the impact of estimating $K$ with $\hat{K}$ and the impact of estimating $P_B^{\perp}$ with $\hat{P}_B^{\perp}$ in our method. The baseline method we compared against was the naive maximum likelihood estimator that ignores the hidden variables. This is denoted as \textsc{naive mle}  and was implemented with the \texttt{glm} function in \texttt{R}.

\subsection{Evaluating the Approximation Bias}
As seen from Theorem \ref{thm_population}, the first-order approximation bias $\big|\big|F^*-\Theta\big|\big|_F/\sqrt{M}$ and the projected second-order approximation bias $\big|\big|P_B^{\perp}F^*-\Theta\big|\big|_F/\sqrt{M}$ decays with $p$ with order $O(1/\sqrt{p})$ and $O(1/p)$, respectively. This characterizes the inherent and unavoidable gap between the \say{true} parameter in the misspecified GLM and the true parameter in the correctly specified GLM. We set $M=3$, $\eta=10$, and the results were averaged over $r=20$ repetitions. As obtaining the \say{true} $F^{*}$ amounts to finding the solution to the estimating equation (\ref{score_reweighted}) for each $1 \leq m \leq M$ which does not have a closed form solution, we instead compute the solution to the modified quasi-likelihood function in (\ref{MLE}) with an extremely large $n=2\times10^5$ via the Monte Carlo approach. The results are shown in the left graph in Figure \ref{fig_1}. It is apparent that as $p$ increases, both forms of the approximation bias indeed decay with the projected bias being much smaller than the non-projected bias. Surprisingly, the projected approximation bias is very close to 0, even if $p$ is as small as 3. This confirms the theory in Theorem \ref{thm_population} and also validates the projection step in our \textsc{G-hive} method.
\begin{figure}[h!]
\centering
\includegraphics[width=1\linewidth]{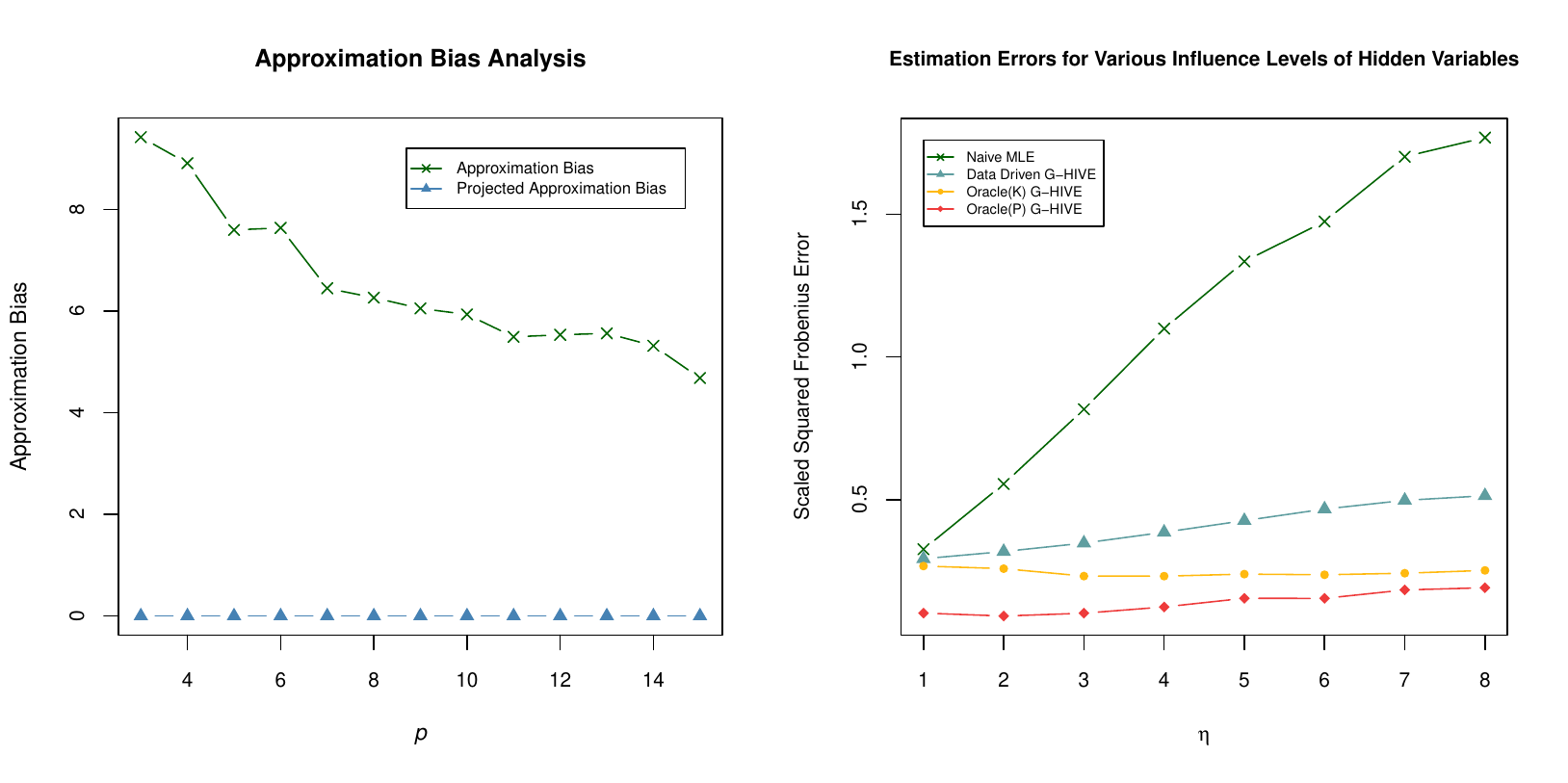}
\caption{The (left) graph shows the approximation bias and the projected approximation bias which correspond to $||F^*-\Theta||_F/\sqrt{M}$ and $||P_B^{\perp}F^*-\Theta||_F/\sqrt{M}$, respectively, in the $M=3,~K=3,~\eta=10$ setting with $p \in \{3,...,15\}$ and the number of repetitions being $r=20$. The (right) graph shows the estimation error $||\widehat{\Theta}-\Theta||_F^2/\sqrt{pM}$ where we vary $\eta$ to take values in $\{1,2,...,8\}$ and have the setting of $p=M=4,~K=3,~ n=100$ averaged over $r=500$ repetitions.}\label{fig_1}
\end{figure}

\subsection{Evaluating the Estimation Error of $\hat\Theta$}
Next we evaluate the estimation error of $\hat{\Theta}$ in three scenarios: (1) when we vary the level of influence of the hidden variables through the magnitude of $\eta$, (2) when we increase the sample size $n$, and (3) when we increase the dimension of the response $M$. For (1) we vary $\eta$ to take values in $\{1,2,...,8\}$ and have the setting of $p=M=4,~K=3,~ n=100$ averaged over $r=500$ repetitions. The results are shown in the right graph in Figure \ref{fig_1}. Recall that the larger the $\eta$ value, the larger the effect of the hidden variables and thus, a more challenging simulation setting. It is not surprising that the \textsc{naive mle} method performs gradually worse, as the magnitude of $\eta$ increases, the MLE is obtained with respect to a model that is becoming more and more misspecified. In contrast, the three \textsc{G-hive} based methods are more robust to the effect of $\eta$, which verifies that the signature projection step in \textsc{G-hive} mitigates the effects of misalignment between the misspecified model and the true model. Lastly, it is reasonable to see an increase in performance as we supply more model information to the estimators such as the true $K$ value or the true projection matrix $P_B$.

For (2), to verify the consistency of our estimator, we gradually increase $n$ to take values from $100$ to $400$ in increments of $50$ and average over $r=200$ repetitions with the model setting being identical otherwise ($M=p=4,~K=3,~\eta=4$). The results are shown in the left graph of Figure \ref{fig_2}, and it is apparent that all four methods show an improvement in performance as $n$ increases.

For (3), to explore performance in high-dimensional response settings, we vary $M$ to take values in $\{4,8,12,16,20\}$ and have the setting of $p=4,~K=3,~ n=200$ averaged over $r=100$ repetitions. The results are shown in the right graph of Figure \ref{fig_2} and they indicate that as $M$ grows, the estimation error slightly increases for all four methods. This coincides with the estimation error of our estimator discussed in Remark \ref{rem_rate} in that the term $\sqrt{p\log (p \vee M)/n}$ in (\ref{eq_final_rate_1}) scales with $\sqrt{\log M}$.  The three \textsc{G-hive} estimators uniformly outperforming the \textsc{naive mle} method shows that \textsc{G-hive} is viable and competitive for a wide range of response dimension values.

\begin{figure}[h!]
\centering
\includegraphics[width=1\linewidth]{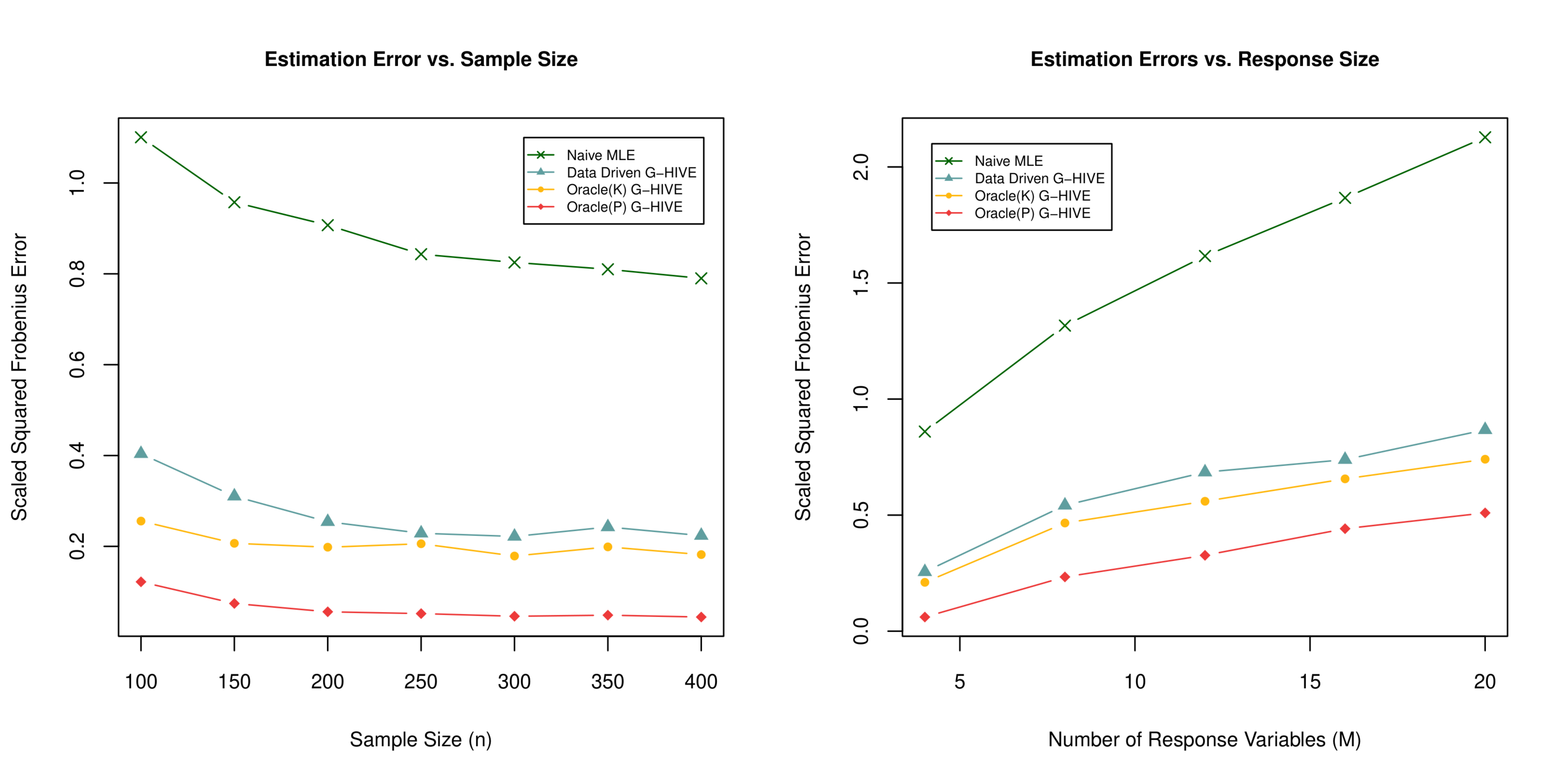}
\caption{The (left) graph shows the estimation error $||\widehat{\Theta}-\Theta||_F^2/\sqrt{pM}$ when we vary $n$ to take values from $100$ to $400$ in increments of $50$ and have the setting of $p=M=4,~K=3,~ \eta=4$ averaged over $r=200$ repetitions. The (right) graph shows the estimation error $||\widehat{\Theta}-\Theta||_F^2/\sqrt{pM}$ when we vary $M$ to take values in $\{4,8,12,16,20\}$ and have the setting of $p=4,~K=3,~ n=200$ averaged over $r=100$ repetitions.}\label{fig_2}
\end{figure}

\subsection{Evaluating the asymptotic normality of $\hat\Theta$}
Lastly we evaluate the asymptotic normality of $\hat{\Theta}$ with our proposed \textsc{data driven G-hive}. We are in the similar setting of $M=p=4,~K=3,~\eta=4$ and have results for $n=40$ and $n=70$ averaged over $r=100$ repetitions with $\alpha=0.05$. While Theorem \ref{thm_inference} only provides coverage guarantees for $ P^{\perp}_B F^{*}$, we also include the coverage probability pertaining to the true $\Theta$. Additionally, while Theorem \ref{thm_inference} provides inference guarantees for $u^T P^{\perp}_B F^{*} v$ for any real vectors $u \in \RR^M,~ v \in \RR^p$, for simplicity, we fix $u=v=[1,0,0,0]^T$ corresponding to \((P^{\perp}_B{F^*})_{11}\). We include in Table \ref{table_1} the estimated standard error $\hat{s}_n/\sqrt{n}$, the confidence interval lengths, and the coverage probabilities for both our \textsc{data driven G-hive} method and the \textsc{naive mle} method implemented with the \texttt{glm} function in \texttt{R}. For \textsc{naive mle}, the constructed confidence intervals are overly narrow which is to be expected for Wald confidence intervals for misspecified models. Coverage probabilities worsen as the sample size increases, which is to be expected as model misspecification bias is left unaddressed. Thus, ignoring the hidden variables when fitting the GLM yields misleading inference results. On the other hand, for \textsc{data driven G-hive}, the coverage probabilities closely match the $1-\alpha$ level, showing the validity of the constructed confidence interval. It is interesting to see that the coverage probability for the true $\Theta$ is close to $1-\alpha$ as well. This implies that while the approximation bias dominates the stochastic error of the estimator and prevents us from having standard inference results for $\Theta$ as shown in (\ref{eq_inf_F_invalid}), in practice, for large enough $n$, the inference results for $P^{\perp}_B F^{*}$ can be used as a proxy for $\Theta$. 

\begin{table}[h!]
\centering
\begin{tabular}{lccccc}
\toprule
& ~~~~\(n\)~~~ & ~CI Length~ & ~Std. Err.~ & Coverage for \((P^{\perp}_B F^*)_{11}\) & Coverage for \(\Theta_{11}\) \\
\midrule
\multirow{2}{*}{\textsc{data driven g-hive}} 
& 40 & 7.77 & 1.98 & 0.98 & 0.98 \\
& 70 & 6.40 & 1.63 & 0.99 & 0.99 \\
\midrule
\multirow{2}{*}{\textsc{naive mle}} 
& 40 & 1.75 & 0.45 & 0.77 & 0.75 \\
& 70 & 1.11 & 0.28 & 0.62 & 0.61 \\
\bottomrule
\end{tabular}
\caption{Inference results showing the length of the confidence interval, standard error of the estimates, and the coverage probabilities for \(u^\top P^{\perp}_B F^{*} v\) and \(u^\top \Theta v\) at level \( \alpha = 0.05 \) for \(u^\top = v^\top = [1,0,0,0]\), averaged over \(r=100\) repetitions.}
\label{table_1}
\end{table}

%\begin{remark}
%    It should be mentioned that to be perfectly aligned with the theoretical analysis, instead of using the naive MLE approach in the EE step, one should use some iterative algorithm (Newton-Raphson, etc.) to get an approximate solution to the estimating equations that address the reweighting issue. However, we found that using the naive MLE step gave uniformly better results compared to when using Newton-Raphson. This, combined with the fact that implementing a naive MLE approach is easier in practice (by implementing \texttt{glm} in \texttt{R}) leads us to recommend using the naive MLE approach in the first step even though there may be slight gap with the theory.
%\end{remark}

%\begin{remark}
%    Using a large $p$ causes separation issues as the chances that a linear combination of a subset of the $p$ explanatory variables perfectly separates the response variable grows very quickly as $p$ increases. Thus, with a fixed $n$ value of 100, under our data generation mechanism, it was realistically feasible to conduct simulations with $p$ taking values up to 20.
%\end{remark}

All in all, these simulation results highlight the soundness of the theory and demonstrate the feasibility and advantages of our \textsc{G-hive} method.

\section{Real Data Analysis}
%\subsection{Lung cancer dataset}
We apply our \textsc{data driven G-hive} procedure to a dataset from \cite{chicco2019computational} regarding mesothelioma, a type of lung cancer. Specifically, this dataset consists of real electronic health records on 324 patients in Turkey of which 96 are diagnosed with mesothelioma and 228 are not. The dataset includes 33 explanatory variables including age, platelet count, white blood cell count, etc., but in order to facilitate the analysis, we included the explanatory variables that were shown to be meaningful in terms of reducing the mean square error (MSE) in \cite{chicco2019computational}. %This left us with 23 variables, but
We removed categorical variables and variables that were highly correlated (a correlation coefficient $\geq 0.91$). We further removed 2 explanatory variables, \say{asbestos exposure} and \say{duration of asbestos exposure} and considered them to be hidden confounders and treated these variables as our unobserved $Z$ variables. In the end, we ran the data analysis on 9 continuous covariates. 

It is known in the medical and biology literature that asbestos, a term applied to mineral species that occur in fibrous forms, can cause chronic inflammation \citep{board2006asbestos}. It is also well known that an inflammatory response causes changes in the components of our blood such as white blood cell count, etc. Thus, it is straightforward to see that asbestos related explanatory variables closely affect other explanatory variables in the dataset (white blood cell count (WCC), etc.). Also, according to \cite{chicco2019computational}, long exposure to asbestos makes mesothelioma very likely. Thus, asbestos related variables affect the response variable \say{diagnosis of mesothelioma} as well. Hence, if we remove asbestos related variables from the dataset, we can consider them hidden confounders that affect both the observed covariates and the response variables and as a result align the real dataset with the model setup we have for our method, \textsc{G-hive}. 

To construct a multivariate response data structure, we include three symptom related response variables from the dataset (\say{chest ache}, \say{dyspnoea}, \say{patient's ability to perform normal tasks}) along with the main response of interest, \say{diagnosis of mesothelioma}. Thus, we have $M=4,~p=9,~n=324$ for our setting. All of the explanatory variables were standardized to have 0 mean and unit variance prior to the data analysis. The resulting $\widehat{\Theta}_{1\cdot}$, i.e. the coefficients pertaining to the diagnosis of mesothelioma are shown below:

\begin{table}[ht]
  \centering
  \caption{The coefficient values relating the explanatory variables to the main response variable (diagnosis of mesothelioma) obtained with \textsc{naive mle} and \textsc{G-hive} in the lung cancer dataset.}
  \label{table_tensor}
  \resizebox{\textwidth}{!}{%
    \begin{tabular}{@{}lrrrrrrrrrr@{}}
      \toprule
      Method & Lung Side & WCC & Platelets & Sedim. & Albumin & Glucose & PLD & Pleural Prot. & Pleural Thick.\\
      \midrule
      \textsc{naive mle} & 0.2057  & -0.1091 & -0.2896 & 0.0895 & 0.1009 & 0.0613 & -0.0496 & -0.1503 & 0.0861 \\
      \textsc{G-hive} & 0.1941 & -0.1116 & -0.4111 & 0.1759 & 0.1148 & 0.0706 & 0.0889 & -0.1395 & 0.1496 \\
      \bottomrule
    \end{tabular}%
  }
\end{table}

As it is impossible to know the ground truth, we rely on the results provided in \cite{chicco2019computational} to gauge the accuracy of our method. The authors in \cite{chicco2019computational} conclude \say{lung side} and 
\say{platelet count} to be the two most important variables in classifying whether a patient has mesothelioma or not, and this is consistent with the findings with our method, even in the setting of having the asbestos related variables considered hidden confounders and removed. In terms of magnitude, the coefficient values pertaining to \say{lung side} and \say{platelet count} are larger than the other features. Since all of the features were standardized beforehand, this is a good indication that our method produces reasonable results. Additionally, in \cite{chicco2019computational}, the authors claim that there is a positive correlation between \say{lung side} and the \say{diagnosis of mesothelioma,} while there is a negative correlation between \say{platelet count} and the \say{diagnosis of mesothelioma.} Our results are aligned with this fact from the literature as well, since the corresponding estimates given in Table \ref{table_tensor} are positive and negative, respectively. While the \textsc{naive mle} method showed similar results, the biggest difference was in the magnitude of the most important covariate, \say{Platelet,} for which our \textsc{G-hive} method better represented the stark effect. Thus, our results appear to be in line with the current medical literature, even in the presence of hidden confounders. 

We also applied our \textsc{G-hive} procedure to analyze another NHANES dataset (\cite{nhanes2017}). We focused on the general estimation ability of \textsc{G-hive} and also highlighted the effect of hidden variables in the context of confounding. The detailed results and discussion are deferred to Appendix \ref{sec_NHANES}.

\section{Discussion}

In this paper we introduced \textsc{G-hive}, a unified framework and implementable pipeline for estimation and approximate inference in multivariate response GLMs with hidden variables that combines a novel bias correcting step with reweighted estimating equations and a spectral decomposition based projection step. More specifically, we define a novel pseudo-parameter $F^*$ via an inverse variance reweighted score, then remove its leading bias term by projecting onto the orthogonal complement of the latent factor column space via PCA on the covariance matrix of reweighted residuals. Theoretically, we derive the convergence rates of the first and second-order approximations, $F^*$ and $P_B^{\perp}F^*$, to the true parameter, $\Theta$. We also establish convergence rates for the estimation error of $\hat F$, $\hat P_B$, and the final proposed estimator, $\hat\Theta=\hat P_B^{\perp}\hat F$. A Berry–Esseen bound that leads to valid Gaussian inference for linear combinations of $P_B^{\perp}F^*$ is also derived. Empirically, our simulation results show that \textsc{G-hive} is much more robust to confounding compared to the baseline method in multiple $p,M,n$ settings, and these robustness and deconfounding benefits of \textsc{G-hive} were shown to extend to real-data analyses with lung cancer data and the NHANES dataset. 

Our approach relies on standard but substantive assumptions. Directions for extending the current work include handling the general $\Sigma_W$ setting and the high-dimensional $p>n$ setting via regularized estimators, but we leave
these topics for future study.
%In this paper, we study the estimation problem for generalized linear models with hidden variables. We proposed to use a reweighting scheme for the residuals to define the misspecified parameter, $F^*$, and showed that taking advantage of this $F^*$ improved the rate for the misspecification error to be consistent as $p \rightarrow \infty$. We also proposed a novel method, G-HIVE, that first incorporates constrained estimating equations in an EE step that utilizes the reweighting scheme and then uses a PCA step to further reduce the misspecification bias to get a faster rate as $p \rightarrow \infty$. We demonstrated in simulations that G-HIVE can outperform the naive MLE method in several settings and we applied G-HIVE to a real mesothelioma dataset to show the feasibility of our method. 

%Although we derived the rate to be $1/n^{1/3}$ under a specific regime $p^3 \asymp M \asymp n$, this relies on the conjecture that $\xi_{\text{rate}} \lesssim \sqrt{p/n}$, which we have not shown in this paper. Investigating this would be the natural next step in extending the paper. Another interesting extension is finding a way to allow $p>n$ and still maintaining consistency. For now, however, we leave them for future study. 

\bibliographystyle{ims}
\bibliography{ref.bib}

\begin{thebibliography}{46}
\expandafter\ifx\csname natexlab\endcsname\relax\def\natexlab#1{#1}\fi
\expandafter\ifx\csname url\endcsname\relax
  \def\url#1{\texttt{#1}}\fi
\expandafter\ifx\csname urlprefix\endcsname\relax\def\urlprefix{URL }\fi

\bibitem[{Ahn and Horenstein(2013)}]{ahn2013eigenvalue}
\textsc{Ahn, S.~C.} and \textsc{Horenstein, A.~R.} (2013).
\newblock Eigenvalue ratio test for the number of factors.
\newblock \textit{Econometrica} \textbf{81} 1203--1227.

\bibitem[{Bai(2003)}]{bai2003inferential}
\textsc{Bai, J.} (2003).
\newblock Inferential theory for factor models of large dimensions.
\newblock \textit{Econometrica} \textbf{71} 135--171.

\bibitem[{Bartholomew et~al.(2011)Bartholomew, Knott and Moustaki}]{bartholomew2011latent}
\textsc{Bartholomew, D.~J.}, \textsc{Knott, M.} and \textsc{Moustaki, I.} (2011).
\newblock \textit{Latent variable models and factor analysis: A unified approach}.
\newblock John Wiley \& Sons.

\bibitem[{Bing et~al.(2023)Bing, Cheng, Feng and Ning}]{bing2023inference}
\textsc{Bing, X.}, \textsc{Cheng, W.}, \textsc{Feng, H.} and \textsc{Ning, Y.} (2023).
\newblock Inference in high-dimensional multivariate response regression with hidden variables.
\newblock \textit{Journal of the American Statistical Association}  1--12.

\bibitem[{Bing et~al.(2022)Bing, Ning and Xu}]{bing2022adaptive}
\textsc{Bing, X.}, \textsc{Ning, Y.} and \textsc{Xu, Y.} (2022).
\newblock Adaptive estimation in multivariate response regression with hidden variables.
\newblock \textit{The Annals of Statistics} \textbf{50} 640--672.

\bibitem[{Breslow and Clayton(1993)}]{breslow1993approximate}
\textsc{Breslow, N.~E.} and \textsc{Clayton, D.~G.} (1993).
\newblock Approximate inference in generalized linear mixed models.
\newblock \textit{Journal of the American statistical Association} \textbf{88} 9--25.

\bibitem[{{Centers for Disease Control and Prevention} and {National Center for Health Statistics}(2018)}]{nhanes2017}
\textsc{{Centers for Disease Control and Prevention}} and \textsc{{National Center for Health Statistics}} (2018).
\newblock {National Health and Nutrition Examination Survey Data}.
\newblock \url{https://www.cdc.gov/nchs/nhanes/}.
\newblock U.S. Department of Health and Human Services, Centers for Disease Control and Prevention, Hyattsville, MD. Accessed [August 4, 2025].

\bibitem[{{\'C}evid et~al.(2020){\'C}evid, B{\"u}hlmann and Meinshausen}]{cevid2020spectral}
\textsc{{\'C}evid, D.}, \textsc{B{\"u}hlmann, P.} and \textsc{Meinshausen, N.} (2020).
\newblock Spectral deconfounding via perturbed sparse linear models.
\newblock \textit{Journal of Machine Learning Research} \textbf{21} 232.

\bibitem[{Chang et~al.(2015)Chang, Guo and Yao}]{chang2015high}
\textsc{Chang, J.}, \textsc{Guo, B.} and \textsc{Yao, Q.} (2015).
\newblock High dimensional stochastic regression with latent factors, endogeneity and nonlinearity.
\newblock \textit{Journal of Econometrics} \textbf{189} 297--312.

\bibitem[{Chicco and Rovelli(2019)}]{chicco2019computational}
\textsc{Chicco, D.} and \textsc{Rovelli, C.} (2019).
\newblock Computational prediction of diagnosis and feature selection on mesothelioma patient health records.
\newblock \textit{PloS one} \textbf{14} e0208737.

\bibitem[{{Committee on Asbestos: Selected Health Effects}(2006)}]{board2006asbestos}
\textsc{{Committee on Asbestos: Selected Health Effects}} (2006).
\newblock \textit{Asbestos: selected cancers}.
\newblock National Academies Press.

\bibitem[{Du et~al.(2025)Du, Wasserman and Roeder}]{du2025simultaneous}
\textsc{Du, J.-H.}, \textsc{Wasserman, L.} and \textsc{Roeder, K.} (2025).
\newblock Simultaneous inference for generalized linear models with unmeasured confounders.
\newblock \textit{Journal of the American Statistical Association}  1--15.

\bibitem[{Elgaddal et~al.(2024)Elgaddal, Kramarow, Weeks and Reuben}]{elgaddal2024arthritis}
\textsc{Elgaddal, N.}, \textsc{Kramarow, E.~A.}, \textsc{Weeks, J.~D.} and \textsc{Reuben, C.} (2024).
\newblock Arthritis in adults age 18 and older: United states, 2022.

\bibitem[{ElSayed et~al.(2025)ElSayed, McCoy, Aleppo, Balapattabi, Beverly, Briggs~Early, Bruemmer, Ebekozien, Echouffo-Tcheugui, Ekhlaspour et~al.}]{american20252}
\textsc{ElSayed, N.~A.}, \textsc{McCoy, R.~G.}, \textsc{Aleppo, G.}, \textsc{Balapattabi, K.}, \textsc{Beverly, E.~A.}, \textsc{Briggs~Early, K.}, \textsc{Bruemmer, D.}, \textsc{Ebekozien, O.}, \textsc{Echouffo-Tcheugui, J.~B.}, \textsc{Ekhlaspour, L.} \textsc{et~al.} (2025).
\newblock 2. diagnosis and classification of diabetes: Standards of care in diabetes—2025.
\newblock \textit{Diabetes Care} \textbf{48}.

\bibitem[{Fan et~al.(2008)Fan, Fan and Lv}]{fan2008high}
\textsc{Fan, J.}, \textsc{Fan, Y.} and \textsc{Lv, J.} (2008).
\newblock High dimensional covariance matrix estimation using a factor model.
\newblock \textit{Journal of Econometrics} \textbf{147} 186--197.

\bibitem[{Fan et~al.(2013)Fan, Liao and Mincheva}]{fan2013large}
\textsc{Fan, J.}, \textsc{Liao, Y.} and \textsc{Mincheva, M.} (2013).
\newblock Large covariance estimation by thresholding principal orthogonal complements.
\newblock \textit{Journal of the Royal Statistical Society: Series B (Statistical Methodology)} \textbf{75} 603--680.

\bibitem[{Fan et~al.(2024)Fan, Lou and Yu}]{fan2024latent}
\textsc{Fan, J.}, \textsc{Lou, Z.} and \textsc{Yu, M.} (2024).
\newblock Are latent factor regression and sparse regression adequate?
\newblock \textit{Journal of the American Statistical Association} \textbf{119} 1076--1088.

\bibitem[{G{\"o}tze et~al.(2021)G{\"o}tze, Sambale and Sinulis}]{gotze2021concentration}
\textsc{G{\"o}tze, F.}, \textsc{Sambale, H.} and \textsc{Sinulis, A.} (2021).
\newblock Concentration inequalities for polynomials in $\alpha$-sub-exponential random variables.
\newblock \textit{Electronic Journal of Probability} \textbf{26} 1--22.

\bibitem[{Guo et~al.(2022)Guo, {\'C}evid and B{\"u}hlmann}]{guo2022doubly}
\textsc{Guo, Z.}, \textsc{{\'C}evid, D.} and \textsc{B{\"u}hlmann, P.} (2022).
\newblock Doubly debiased lasso: High-dimensional inference under hidden confounding.
\newblock \textit{Annals of statistics} \textbf{50} 1320.

\bibitem[{Horn and Johnson(1994)}]{horn1994topics}
\textsc{Horn, R.~A.} and \textsc{Johnson, C.~R.} (1994).
\newblock \textit{Topics in matrix analysis}.
\newblock Cambridge university press.

\bibitem[{Hu et~al.(2001)Hu, Manson, Stampfer, Colditz, Liu, Solomon and Willett}]{hu2001diet}
\textsc{Hu, F.~B.}, \textsc{Manson, J.~E.}, \textsc{Stampfer, M.~J.}, \textsc{Colditz, G.}, \textsc{Liu, S.}, \textsc{Solomon, C.~G.} and \textsc{Willett, W.~C.} (2001).
\newblock Diet, lifestyle, and the risk of type 2 diabetes mellitus in women.
\newblock \textit{New England journal of medicine} \textbf{345} 790--797.

\bibitem[{Huber et~al.(2004)Huber, Ronchetti and Victoria-Feser}]{huber2004estimation}
\textsc{Huber, P.}, \textsc{Ronchetti, E.} and \textsc{Victoria-Feser, M.-P.} (2004).
\newblock Estimation of generalized linear latent variable models.
\newblock \textit{Journal of the Royal Statistical Society Series B: Statistical Methodology} \textbf{66} 893--908.

\bibitem[{Irizarry et~al.(2005)Irizarry, Warren, Spencer, Kim, Biswal, Frank, Gabrielson, Garcia, Geoghegan, Germino et~al.}]{irizarry2005multiple}
\textsc{Irizarry, R.~A.}, \textsc{Warren, D.}, \textsc{Spencer, F.}, \textsc{Kim, I.~F.}, \textsc{Biswal, S.}, \textsc{Frank, B.~C.}, \textsc{Gabrielson, E.}, \textsc{Garcia, J.~G.}, \textsc{Geoghegan, J.}, \textsc{Germino, G.} \textsc{et~al.} (2005).
\newblock Multiple-laboratory comparison of microarray platforms.
\newblock \textit{Nature methods} \textbf{2} 345--350.

\bibitem[{Katsaouni et~al.(2021)Katsaouni, Tashkandi, Wiese and Schulz}]{katsaouni2021machine}
\textsc{Katsaouni, N.}, \textsc{Tashkandi, A.}, \textsc{Wiese, L.} and \textsc{Schulz, M.~H.} (2021).
\newblock Machine learning based disease prediction from genotype data.
\newblock \textit{Biological Chemistry} \textbf{402} 871--885.

\bibitem[{Koltchinskii and Lounici(2017)}]{koltchinskii2017concentration}
\textsc{Koltchinskii, V.} and \textsc{Lounici, K.} (2017).
\newblock Concentration inequalities and moment bounds for sample covariance operators.
\newblock \textit{Bernoulli}  110--133.

\bibitem[{Lam and Yao(2012)}]{lam2012factor}
\textsc{Lam, C.} and \textsc{Yao, Q.} (2012).
\newblock Factor modeling for high-dimensional time series: inference for the number of factors.
\newblock \textit{The Annals of Statistics}  694--726.

\bibitem[{Lee et~al.(2017)Lee, Sun, Wright and Zou}]{lee2017improved}
\textsc{Lee, S.}, \textsc{Sun, W.}, \textsc{Wright, F.~A.} and \textsc{Zou, F.} (2017).
\newblock An improved and explicit surrogate variable analysis procedure by coefficient adjustment.
\newblock \textit{Biometrika} \textbf{104} 303--316.

\bibitem[{Leek and Storey(2007)}]{leek2007capturing}
\textsc{Leek, J.~T.} and \textsc{Storey, J.~D.} (2007).
\newblock Capturing heterogeneity in gene expression studies by surrogate variable analysis.
\newblock \textit{PLoS genetics} \textbf{3} e161.

\bibitem[{Luo and Wei(2019)}]{luo2019batch}
\textsc{Luo, X.} and \textsc{Wei, Y.} (2019).
\newblock Batch effects correction with unknown subtypes.
\newblock \textit{Journal of the American Statistical Association} .

\bibitem[{McCulloch(2001)}]{mcculloch2001generalized}
\textsc{McCulloch, C.} (2001).
\newblock \textit{Generalized, linear, and mixed models}.
\newblock Wiley.

\bibitem[{McKennan and Nicolae(2019)}]{mckennan2019accounting}
\textsc{McKennan, C.} and \textsc{Nicolae, D.} (2019).
\newblock Accounting for unobserved covariates with varying degrees of estimability in high-dimensional biological data.
\newblock \textit{Biometrika} \textbf{106} 823--840.

\bibitem[{Mincer(1974)}]{mincer1974schooling}
\textsc{Mincer, J.} (1974).
\newblock \textit{Schooling, experience, and earnings. Human behavior \& social institutions no. 2.}
\newblock ERIC.

\bibitem[{Ning and Liu(2017)}]{ning2017general}
\textsc{Ning, Y.} and \textsc{Liu, H.} (2017).
\newblock A general theory of hypothesis tests and confidence regions for sparse high dimensional models.
\newblock \textit{The Annals of Statistics} \textbf{45} 158--195.

\bibitem[{Ouyang et~al.(2023)Ouyang, Tan and Xu}]{ouyang2023high}
\textsc{Ouyang, J.}, \textsc{Tan, K.~M.} and \textsc{Xu, G.} (2023).
\newblock High-dimensional inference for generalized linear models with hidden confounding.
\newblock \textit{The Journal of Machine Learning Research} \textbf{24} 14030--14090.

\bibitem[{Parikh et~al.(2008)Parikh, Pencina, Wang, Benjamin, Lanier, Levy, D'Agostino~Sr, Kannel and Vasan}]{parikh2008risk}
\textsc{Parikh, N.~I.}, \textsc{Pencina, M.~J.}, \textsc{Wang, T.~J.}, \textsc{Benjamin, E.~J.}, \textsc{Lanier, K.~J.}, \textsc{Levy, D.}, \textsc{D'Agostino~Sr, R.~B.}, \textsc{Kannel, W.~B.} and \textsc{Vasan, R.~S.} (2008).
\newblock A risk score for predicting near-term incidence of hypertension: the framingham heart study.
\newblock \textit{Annals of internal medicine} \textbf{148} 102--110.

\bibitem[{Pearl(2009)}]{pearl2009causality}
\textsc{Pearl, J.} (2009).
\newblock \textit{Causality}.
\newblock Cambridge university press.

\bibitem[{Portnoy(1984)}]{portnoy1984asymptotic}
\textsc{Portnoy, S.} (1984).
\newblock Asymptotic behavior of m-estimators of p regression parameters when p 2/n is large. i. consistency.
\newblock \textit{The Annals of Statistics}  1298--1309.

\bibitem[{Schur(2014)}]{nhanesA}
\textsc{Schur, A.} (2014).
\newblock nhanesa: Nhanes data retrieval.
\newblock \url{https://CRAN.R-project.org/package=nhanesA}.
\newblock R package version 0.6.6.

\bibitem[{Sun et~al.(2024)Sun, Ma and Xia}]{sun2024decorrelating}
\textsc{Sun, Y.}, \textsc{Ma, L.} and \textsc{Xia, Y.} (2024).
\newblock A decorrelating and debiasing approach to simultaneous inference for high-dimensional confounded models.
\newblock \textit{Journal of the American Statistical Association} \textbf{119} 2857--2868.

\bibitem[{Vershynin(2018)}]{vershynin2018high}
\textsc{Vershynin, R.} (2018).
\newblock \textit{High-dimensional probability: An introduction with applications in data science}, vol.~47.
\newblock Cambridge university press.

\bibitem[{Wang et~al.(2017)Wang, Zhao, Hastie and Owen}]{wang2017confounder}
\textsc{Wang, J.}, \textsc{Zhao, Q.}, \textsc{Hastie, T.} and \textsc{Owen, A.~B.} (2017).
\newblock Confounder adjustment in multiple hypothesis testing.
\newblock \textit{Annals of statistics} \textbf{45} 1863.

\bibitem[{Wang and Shah(2025)}]{wang2025latent}
\textsc{Wang, Y.} and \textsc{Shah, R.} (2025).
\newblock Latent confounding in high-dimensional nonlinear models.
\newblock \textit{arXiv preprint arXiv:2508.06274} .

\bibitem[{Wedderburn(1974)}]{wedderburn1974quasi}
\textsc{Wedderburn, R.~W.} (1974).
\newblock Quasi-likelihood functions, generalized linear models, and the gauss—newton method.
\newblock \textit{Biometrika} \textbf{61} 439--447.

\bibitem[{White(1982)}]{white1982maximum}
\textsc{White, H.} (1982).
\newblock Maximum likelihood estimation of misspecified models.
\newblock \textit{Econometrica: Journal of the econometric society}  1--25.

\bibitem[{Wilson et~al.(2007)Wilson, Meigs, Sullivan, Fox, Nathan and D’Agostino}]{wilson2007prediction}
\textsc{Wilson, P.~W.}, \textsc{Meigs, J.~B.}, \textsc{Sullivan, L.}, \textsc{Fox, C.~S.}, \textsc{Nathan, D.~M.} and \textsc{D’Agostino, R.~B.} (2007).
\newblock Prediction of incident diabetes mellitus in middle-aged adults: the framingham offspring study.
\newblock \textit{Archives of internal medicine} \textbf{167} 1068--1074.

\bibitem[{Yu et~al.(2015)Yu, Wang and Samworth}]{yu2015useful}
\textsc{Yu, Y.}, \textsc{Wang, T.} and \textsc{Samworth, R.~J.} (2015).
\newblock A useful variant of the davis--kahan theorem for statisticians.
\newblock \textit{Biometrika} \textbf{102} 315--323.

\end{thebibliography}

\newpage

\appendix

\bigskip
\bigskip

\noindent {\bf \Large Appendix}\\

\noindent In the appendix, we use $C$ as a generic constant, which can be different in different lines. 

%\section{Comments on Score/Likelihood Equations} \label{section_score_function}
%If we are given a (misspecified) generalized linear model with a conditional pdf 
%\[g(Y_m|X)=e^{\{Y_m\cdot(F_m X ) - b(F_m X)\}/\phi ~+ ~c(Y_m, \phi)}\]
%with $\phi=1$ for simplicity, then $E[Y_m|X] = b'(F_m X)$ holds and thus, 
%If we define $F_m$ as the true parameter that satisfies $E[Y_m - b'(F_mX)|X]=0$, the following hold:
%\begin{align}
    %& E_{Y_m|X,Z}\Big[\big\{Y_m - b'(F_mX)\big\}X^T\Big|X,Z\Big] = E_{Y_m|X,Z}\Big[\big\{Y_m - b'(F_mX)\big\}\Big|X,Z\Big]X^T = 0 \nonumber \\
    %& E_{Y_m|X}\Big[\big\{Y_m - b'(F_mX)\big\}X^T\Big| X\Big] = E_Z\bigg(E_{Y_m|X,Z}\Big[\big\{Y_m - b'(F_mX)\big\}X^T\Big|X,Z\Big]\bigg)=0 \nonumber \\
%    & E \Big[\big\{Y_m - b'(F_mX)\big\}X^T\Big] = E_X\bigg(E_{Y_m|X}\Big[\big\{Y_m - b'(F_mX)\big\}X^T\Big|X\Big]\bigg)=0 \label{score_eq_MLE}\\
%    & E_{Y_m|X}\bigg[\bigg\{\frac{Y_m - b'(F_mX)}{b''(F_m X)}\bigg\}X^T\bigg|X\bigg] = 
%    \frac{1}{b''(F_m X)}\cdot E_{Y_m|X}\Big[\big\{Y_m - b'(F_mX)\big\}\Big|X\Big]\cdot X^T =0 \nonumber \\
%    & E \bigg[\bigg\{\frac{Y_m - b'(F_mX)}{b''(F_m X)}\bigg\}X^T\bigg] = E_X \Bigg( E_{Y_m|X}\bigg[\bigg\{\frac{Y_m - b'(F_mX)}{b''(F_m X)}\bigg\}X^T\bigg|X\bigg]\Bigg) = 0 \label{score_eq_bias}  
%\end{align}
%So, when analyzing $||\hat{F}_m - \Theta_m||_2 \leq ||\hat{F}_m- F_m||_2 + ||F_m - \Theta_m||_2$, the first MLE term can be analyzed with the score equation from (\ref{score_eq_MLE}) and the second misspecification bias term can be analyzed with the re-weighted score equation from (\ref{score_eq_bias}).

\section{Remark on Assumption \ref{assum_pervasiveness}} \label{sec_per}
In this section we provide an example of when Assumption \ref{assum_pervasiveness}
 is satisfied. Suppose each element in the $M \times K$ matrix $B$ is i.i.d. from a centered Gaussian distribution with a bounded variance. Without loss of generality, let's assume the variance $\sigma^2=1$. By standard high probability bounds for Gaussian covariance estimation (\cite{vershynin2018high}) and Weyl's inequality (\cite{horn1994topics}), it can be shown that the smallest eigenvalue of $B^TB$ is on the order of $M$. The same logic can be applied to $A$ to show that its smallest eigenvalue is on the order of $p$ with high probability. Standard Gaussian covariance estimation bounds give us with probability $\geq 1-2e^{-t}$ that
\begin{align*}
    \Bigg|\Bigg|\frac{1}{M}\sum_{i=1}^M B^TB - \EE\big[B^TB\big]\Bigg|\Bigg|_{\text{op}} ~&\leq~ C \cdot \Bigg( \sqrt{\frac{K+t}{M}} ~+~ \frac{K + t}{M}\Bigg) \\
    \Rightarrow~~\Bigg|\Bigg|\frac{1}{M}\sum_{i=1}^M B^TB - I_K\Bigg|\Bigg|_{\text{op}} ~&\leq~ C \cdot \Bigg( \sqrt{\frac{K+t}{M}} ~+~ \frac{K + t}{M}\Bigg)
\end{align*}
 for some constant $C > 0$ that does not depend on $M$. This holds because we can regard the $M$ columns in $B^T$ as independent realizations of $N(0,I_K)$. Then we have from Weyl's inequality that
 \begin{align*}
    \lambda_{\min}\bigg(\frac{1}{M}B^TB\bigg) ~&\geq~ \lambda_{\min}\bigg(\frac{1}{M}B^TB - I_K\bigg)  + \lambda_{\min}\bigg(I_K\bigg) \\
    ~&\geq~ -\Bigg|\Bigg|\frac{1}{M}\sum_{i=1}^M B^TB - I_K\Bigg|\Bigg|_{\text{op}}+\lambda_{\min}\bigg(I_K\bigg) \\
    ~&\geq~ 1 ~-~ C \cdot \Bigg( \sqrt{\frac{K+t}{M}} ~+~ \frac{K + t}{M}\Bigg)
 \end{align*}
 where the last line holds with probability $\geq 1-2e^{-t}$. Choosing $ t = \log M$, it is apparent that 
 \begin{align*}
     \lambda_{\min}(B^TB) ~&\geq~ M ~-~ C\Big(\sqrt{M\log M} ~+~ \log M ~+~ K\Big) \\
     ~&\gtrsim~ M
 \end{align*}

\section{Collection of Proofs}\label{sec_theory_proof}

\begin{lemma}\label{lem_tildeZ}
Under Assumptions \ref{assum_ident}- \ref{assum_pervasiveness}, we have $$
||B_m(Z-\tilde{Z})||_{\psi_2} \leq c/\sqrt{p}$$ 
for some fixed constant $c>0$, where $\tilde{Z}:=\Sigma_ZA^T(A\Sigma_Z A^T + I_p)^{-1}X$.  In addition, $\Sigma_X^{-1/2}X$ is a centered  sub-Gaussian random vector with bounded sub-Gaussian norm.  
\end{lemma}
\begin{proof}
From the definition of $\tilde{Z}:=\Sigma_ZA^T(A\Sigma_Z A^T + I_p)^{-1}X$ and the latent factor model $X=AZ+W$, we have the following:
\begin{align*}
    B_m(Z-\tilde{Z}) ~&=B_m(I-\Sigma_ZA^T(A\Sigma_Z A^T + I_p)^{-1}A)Z-B_m\Sigma_ZA^T(A\Sigma_Z A^T + I_p)^{-1}W\\
    &=B_m(\Sigma_Z-\Sigma_ZA^T(A\Sigma_Z A^T + I_p)^{-1}A\Sigma_Z)\Sigma_Z^{-1}Z-B_m\Sigma_ZA^T(A\Sigma_Z A^T + I_p)^{-1}W\\
    &=B_m(\Sigma_Z^{-1}+ A^TA)^{-1}\Sigma_Z^{-1}Z-B_m(\Sigma_Z^{-1}+ A^TA)^{-1}A^TW,
\end{align*}
where the last line follows from the block matrix inverse formula. Since $W$ and $Z$ are sub-Gaussian random vectors, we merely need to compute the Euclidean norms of $B_m(\Sigma_Z^{-1}+ A^TA)^{-1}A^T$ and $B_m(\Sigma_Z^{-1}+ A^TA)^{-1}\Sigma_Z^{-1}$,
\begin{align*}
    ||B_m(\Sigma_Z^{-1}+ A^TA)^{-1}A^T||_2 ~&=~ \sqrt{B_m(\Sigma_Z^{-1}+ A^TA)^{-1}A^TA(\Sigma_Z^{-1}+ A^TA)^{-1}B_m^T}\\
    ~&\leq~ ||B_m||_2 \cdot ||(\Sigma_Z^{-1}+ A^TA)^{-1}A^T||_{\text{op}} \\
    ~&\leq~ ||B_m||_2 \cdot ||(\Sigma_Z^{-1}+ A^TA)^{-1}||_{\text{op}}\cdot ||A||_{\text{op}}\\
    ~&\leq~ ||B_m||_2\cdot\frac{\sqrt{\lambda_{\max}(A^TA)}}{\lambda_{\min}(A^TA) + 1/\lambda_{\max}(\Sigma_Z)}\\
    ~&\lesssim \frac{1}{\sqrt{p}},\\
    ||B_m(\Sigma_Z^{-1}+ A^TA)^{-1}\Sigma_Z^{-1}||_2 ~&\leq~ ||B_m||_2 \cdot ||(\Sigma_Z^{-1}+ A^TA)^{-1}\Sigma_Z^{-1}||_{\text{op}}\\
    ~&\leq~ ||B_m||_2 \cdot ||(\Sigma_Z^{-1}+ A^TA)^{-1}||_{\text{op}}\cdot||\Sigma_Z^{-1}||_{\text{op}} \\
    ~&\lesssim \frac{1}{p}.
\end{align*}
Note that $\Sigma_X^{-1/2}X=\Sigma_X^{-1/2}AZ+\Sigma_X^{-1/2}W$. Following a similar derivation, we can show that $\|\Sigma_X^{-1/2}A\|_{\textrm{op}}$ and $\|\Sigma_X^{-1/2}\|_{\textrm{op}}$ are upper bounded by a constant. Thus, $\Sigma_X^{-1/2}X$ is a centered  sub-Gaussian vector. This completes the proof.
\end{proof}

\subsection{Proof for Theorem \ref{thm_expectation}}
\begin{proof}
    For ease of notation, for this proof, we denote $F_m^*$ as $F_m$. Recall from (\ref{true_error}) that 
    $$\epsilon_m := \frac{Y_m-b'(\Theta_m X + B_mZ)}{b''(\Theta_m X + B_mZ)}.$$
    Thus, $Y_m = b'(\Theta_m X + B_mZ)+\epsilon_m\cdot b''(\Theta_m X + B_mZ)$. Plugging this into (\ref{score_reweighted}) which is
    $$\EE\bigg[\bigg\{\frac{Y_m - b'(F_mX)}{b''(F_m X)}\bigg\}X^T \bigg]=0$$ we get
    \begin{align*}
        \EE\Bigg[\bigg\{\frac{b'(\Theta_m X + B_mZ)-b'(F_m X)}{b''(F_m X)}~+~ \epsilon_m\cdot\frac{b''(\Theta_m X + B_mZ)}{b''(F_m X)}\bigg\}X^T\Bigg]~=~0.
    \end{align*}
    The second term is $0$ as $\EE(\epsilon_m|X,Z)=0$. By Taylor expansion we have
    \begin{align*}
        b'(\Theta_m X + B_mZ)= b'(F_m X) ~+~ b''(F_m X)\cdot(\Theta_m X + B_m Z - F_m X) ~+~ \frac{b'''(\delta_m)}{2}\cdot(\Theta_m X + B_m Z - F_m X)^2
    \end{align*}
    for some intermediate $\delta_m$ between $F_m X$ and $\Theta_m X + B_m Z$. Rearranging we get
    \begin{align}
        \EE\Big[(\Theta_m - F_m&)XX^T + B_mZX^T\Big] ~+~\EE\bigg[\frac{b'''(\delta_m)}{2b''(F_m X)}\cdot(\Theta_m X + B_m Z - F_m X)^2\cdot X^T\bigg]~=~0, \nonumber
        \end{align}    
        \begin{align}
        F_m - \Theta_m ~&=~ B_m\cdot\EE\big[ZX^T\big]\cdot\Big\{\EE\big(XX^T\big)\Big\}^{-1} ~+~ \EE\bigg[\frac{b'''(\delta_m)}{2b''(F_m X)}\cdot(\Theta_m X + B_m Z - F_m X)^2\cdot X^T\bigg]\cdot\Big\{\EE\big(XX^T\big)\Big\}^{-1} \nonumber \\
        ~&=~B_m\Sigma_Z A^T(A\Sigma_Z A^T + I_p)^{-1}~+~\EE\bigg[\frac{b'''(\delta_m)}{2b''(F_m X)}\cdot(\Theta_m X + B_m Z - F_m X)^2\cdot X^T\bigg]\cdot\Sigma_X^{-1}. \label{eq.2}
    \end{align}
    We define the following set 
    \begin{align}
        \Omega_{m}:=~ \bigg\{F_m \in \RR^{1 \times p}~:~ \EE\Big[\big(\Theta_mX + B_m Z - F_m X\big)^4\Big] ~\leq~ \gamma^4 \bigg\}, \label{class}
    \end{align}
    where $\gamma>0$ is to be specified later on. In the following, we will show that $F_m$ implicitly given by (\ref{eq.2}) belongs to $\Omega_{m}$. To this end, we plug (\ref{eq.2}) into the condition in (\ref{class}) to obtain the $\gamma$ that will satisfy this inequality. 
    \begin{align*}
        \big(\Theta_mX + B_m Z - F_m X\big) ~&=~ (\Theta_m - F_m)X + B_m Z\\
        ~&=~ -B_m\Sigma_Z A^T(A\Sigma_Z A^T + I_p)^{-1}X ~+~ B_m Z \\
        ~&~ \hspace{3cm} -~\EE\bigg[\frac{b'''(\delta_m)}{2b''(F_m X)}\cdot(\Theta_m X + B_m Z - F_m X)^2\cdot X^T\bigg]\cdot\Sigma_X^{-1}X \\
        ~&=~ B_m(Z-\tilde{Z}) ~-~\EE\bigg[\frac{b'''(\delta_m)}{2b''(F_m X)}\cdot(\Theta_m X + B_m Z - F_m X)^2\cdot X^T\bigg]\cdot\Sigma_X^{-1}X, 
        \end{align*}
    where $\tilde{Z}:=\Sigma_ZA^T(A\Sigma_Z A^T + I_p)^{-1}X$. Using the inequality $(a+b)^2 \leq 2a^2 + 2b^2$  twice, we get:
    \begin{align}
        \big(\Theta_mX + B_m Z - F_m X\big)^2 &\leq~ 2\cdot\Big\{ B_m(Z-\tilde{Z})\Big\}^2 \nonumber \\
        &\hspace{1cm}~+~ 2\cdot\Bigg\{ \EE\bigg[\frac{b'''(\delta_m)}{2b''(F_m X)}\cdot(\Theta_m X + B_m Z - F_m X)^2\cdot (\Sigma_X^{-1/2}X)^T\bigg]\cdot\Sigma_X^{-1/2}X\Bigg\}^2 \nonumber \\
        \big(\Theta_mX + B_m Z - F_m X\big)^4 &\leq~ 8\cdot\Big\{ B_m(Z-\tilde{Z})\Big\}^4 \nonumber \\
        &\hspace{1cm}~+~ 8\cdot\Bigg\{ \EE\bigg[\frac{b'''(\delta_m)}{2b''(F_m X)}\cdot(\Theta_m X + B_m Z - F_m X)^2\cdot (\Sigma_X^{-1/2}X)^T\bigg]\cdot\Sigma_X^{-1/2}X\Bigg\}^4 \nonumber
        \end{align}
    and     
        \begin{align}
        \EE\Big[\big(\Theta_mX + B_m Z - F_m X\big)^4\Big] &\leq~ 8\cdot\EE\bigg[\Big\{ B_m(Z-\tilde{Z})\Big\}^4\bigg] \nonumber \\ &~+~ 8\cdot\EE\Bigg[\Bigg\{ \EE\bigg[\frac{b'''(\delta_m)}{2b''(F_m X)}\cdot(\Theta_m X + B_m Z - F_m X)^2\cdot (\Sigma_X^{-1/2}X)^T\bigg]\cdot\Sigma_X^{-1/2}X\Bigg\}^4\Bigg].  \label{eq. two terms}
    \end{align}
    The first term in (\ref{eq. two terms}) can be bounded by noting the sub-Gaussian property $||B_m(Z-\tilde{Z})||_{\psi_2} \leq c/\sqrt{p}$ in Lemma \ref{lem_tildeZ},
    \begin{align}\label{sub_G norm}
        \underset{q \geq 1}{\sup}~\frac{1}{\sqrt{q}}\bigg\{\EE\Big|B_m(Z-\tilde{Z})\Big|^q\bigg\}^{1/q} ~&\leq~ \frac{c}{\sqrt{p}},~~\textrm{with $q=4$ implying}~~
        \EE\Big\{B_m(Z-\tilde{Z})\Big\}^4 ~\leq~ \frac{16c^4}{p^2}.
    \end{align}
    For the second term in (\ref{eq. two terms}), Lemma \ref{lem_tildeZ} implies $\Sigma_X^{-1/2}X$ is a centered, sub-Gaussian random vector, where $\Sigma_X^{-1/2}X$ has a bounded $\psi_2$-norm, i.e., that $||v^T(\Sigma_X^{-1/2}X)||_{\psi_2}\leq c_x$ for any $||v||_2=1$ for some fixed constant $c_x > 0$. Then, following the definition of the sub-Gaussian norm, like in (\ref{sub_G norm}), we derive:
    \begin{align}
        \underset{||v||_2=1}{\sup}~\bigg\{\EE\Big|v^T(\Sigma_X^{-1/2}X)\Big|^4\bigg\}~\leq~16c_x^4,~~\textrm{and}~~        \underset{||v||_2=1}{\sup}~\bigg\{\EE\Big|v^T(\Sigma_X^{-1/2}X)\Big|^2\bigg\}~\leq~2c_x^2. \label{sup_subG}
    \end{align}
For ease of notation, let us denote
    \[
    \phi_m:=\frac{b'''(\delta_m)}{2b''(F_m X)}\cdot(\Theta_m X + B_m Z - F_m X)^2\cdot (\Sigma_X^{-1/2}X)^T.
    \]
    Then we have the following algebraic derivation:
    \begin{align}
        &\EE\Bigg[\Bigg\{ \EE\bigg[\frac{b'''(\delta_m)}{2b''(F_m X)}\cdot(\Theta_m X + B_m Z - F_m X)^2\cdot (\Sigma_X^{-1/2}X)^T\bigg]\cdot\Sigma_X^{-1/2}X\Bigg\}^4\Bigg] \nonumber \\
        &\hspace{2cm}=\|\EE(\phi_m)\|_2^4\cdot  \EE\Bigg\{\frac{\EE(\phi_m)}{\|\EE(\phi_m)\|_2}\Sigma_X^{-1/2}X\Bigg\}^4\nonumber\\
        &\hspace{2cm}\leq \|\EE(\phi_m)\|_2^4 \cdot \underset{||v||_2=1}{\sup}~\bigg\{\EE\Big|v^T(\Sigma_X^{-1/2}X)\Big|^4\bigg\}\nonumber\\
        &\hspace{2cm}\leq \|\EE(\phi_m)\|_2^4 \cdot 16c_x^4, \label{eq.two terms 2}
    \end{align}
    where the last step follows from (\ref{sub_G norm}).  The first term in (\ref{eq.two terms 2}) can be bounded by:
    \begin{align*}
        \Big|\Big| \EE\big[\phi_m\big]\Big|\Big|_2^4 ~&=~ \underset{||v||_2=1}{\sup} \bigg|\EE\Big[\phi_m v\Big]\bigg|^4 \\
        ~&=~ \underset{||v||_2=1}{\sup} \bigg|\EE\Big[ \frac{b'''(\delta_m)}{2b''(F_m X)}\cdot(\Theta_m X + B_m Z - F_m X)^2\cdot (\Sigma_X^{-1/2}X)^T v\Big]\bigg|^4 \\
         ~&\leq\Bigg\{\EE\bigg[ \Big\{\frac{b'''(\delta_m)}{2b''(F_m X)}\Big\}^2\cdot(\Theta_m X + B_m Z - F_m X)^4 \bigg] \cdot \underset{||v||_2=1}{\sup} \EE\Big[ (\Sigma_X^{-1/2}X)^T v \Big]^2 \Bigg\}^2 \\
         ~&\leq \Big(\frac{C_3}{2C_1}\Big)^4 \cdot \gamma^8 \cdot 4 c_x^4,
    \end{align*}
    where the third line follows from the Cauchy-Schwartz inequality and the last step follows from (\ref{sup_subG}).

    Thus, combining these results with (\ref{eq. two terms}), we get the following
    \begin{align*}
        \EE\Big[\big(\Theta_mX + B_m Z - F_m X\big)^4\Big] ~&\leq~ \frac{128 c^4}{p^2}~+~ \Big(\frac{c_x^2 \cdot C_3}{C_1}\Big)^4\cdot \gamma^8.
    \end{align*}
By setting $\gamma^4=C/p^2$ for some constant $C$ sufficiently large (where $C$ does not depend on $m$), we have $\frac{128 c^4}{p^2}+(\frac{c_x^2 \cdot C_3}{C_1})^4\cdot \gamma^8\leq \gamma^4$, which implies $F_m^* \in \Omega_{m}$. 

In the following, we will bound $||F_m - \Theta_m||_2^2$.  Note from line (\ref{eq.2}), we have
\begin{align*}
    ||F_m - \Theta_m||_2^2 ~&\leq~ 2\cdot\bigg\{\Big|\Big|B_m\Sigma_Z A^T(A\Sigma_Z A^T + I_p)^{-1}\Big|\Big|_2^2 \\
    ~& \hspace{1cm} +~ \bigg|\bigg|\EE\bigg[\frac{b'''(\delta_m)}{2b''(F_m X)}\cdot(\Theta_m X + B_m Z - F_m X)^2\cdot (\Sigma_X^{-1/2}X)^T\bigg]\cdot\Sigma_X^{-1/2}\bigg|\bigg|_2^2 \Bigg\}.
\end{align*}
Note that from Assumption \ref{assum_pervasiveness}, the first term can be upper bounded by
\begin{align*}
    \big|\big|B_m\Sigma_Z A^T(A\Sigma_Z A^T + I_p)^{-1}\big|\big|_2^2 ~&=~ \big|\big|B_m(\Sigma_Z^{-1} + A^TA)^{-1}A^T\big|\big|_2^2 \\
    ~&\leq~ ||B_m||_2^2 \cdot \big|\big|(\Sigma_Z^{-1} + A^TA)^{-1}A^TA(\Sigma_Z^{-1} + A^TA)^{-1}\big|\big|_{\text{op}}\\
    ~&=~ ||B_m||_2^2 \cdot \Big[\lambda_{\max}\big((\Sigma_Z^{-1} + A^TA)^{-1}\big)\Big]^2 \cdot \lambda_{\max}(A^TA)\\
    ~&=~\frac{||B_m||_2^2 \cdot \lambda_{\max}(A^TA)}{\Big[\lambda_{\min}(\Sigma_Z^{-1} + A^TA)\Big]^2} \\
    ~&\leq~ \frac{C_4^2 \cdot \kappa_{A,2}\cdot p} {\Big[(1/\kappa_{Z,2}) + \kappa_{A,1}\cdot p\Big]^2}\\
    ~&\lesssim~ \frac{1}{p}.
\end{align*}
The second term can be upper bounded by
\begin{align}
    &\bigg|\bigg|\EE\bigg[\frac{b'''(\delta_m)}{2b''(F_m X)}\cdot(\Theta_m X + B_m Z - F_m X)^2\cdot (\Sigma_X^{-1/2}X)^T\bigg]\cdot\Sigma_X^{-1/2}\bigg|\bigg|_2^2 \nonumber \\
    ~&\leq \bigg|\bigg|\EE\bigg[\frac{b'''(\delta_m)}{2b''(F_m X)}\cdot(\Theta_m X + B_m Z - F_m X)^2\cdot (\Sigma_X^{-1/2}X)^T\bigg]\bigg|\bigg|_2^2 \cdot \lambda_{\max}(\Sigma_X^{-1}) \nonumber \\
    ~&=~ \underset{||v||_2=1}{\sup}~\Bigg|\EE\bigg[\frac{b'''(\delta_m)}{2b''(F_m X)}\cdot(\Theta_m X + B_m Z - F_m X)^2\cdot (\Sigma_X^{-1/2}X)^T v\bigg]\Bigg|^2\cdot \lambda_{\max}(\Sigma_X^{-1})\nonumber \\
    ~&\leq~\EE\bigg[\Big(\frac{b'''(\delta_m)}{2b''(F_m X)}\Big)^2\cdot(\Theta_m X + B_m Z - F_m X)^4 \bigg]\cdot \underset{||v||_2=1}{\sup} \EE|(\Sigma_X^{-1/2}X)^T v|^2\cdot \lambda_{\max}(\Sigma_X^{-1})\nonumber \\
    ~&\lesssim~ \frac{1}{p^2}, \label{eq.rate}
\end{align}
where in the last step Assumption \ref{assum_glm} implies $\Big|\frac{b'''(\delta_m)}{2b''(F_m X)}\Big|$ is bounded, $\EE[\big(\Theta_mX + B_m Z - F_m X\big)^4]\lesssim 1/p^2$ due to $F_m^* \in \Omega_{m}$,  and we also apply Assumption \ref{assum_pervasiveness}. Thus, we have $||F_m - \Theta_m||_2^2 \lesssim (1/p)$ uniformly over $m$, and $Rem'$ (the term in (\ref{eq.rate})) satisfies $\max_{1\leq m\leq M}\|Rem_m'\|^2_2=O(1/p^2)$. 

Compiling this into a matrix with $M$ rows, we get the following:
\begin{align}
F - \Theta ~&=~ B\Sigma_Z A^T(A\Sigma_Z A^T + I_p)^{-1} ~+~ \begin{bmatrix}
    \EE\Big[\frac{b'''(\delta_1)}{2b''(F_1 X)}\cdot(\Theta_1 X + B_1 Z - F_1 X)^2\cdot (\Sigma_X^{-1/2}X)^T\Big] \\
    \EE\Big[\frac{b'''(\delta_2)}{2b''(F_2 X)}\cdot(\Theta_2 X + B_2 Z - F_2 X)^2\cdot (\Sigma_X^{-1/2}X)^T\Big]\\
    : \\
    \EE\Big[\frac{b'''(\delta_M)}{2b''(F_M X)}\cdot(\Theta_M X + B_M Z - F_M X)^2\cdot (\Sigma_X^{-1/2}X)^T\Big]
\end{bmatrix}\cdot\Sigma_X^{-1/2} \nonumber \\
~&:=~ B\Sigma_Z A^T(A\Sigma_Z A^T + I_p)^{-1} ~+~ R\color{black}{\cdot\Sigma_X^{-1/2}}. \label{eq.misspecification}
\end{align}
It is trivial to note that 
$$\big|\big|F - \Theta\big|\big|_F^2 ~\lesssim~ \frac{M}{p} ~+~ \frac{M}{p^2}$$
since we get the rates for $||F_m - \Theta_m||_2^2$ uniformly over $1 \leq m \leq M$. To get the last result in the theorem, we multiply both sides of (\ref{eq.misspecification}) by $P_B^{\perp}$. This eliminates the first term as the $P_B^{\perp} B = 0$. Thus, we are left with 
\begin{align*}
    P_B^{\perp}F - P_B^{\perp}\Theta ~&=~ P_B^{\perp}F - \Theta ~=~ P_B^{\perp} R\cdot\Sigma_X^{-1/2},\\
     \big|\big|P_B^{\perp}F - \Theta \big|\big|_F^2 ~&\leq~\big|\big|P_B^{\perp}\big|\big|_{\text{op}}^2\cdot \big|\big|R\Sigma_X^{-1/2}\big|\big|_F^2  \\
     ~&=~ \big|\big|R\Sigma_X^{-1/2}\big|\big|_F^2\\
     ~&\lesssim~ \frac{M}{p^2},
\end{align*}
where the last line follows from the derivation in (\ref{eq.rate}). This concludes the proof.
\end{proof}

\subsection{Proof of Theorem \ref{thm_hatF}}
Since sample splitting does not change the proof, for simplicity we omit sampling splitting in the proof and define the estimator $\hat F_m$ as the local maximizer of $Q_m(F)$ using all $n$ data points. Let $L_m(F)=-Q_m(F)$. Our goal is to show that there exists a local minimizer $\hat \Delta_m$ of $L_m(F_m^*+\Delta)$ such that $\hat \Delta_m\in \mathcal{C}$ for all $1\leq m\leq M$, where $\mathcal{C}=\{\Delta\in\RR^p: \|\Delta\|_2\leq r\}$ and $r=C\sqrt{p\log (M\vee p)/n}$ for some constant $C$ large enough. To this end, it suffices to show that the event
\begin{align*}
\cap_{1\leq m\leq M} \Big\{\inf_{\Delta\in \partial\mathcal{C}}L_m(F_m^*+\Delta)-L_m(F_m^*)>0\Big\}
\end{align*}
holds with probability tending to 1, where  $\partial\mathcal{C}=\{\Delta\in\RR^p: \|\Delta\|_2= r\}$. Applying the mean value theorem, we have for any  $\Delta\in \partial\mathcal{C}$, 
\begin{align*}
L_m(F_m^*+\Delta)-L_m(F_m^*)&=\nabla L_m(F_m^*)\Delta+\frac{1}{2}\Delta^T\nabla^2 L_m(F_m^*+t\Delta)\Delta\\
&\geq -\|\nabla L_m(F_m^*)\|_2 r+\frac{1}{2}r^2\lambda_{\min}(\nabla^2 L_m(F_m^*+t\Delta))\\
&\geq -C'\sqrt{\frac{p\log (M\vee p)}{n}}\cdot C\sqrt{\frac{p\log (M\vee p)}{n}}+\frac{1}{2}C^2\frac{p\log (M\vee p)}{n}\cdot C''
\end{align*}
under the following two events, 
$$
E_1=\Big\{\max_{1\leq m\leq M}\|\nabla L_m(F_m^*)\|_2\leq C'\sqrt{\frac{p\log (M\vee p)}{n}} \Big\},~~E_2=\Big\{\min_{1\leq m\leq M}\lambda_{\min}(\nabla^2 L_m(F_m^*+t\Delta))\geq C''\Big\},
$$
where $C'$ and $C''$ are constants. Provided $C>2C'/C''$, we obtain that $L_m(F_m^*+\Delta)-L_m(F_m^*)>0$. In the following, we will show that $\PP(E_1)\rightarrow 1$ and $\PP(E_2)\rightarrow 1$. Recall that 
$$
\nabla L_m(F_m^*)=-\frac{1}{n}\sum_{i=1}^n \bigg\{\frac{Y^{(i)}_m - b'(F_m^* X^{(i)})}{b''(F_m^* X^{(i)})}\bigg\}X^{(i)}.
$$
By definition we notice that $\nabla L_m(F_m^*)$ has mean zero. Since $b''(F_m^* X^{(i)})\geq C_1$, all entries of $X^{(i)}$ are bounded, and $\|Y_m - b'(F_m^* X)\|_{\psi_1}\leq \|Y_m- b'(\Theta_mX+B_m Z)\|_{\psi_1}+ \|b'(\Theta_mX+B_m Z)- b'(F_m^* X)\|_{\psi_1}$ is bounded, Bernstein inequality implies that $\PP\Big(|(\nabla L_m(F_m^*))_j|\geq C'\sqrt{\frac{\log (M\vee p)}{n}}\Big)\leq (p\vee M)^{-3}$ for some constant $C'$. Applying the union bound, we can show that 
$$
\max_{1\leq m\leq M}\|\nabla L_m(F_m^*)\|_\infty\leq C'\sqrt{\frac{\log (M\vee p)}{n}}
$$
with high probability, which further implies $\PP(E_1)\rightarrow 1$. 

For the event $E_2$, let us denote $\tilde F_m=F_m^*+t\Delta$ and $\zeta^{(i)}_m(\tilde F_m)=\frac{(Y^{(i)}_m - b'(\tilde F_m X^{(i)}))b'''(\tilde F_m X^{(i)})}{\{b''(\tilde F_m X^{(i)})\}^2}$. We first normalize $\nabla^2 L_m(\tilde F_m)$ since $X$ has spiked eigenvalues under our factor model assumption. Specifically, 
\begin{align*}
\lambda_{\min}(\nabla^2 L_m(\tilde F_m))&=\lambda_{\min}(\Sigma_X^{-1/2}\Sigma_X^{1/2}\nabla^2 L_m(\tilde F_m)\Sigma_X^{1/2}\Sigma_X^{-1/2})\\
&\geq \lambda_{\min}(\Sigma_X^{-1/2}\nabla^2 L_m(\tilde F_m)\Sigma_X^{-1/2}),
\end{align*}
since the smallest eigenvalue of $\Sigma_X=A\Sigma_ZA^T+I$ is no smaller than 1. For notational simplicity, we set $\nabla^2 \tilde L_m(\tilde F_m)=\Sigma_X^{-1/2}\nabla^2 L_m(\tilde F_m)\Sigma_X^{-1/2}$. Then we have
\begin{align*}
\nabla^2 \tilde L_m(\tilde F_m)&=\frac{1}{n}\sum_{i=1}^n \Big(1+\zeta^{(i)}_m(\tilde F_m)\Big)\tilde X^{(i)}\tilde X^{(i)T}\\
&=\frac{1}{n}\sum_{i=1}^n \Big(1+\zeta^{(i)}_m( F^*_m)\Big)\tilde X^{(i)}\tilde X^{(i)T}+\frac{1}{n}\sum_{i=1}^n \Big(\zeta^{(i)}_m(\tilde F_m)-\zeta^{(i)}_m(F^*_m)\Big)\tilde X^{(i)}\tilde X^{(i)T}\\
&=\EE(1+\zeta^{(i)}_m(F^*_m))\tilde X^{(i)}\tilde X^{(i)T}+I_1+I_2,
\end{align*}
where $\tilde X^{(i)}=\Sigma_X^{-1/2}X^{(i)}$, and
\begin{align*}
I_1~&=~\frac{1}{n}\sum_{i=1}^n \Big(1+\zeta^{(i)}_m( F^*_m)\Big)\tilde X^{(i)} \tilde X^{(i)T}-\EE(1+\zeta^{(i)}_m( F^*_m))\tilde X^{(i)}\tilde X^{(i)T}\\
I_2~&=~\frac{1}{n}\sum_{i=1}^n \Big(\zeta^{(i)}_m(\tilde F_m)-\zeta^{(i)}_m(F^*_m)\Big)\tilde X^{(i)}\tilde X^{(i)T}.
\end{align*}
By Lemma \ref{lem_tildeZ}, $\tilde X^{(i)}$ is a sub-Gaussian vector, so $\tilde X_j^{(i)} \tilde X_k^{(i)}$ is sub-exponential. Therefore, due to the boundedness of $|b''|,~|b'''|$ and the fact that $Y_m^{(i)}-b'(F_m^* X^{(i)})$ is sub-exponential from Assumption \ref{assum_bound} and \ref{assum_glm}, $(1+\zeta^{(i)}_m( F^*_m))\tilde X_j^{(i)} \tilde X_k^{(i)}$  is $1/2$-sub-exponential. Lemma \ref{GenBern} implies that, 
$$
\max_{1\leq m\leq M}\|I_1\|_{\max} \lesssim \sqrt{\frac{\log (p\vee M)}{n}},
$$
provided $\{\log (M\vee p)\}^3=O(n)$, which implies 
$$
\max_{1\leq m\leq M}\|I_1\|_F\lesssim p\sqrt{\frac{\log (p\vee M)}{n}}
$$
with high probability. Furthermore, writing out $\zeta_m^{(i)}$ and $\zeta_m^{(i)'}$, under Assumption \ref{assum_glm}, it is apparent that $\Big|\zeta_m^{(i)'}\Big|$ is bounded by constants and a random term, $F_mX^{(i)}$. Thus, $\zeta^{(i)}_m(F_m)$ is Lipschitz in terms of $F_mX^{(i)}$. Using this fact, the Mean Value Theorem, and the rate for the maximum of a sub-exponential random variable, we get
$$
\max_{1\leq i\leq n}\max_{1\leq m\leq M}|\zeta^{(i)}_m(\tilde F_m)-\zeta^{(i)}_m(F^*_m)|\lesssim  \|\Delta\|_1 \log (M\vee n), 
$$
This then implies that 
\begin{align*}
\max_{1\leq m\leq M}\|I_2\|_{\text{op}}&\lesssim \|\Delta\|_1 \log (M\vee n)\cdot \Big\|\frac{1}{n}\sum_{i=1}^n\tilde X^{(i)}\tilde X^{(i)T}\Big\|_{\text{op}}\\
&\lesssim p\sqrt{\frac{\log (p\vee M)}{n}} \log (M\vee n) \cdot \bigg(\|I_p\|_{\text{op}}+\sqrt{\frac{p}{n}}\bigg)\\
&\lesssim p\sqrt{\frac{\log (p\vee M)}{n}} \log (M\vee n), 
\end{align*}
where the second lines follows from plugging in our choice of rate for $||\Delta||_2 = r=\sqrt{p\log(p \vee M)/n}$ and the moment bound $\EE\|\frac{1}{n}\sum_{i=1}^n\tilde X^{(i)}\tilde X^{(i)T}-I_p\|_{\text{op}}\lesssim \|I_p\|_{\text{op}}\sqrt{p/n}$. For details of the latter, see the proof of Lemma \ref{thm_third_term}. Combining these results, Weyl's inequality implies
$$
\Big|\lambda_{\min}\big(\nabla^2 \tilde L_m(\tilde F_m)\big)-\lambda_{\min}\big(\EE(1+\zeta^{(i)}_m(F^*_m))\tilde X^{(i)} \tilde X^{(i)T}\big)\Big|\leq \|I_1\|_{\text{op}}+\|I_2\|_{\text{op}}\lesssim p\sqrt{\frac{\log (p\vee M)}{n}}\log (M\vee n),
$$
uniformly over $m$. Therefore, 
\begin{align}\label{eq_lower_eigen}
\lambda_{\min}\big(\nabla^2\tilde L_m(\tilde F_m)\big)\geq \lambda_{\min}\big(\EE(1+\zeta^{(i)}_m(F^*_m))\tilde X^{(i)}\tilde X^{(i)T}\big)-p\sqrt{\frac{\log (p\vee M)}{n}}\log (M\vee n).
\end{align}
Finally, we focus on 
\begin{align*}
\EE(1+\zeta^{(i)}_m(F^*_m))&\tilde X^{(i)}\tilde X^{(i)T}\\
&=\EE\bigg(1+\frac{(b'(\Theta_mX^{(i)}+B_m Z^{(i)}) - b'( F^*_m X^{(i)}))b'''( F^*_m X^{(i)})}{\{b''( F^*_m X^{(i)})\}^2}\bigg)\tilde X^{(i)}\tilde X^{(i)T}\\
&=\EE\bigg(1+\frac{(b''(\xi)(\Theta_mX^{(i)}+B_m Z^{(i)}- F^*_m X^{(i)})b'''( F^*_m X^{(i)})}{\{b''( F^*_m X^{(i)})\}^2}\bigg)\tilde X^{(i)}\tilde X^{(i)T},
\end{align*}
where $\xi$ is an intermediate value between $\Theta_mX^{(i)}+B_m Z^{(i)}$ and $F^*_m X^{(i)}$. 
Notice that
\begin{equation}\label{eq_upper_taylor}
    \Bigg|\frac{(b''(\xi)(\Theta_mX^{(i)}+B_m Z^{(i)}- F^*_m X^{(i)})b'''( F^*_m X^{(i)})}{\{b''( F^*_m X^{(i)})\}^2}\Bigg|\leq C|\Theta_mX^{(i)}+B_m Z^{(i)}- F^*_m X^{(i)}|,
\end{equation}
for some constant $C$. To lower bound $\lambda_{\min}(\EE(1+\zeta^{(i)}_m(F^*_m))\tilde X^{(i)}\tilde X^{(i)T})$, we apply the following truncation technique. By Weyl's inequality again, we have 
\begin{align}
\lambda_{\min}\Big(\EE(1+\zeta^{(i)}_m(F^*_m))\tilde X^{(i)}\tilde X^{(i)T}\Big)&\geq \lambda_{\min}\Big(\EE(1+\zeta^{(i)}_m(F^*_m)I(|\Psi_m|\leq \frac{1}{2C}))\tilde X^{(i)}\tilde X^{(i)T}\Big)\nonumber\\
&~~~~~~~-\Big\|\EE\Big(\zeta^{(i)}_m(F^*_m)I(|\Psi_m|\geq \frac{1}{2C})\tilde X^{(i)}\tilde X^{(i)T}\Big)\Big\|_F,\label{eq_lower_eigen2}
\end{align}
where $\Psi_m=\Theta_mX^{(i)}+B_m Z^{(i)}- F^*_m X^{(i)}$. Thus, 
$$
\lambda_{\min}\Big(\EE(1+\zeta^{(i)}_m(F^*_m)I(|\Psi_m|\leq \frac{1}{2C}))\tilde X^{(i)}\tilde X^{(i)T}\Big)\geq \lambda_{\min}\Big(\EE(1-\frac{1}{2})\tilde X^{(i)}\tilde X^{(i)T}\Big)= 1/2.$$
For the second term on the right hand side of (\ref{eq_lower_eigen2}), we have
\begin{align*}
\Big\|\EE\Big(\zeta^{(i)}_m(F^*_m)I(|\Psi_m|\geq \frac{1}{2C})\tilde X^{(i)}\tilde X^{(i)T}\Big)\Big\|^2_F&\lesssim \EE\Big(\Psi_m^2I(|\Psi_m|\geq \frac{1}{2C})\sum_{1\leq j, k\leq p}(\tilde X^{(i)}_j\tilde X^{(i)}_k)^2\Big)\\
&\leq \Big\{\EE\Psi_m^4 I(|\Psi_m|\geq \frac{1}{2C})\Big\}^{1/2} \Big\{\EE [\sum_{1\leq j, k\leq p}(\tilde X^{(i)}_j\tilde X^{(i)}_k)^2]^2\Big\}^{1/2} \\
&\leq p^2\Big\{\EE|\Psi_m|^8\Big\}^{1/4}\cdot \Big\{\PP(|\Psi_m|\geq \frac{1}{2C})\Big\}^{1/4}\\
&\lesssim p^2\cdot p^{-1} \exp(-p/C)
\end{align*}
where the second line follows from the Cauchy-Schwarz inequality, the third line is from the sub-Gaussian property of $\tilde X_j^{(i)}$ and the Cauchy-Schwarz inequality again, and the last step is due to $\|\Psi_m\|_{\psi_2}\lesssim p^{-1/2}$ implied by the proof of Theorem \ref{thm_expectation}. Plugging these results into (\ref{eq_lower_eigen2}), we have
\begin{align}\label{eq_min_eigen_EH}
\lambda_{\min}\Big(\EE(1+\zeta^{(i)}_m(F^*_m))\tilde X^{(i)}\tilde X^{(i)T}\Big)\geq C,
\end{align}
which implies $\PP(E_2)\rightarrow 1$ in view of (\ref{eq_lower_eigen}). As a by-product, we have
\begin{align}\label{eq_min_eigen_G}
\lambda_{\min}(G_m)\geq \lambda_{\min}\Big(\EE(1+\zeta^{(i)}_m(F^*_m))\tilde X^{(i)}\tilde X^{(i)T}\Big)\geq C.
\end{align}

The rest of the proof is to show (\ref{eq_thm_hatF_linear}). Recall that $\nabla L_m(\hat F_m)=0$. This combined with the mean value theorem gives us that 
$$
(\hat F_m-F_m^*)^T=\bigg\{\frac{1}{n}\sum_{i=1}^n (1+\tilde\zeta^{(i)}_m)X^{(i)}X^{(i)T}\bigg\}^{-1}\frac{1}{n}\sum_{i=1}^n \frac{Y^{(i)}_m - b'(F_m^* X^{(i)})}{b''(F_m^* X^{(i)})} X^{(i)},
$$
where $\tilde\zeta^{(i)}_m=\zeta^{(i)}_m(\tilde F_m)$. Then for any $v\in\RR^p$,
\begin{align}
&(\hat F_m-F_m^*)v\nonumber\\
&~~=\frac{1}{n}\sum_{i=1}^n \frac{Y^{(i)}_m - b'(F_m^* X^{(i)})}{b''(F_m^* X^{(i)})} v^TG_m^{-1}X^{(i)}\nonumber\\
&~~~~~~+v^T\bigg[\Big(\frac{1}{n}\sum_{i=1}^n (1+\zeta^{(i)}_m)X^{(i)}X^{(i)T}\Big)^{-1}-G_m^{-1}\bigg]\frac{1}{n}\sum_{i=1}^n \frac{Y^{(i)}_m - b'(F_m^* X^{(i)})}{b''(F_m^* X^{(i)})} X^{(i)}\label{eq_linear_1}\\
&~~~~~~+v^T\bigg[\Big(\frac{1}{n}\sum_{i=1}^n (1+\tilde\zeta^{(i)}_m)X^{(i)}X^{(i)T}\Big)^{-1}-\Big(\frac{1}{n}\sum_{i=1}^n (1+\zeta^{(i)}_m)X^{(i)}X^{(i)T}\Big)^{-1}\bigg]\frac{1}{n}\sum_{i=1}^n \frac{Y^{(i)}_m - b'(F_m^* X^{(i)})}{b''(F_m^* X^{(i)})} X^{(i)},\label{eq_linear_2}
\end{align}
where $\zeta^{(i)}_m=\zeta^{(i)}_m(F^*_m)$. So it remains to control the two terms in (\ref{eq_linear_1}) and (\ref{eq_linear_2}) respectively. From the analysis of the term $I_1$ and event $E_1$ defined previously, using the identity $A^{-1} - B^{-1} = -A^{-1}(A-B)B^{-1}$, the term (\ref{eq_linear_1}) is upper bounded by
\begin{align*}
&\max_{1\leq m\leq M}\Bigg\|\Big(\frac{1}{n}\sum_{i=1}^n (1+\zeta^{(i)}_m)X^{(i)}X^{(i)T}\Big)^{-1}-G_m^{-1}\Bigg\|_{\text{op}}\Bigg\|\frac{1}{n}\sum_{i=1}^n \frac{Y^{(i)}_m - b'(F_m^* X^{(i)})}{b''(F_m^* X^{(i)})} X^{(i)}\Bigg\|_2\\
&\lesssim \max_{1\leq m\leq M}\Bigg\|\Big(\frac{1}{n}\sum_{i=1}^n (1+\zeta^{(i)}_m)X^{(i)}X^{(i)T}\Big)^{-1}\Bigg\|_{\text{op}}\Bigg\|\frac{1}{n}\sum_{i=1}^n (1+\zeta^{(i)}_m)X^{(i)}X^{(i)T}-G_m\Bigg\|_{\text{op}} \|G_m^{-1}\|_{\text{op}}\sqrt{\frac{p\log (M\vee p)}{n}}\\
&\lesssim \max_{1\leq m\leq M}\Bigg\|\frac{1}{n}\sum_{i=1}^n (1+\zeta^{(i)}_m)X^{(i)}X^{(i)T}-G_m\Bigg\|_{F} \sqrt{\frac{p\log (M\vee p)}{n}}\\
&\lesssim~ \frac{p^{3/2}\log (M\vee p)}{n}
\end{align*}
where we use (\ref{eq_min_eigen_G}), combined with the concentration bound for the term $I_1$ above and Weyl's inequality to show that the smallest eigenvalue of $\frac{1}{n}\sum_{i=1}^n (1+\zeta^{(i)}_m)X^{(i)}X^{(i)T}$ is also lower bounded by a constant. Following a similar argument, we can rewrite (\ref{eq_linear_2}) as
\begin{align*}
&\Big\|v^T\Sigma_X^{-1/2}(\tilde H^{-1}-\hat H^{-1})\Sigma_X^{-1/2}\frac{1}{n}\sum_{i=1}^n \frac{Y^{(i)}_m - b'(F_m^* X^{(i)})}{b''(F_m^* X^{(i)})} X^{(i)}\Big\|_2\\
&~~~~~~~~~~~~~~~~=\Big\|v^T\Sigma_X^{-1/2}\tilde H^{-1}(\hat H-\tilde H)\hat H^{-1}\Sigma_X^{-1/2}\frac{1}{n}\sum_{i=1}^n \frac{Y^{(i)}_m - b'(F_m^* X^{(i)})}{b''(F_m^* X^{(i)})} X^{(i)}\Big\|_2\\
&~~~~~~~~~~~~~~~~\lesssim \|\tilde H^{-1}\|_{\text{op}} \|\hat H-\tilde H\|_{\text{op}}\|\hat H^{-1}\|_{\text{op}}\sqrt{\frac{p\log (M\vee p)}{n}}
\end{align*}
where $\tilde H=\frac{1}{n}\sum_{i=1}^n (1+\tilde\zeta^{(i)}_m)\tilde X^{(i)}\tilde X^{(i)T}$
and $\hat H=\frac{1}{n}\sum_{i=1}^n (1+\zeta^{(i)}_m)\tilde X^{(i)}\tilde X^{(i)T}$. Recall that from the analysis of the term $I_2$ above we have 
$$
\|\hat H-\tilde H\|_{\text{op}}\lesssim p\sqrt{\frac{\log (p\vee M)}{n}} \log (M\vee n)=o_p(1),
$$
and from the analysis of the term $I_1$
, we have
$$
\|\hat H-\EE \hat H\|_{\text{op}}\leq \|\hat H-\EE \hat H\|_{F}\lesssim p\sqrt{\frac{\log (p\vee M)}{n}}=o_p(1),
$$
where the minimum eigenvalue of $\EE \hat H$ is lower bounded by a constant as shown in  (\ref{eq_min_eigen_EH}). Thus, (\ref{eq_linear_2}) is upper bounded by $O_p\big(p^{3/2}\frac{\log (p\vee M)\log (n\vee M)}{n}\big)$. This completes the proof of (\ref{eq_thm_hatF_linear}).

\subsection{Proof for Theorem \ref{thm_projection}}
\begin{proof}
Recall that the three residuals are
$$\hat{\epsilon}_m = \frac{Y_m - b'(\hat{F}_mX)}{b''(\hat{F}_mX)},~~\bar{\epsilon}_m = \frac{Y_m - b'(F_m^* X)}{b''(F_m^* X)},~~\epsilon_m = \frac{Y_m-b'(\Theta_m X + B_mZ)}{b''(\Theta_m X + B_mZ)}.$$
We use the upper-script $(i)$ to indicate the r.v from the $i$th sample. For simplicity, we use $\EE_n \epsilon$ to denote $\frac{1}{n}\sum_{i=1}^n \epsilon^{(i)}$ and $\EE_n (\epsilon|X)$ to denote $\frac{1}{n}\sum_{i=1}^n \EE(\epsilon^{(i)}|X^{(i)})$. Finally, due to sample splitting, for simplicity we can just equivalently assume that $\hat{F}_m$ is independent of $Y^{(i)}$ and $X^{(i)}$. 

We also have $\hat P^{\perp}_B=I-\hat V\hat V^T$, $P^{\perp}_B=I-VV^T$, where $\hat{V}$ is the first $K$ eigenvectors of $\hat \Sigma = \frac{1}{n} \hat\epsilon^{(i)}\hat\epsilon^{(i)T}$ and $V$ is the first $K$ left singular vectors of $B$ and is also the first $K$ eigenvectors of $\EE[B(Z-\tilde Z)(Z- \tilde Z)^TB^T] = B(\Sigma_Z^{-1}+A^TA)^{-1}B^T$. We thus recall, define, and denote $\Sigma=B(\Sigma_Z^{-1}+A^TA)^{-1}B^T$ and $\hat\Sigma=\EE_n \hat\epsilon^{\otimes 2}$ accordingly.

By the Davis-Kahan theorem from Lemma \ref{DavisKahan}, we get the following inequality:
$$
\|\hat V O-V\|_F\lesssim \frac{\|\hat\Sigma-\Sigma\|_F}{\lambda_K(\Sigma)-\lambda_{K+1}(\Sigma)},
$$
where $O$ is some orthogonal matrix. It is easily seen that
\begin{align*}
    \|\hat{P}_{B}^{\perp}-P_{B}^{\perp}\|_{F} &~=~ \|\hat{V}\hat{V}^T - VV^T\|_F\\
    &~=~ \|\hat{V}\hat{V}^T - V{O}^T\hat{V}^T + V{O}^T\hat{V}^T - VV^T\|_F \\
    & ~=~ \|(\hat{V}{O}-V){O}^T\hat{V}^T + V(O^T\hat V^T-V^T)\|_F\\
    &~\leq ~2\|\hat{V}O-V\|_F.
\end{align*}
Since $\lambda_{K+1}(\Sigma)=0$ and $\lambda_K(\Sigma)\geq CM/p$ by Lemma \ref{thm_first_term}, we obtain that
\begin{equation}\label{eq_thm_projection_pf_1}
\|\hat{P}_{B}^{\perp}-P_{B}^{\perp}\|_{F}\lesssim \frac{p}{M}\|\hat\Sigma-\Sigma\|_F.
\end{equation}
It remains to bound $\|\hat\Sigma-\Sigma\|_F$. Note that we can decompose $\|\hat\Sigma-\Sigma\|_F$ as follows:
\begin{equation}\label{eq_thm_projection_pf_2}
\|\hat\Sigma-\Sigma\|_F\leq I_1+I_2+I_3+I_4,
\end{equation}
where 
$$
I_1=\|(\EE_n-\EE)\epsilon^{\otimes 2}\|_F,~~~I_2=\|\EE\epsilon^{\otimes 2}\|_F, ~~~I_3=\|(\EE_n-\EE)\{B(Z-\tilde Z)\}^{\otimes 2}\|_F,
$$
$$
I_4=\|\EE_n[\hat \epsilon^{\otimes 2}-\epsilon^{\otimes 2}-\{B(Z-\tilde Z)\}^{\otimes 2}]\|_F,
$$
and $\tilde Z$ is defined in Lemma \ref{lem_tildeZ}. Lemmas \ref{thm_sec_term} and \ref{thm_third_term} imply
$$
I_1=O_p\Bigg(M\sqrt{\frac{\log M}{n}}\Bigg), ~~I_3=O_p\Bigg(\frac{M}{p}\sqrt{\frac{K}{n}}\Bigg).
$$
Since $Y_m$ and $Y_{m'}$ are independent given $X$ and $Z$, $\EE\epsilon^{\otimes 2}$ is a diagonal matrix. Combined with Assumption \ref{assum_glm}, we have
$$
I_2=O_p\big(\sqrt{M}\big).
$$
The rest of the proof is to bound the last term $I_4$. By writing 
$$
\hat\epsilon_j= \epsilon_j+(\bar\epsilon_j-\epsilon_j)+(\hat\epsilon_j-\bar\epsilon_j),
$$
we obtain via Taylor expansion that
\begin{align*}
\bar\epsilon_j-\epsilon_j&=-(1+\zeta_j)(F_j^*X-\Theta_jX-B_jZ)+\tilde \eta_j (F_j^*X-\Theta_jX-B_jZ)^2\\
&=(1+\zeta_j)[B_j(Z-\tilde Z)-\bar\phi_j\Sigma^{-1/2}X]+\tilde \eta_j (F_j^*X-\Theta_jX-B_jZ)^2\\
&=\underbrace{B_j(Z-\tilde Z)}_{J_{2j}}+\underbrace{\zeta_jB_j(Z-\tilde Z)}_{J_{3j}}-\underbrace{(1+\zeta_j)\bar\phi_j\Sigma^{-1/2}X}_{J_{4j}}+\underbrace{\tilde \eta_j (F_j^*X-\Theta_jX-B_jZ)^2}_{J_{5j}},
\end{align*}
where 
$$
\zeta_j=\frac{(Y_j - b'(\Theta_j X + B_j Z))b'''(\Theta_j X + B_j Z)}{\{b''(\Theta_j X + B_j Z\}^2},
$$
\[
\phi_j=\frac{b'''(\delta_j)}{2b''(F^*_j X)}\cdot(\Theta_j X + B_j Z - F^*_j X)^2\cdot (\Sigma_X^{-1/2}X)^T
\]
is defined in the proof of Theorem \ref{thm_expectation} with $\bar\phi_j=\EE(\phi_j)$, and
$$
\tilde \eta_j=-\frac{[-b''(t_j)b'''(t_j)+(Y_j-b'(t_j))b''''(t_j)]\{b''(t_j)\}^2-2b''(t_j)b'''(t_j)(Y_j-b'(t_j))b'''(t_j)}{\{b''(t_j\}^4}
$$
where $t_j$ is some intermediate value between $F_j^*X$ and $\Theta_j X + B_j Z$. In addition, we have
$$
\hat\epsilon_j-\bar\epsilon_j=-\underbrace{(1+\tilde\zeta_j)(\hat F_j-F_j^*)X}_{J_{6j}},
$$
where  $\tilde \zeta_j$ is defined in the same way as $\zeta_j$ with $F_j^*X$ replaced by some intermediate value  between $F_j^*X$ and $\hat F_j X$. Summarizing all the terms above, for a fixed observation $(i)$, looking at the column vector constructed by combining all $1 \leq j \leq M$, we have $\hat\epsilon= \sum_{s=1}^6 J_s$, where $J_1=\epsilon$ and $J_s$ is the column vector consisting of $J_{sj}$ for $s\geq 2$. Recall that we have $I_4 = \big|\big|\EE_n\big[(J_1+...+J_6)^{\otimes 2} \big] - \EE_n\big[J_1^{\otimes 2}\big] - \EE_n\big[J_2^{\otimes 2}\big]\big|\big|_F$. Thus, we have 
\begin{align}\label{eq_thm_projection_pf_I4}
I_4\leq \sum_{s=3}^6 \|\EE_n J_s^{\otimes 2}\|_F+\sum_{1\leq s\neq t\leq 6} \|\EE_n J_s J_t^T\|_F.
\end{align}
For each term on the right hand side, we consider the diagonal and off-diagonal terms separately. Using the superscript $(i)$ to explicitly denote the $i$-th observation, note that each diagonal term in $\EE_n J_sJ_t^T$ will have form $\big|\sum_{i=1}^n J_{sj}^{(i)}J_{tj}^{(i)}\big|$ and each off-diagonal term will have form $\big|\sum_{i=1}^n J_{sj}^{(i)}J_{tk}^{(i)}\big|$ where $1 \leq j \neq k \leq M$. Similarly, each diagonal term in $\EE_n J_s^{
\otimes 2
}$ will have form $\big|\sum_{i=1}^n J_{sj}^{(i)2}\big|$ while each off diagonal term will have form $\big|\sum_{i=1}^n J_{sj}^{(i)}J_{sk}^{(i)}\big|$ where $1 \leq j \neq k \leq M$. We start from the off-diagonal terms. For $j\neq k$ and $s=3$, we have $\EE_n J_{sj}J_{sk}=\EE_n \zeta_j\zeta_kB_j(Z-\tilde Z)^{\otimes 2}B_k^T$. Since $\zeta_j\zeta_k$ has mean 0 conditioned on $X,Z$ and $\zeta_j\zeta_k$ is $1/2$-sub-exponential, Lemma \ref{GenBern} implies that conditioned on $X, Z$, 
$$
\max_{j\neq k}|\EE_n J_{sj}J_{sk}|\lesssim \sqrt{\frac{\log M}{n}} \max_{1\leq i\leq n}\max_{j\neq k}|B_j(Z^{(i)}-\tilde Z^{(i)})^{\otimes 2}B_k^T|,
$$
provided $(\log M)^3/n=O(1)$. By Lemma \ref{lem_tildeZ}, we know that $\|B_m(Z-\tilde{Z})\|_{\psi_2} \leq c/\sqrt{p}$, which implies $B_j(Z^{(i)}-\tilde Z^{(i)})^{\otimes 2}B_k^T$ is sub-exponential with norm of order $1/p$. By the tail bound for the maximum of sub-exponential r.v, we can show that 
\begin{align}\label{eq_thm_projection_pf_I4_1}
\max_{j\neq k}|\EE_n J_{3j}J_{3k}|\lesssim \sqrt{\frac{\log M}{n}} \frac{\log (n\vee M)}{p}.
\end{align}
For $s=4$, 
$$
\max_{j\neq k}|\EE_n J_{sj}J_{sk}|\leq \max_{1\leq j\leq M} \Big\{\EE_n (1+\zeta_j)^2(\bar\phi_j\Sigma^{-1/2}X)^2\Big\}^{1/2}\max_{1\leq k\leq M} \Big\{\EE_n (1+\zeta_k)^2(\bar\phi_k\Sigma^{-1/2}X)^2\Big\}^{1/2}.
$$
We use the same logic to bound $A_{j}:=(1+\zeta_j)^2(\bar\phi_j\Sigma^{-1/2}X)^2$. Again,  $(1+\zeta_j)^2$ is $1/2$-sub-exponential and Lemma \ref{GenBern} implies that conditioned on $X, Z$,
$$
\max_{1\leq j\leq M}|\EE_n A_j-\EE_n(A_j|X,Z)|\lesssim \sqrt{\frac{\log M}{n}} \max_{1\leq i\leq n}\max_{j}(\bar\phi_j\Sigma^{-1/2}X^{(i)})^2\lesssim \sqrt{\frac{\log M}{n}} \frac{\log (n\vee M)}{p^2},
$$
provided $(\log M)^3/n=O(1)$, where we use the fact that $\|\bar\phi_j\Sigma^{-1/2}X^{(i)}\|_{\psi_2}\lesssim 1/p$. In addition,   
\begin{align*}
\max_{1\leq j\leq M}|\EE_n (A_j|X,Z)|&\lesssim \max_{1\leq j\leq M}\EE_n (\bar\phi_j\Sigma^{-1/2}X)^2\\
&= \max_{1\leq j\leq M}\Big[(\EE_n-\EE) (\bar\phi_j\Sigma^{-1/2}X)^2+\EE (\bar\phi_j\Sigma^{-1/2}X)^2\Big]. 
\end{align*}
By the sub-Gaussian property, $\EE (\bar\phi_j\Sigma^{-1/2}X)^2\lesssim 1/p^2$. The Bernstein inequality implies 
$$
\max_{1\leq j\leq M}|(\EE_n-\EE) (\bar\phi_j\Sigma^{-1/2}X)^2|\lesssim \sqrt{\frac{\log M}{n}} \frac{1}{p^2}.
$$
Therefore, 
$$
\max_{1\leq j\leq M}|\EE_n (A_j|X,Z)|\lesssim \frac{1}{p^2},
$$
which implies 
$$
\max_{1\leq j\leq M}|\EE_n A_j|\lesssim \sqrt{\frac{\log M}{n}} \frac{\log (n\vee M)}{p^2}+\frac{1}{p^2}\lesssim \frac{1}{p^2}.
$$
Finally, we obtain
\begin{align}\label{eq_thm_projection_pf_I4_2}
\max_{j\neq k}|\EE_n J_{4j}J_{4k}|\lesssim \frac{1}{p^2}. 
\end{align}
For $s=5$, we have 
$$
\max_{j\neq k}|\EE_n J_{sj}J_{sk}|\leq \max_{1\leq j\leq M}\{\EE_n A_{5j}\}^{1/2}\max_{1\leq k\leq M}\{\EE_n A_{5k}\}^{1/2},
$$ 
where $A_{5j}=\tilde \eta^2_j (F_j^*X-\Theta_jX-B_jZ)^4$. Following a similar argument, 
\begin{align*}
\max_{1\leq j\leq M}|\EE_n A_{5j}-\EE_n(A_{5j}|X,Z)|&\lesssim \sqrt{\frac{\log M}{n}} \max_{1\leq i\leq n}\max_{j}(F_j^*X^{(i)}-\Theta_jX^{(i)}-B_jZ^{(i)})^4\\
&\lesssim \sqrt{\frac{\log M}{n}} \frac{(\log (n\vee M))^2}{p^2},
\end{align*}
where we know from the proof of Theorem \ref{thm_expectation} that $\|F_j^*X-\Theta_jX-B_jZ\|_{\psi_2}\lesssim p^{-1/2}$. We can similarly show that
$$
\max_{1\leq j\leq M}|\EE_n (A_{5j}|X,Z)|\lesssim \sqrt{\frac{\log M}{n}} \frac{1}{p^2}+\frac{1}{p^2}\lesssim \frac{1}{p^2},
$$
and therefore under the assumption $(\log M)^5/n=O(1)$, 
$$
\max_{1\leq j\leq M}|\EE_n A_{5j}|\lesssim \sqrt{\frac{\log M}{n}} \frac{(\log (n\vee M))^2}{p^2}+\frac{1}{p^2}\lesssim \frac{1}{p^2}.
$$
Finally, we obtain
\begin{align}\label{eq_thm_projection_pf_I4_3}
\max_{j\neq k}|\EE_n J_{5j}J_{5k}|\lesssim \frac{1}{p^2}. 
\end{align}
For $s=6$, 
$$
\max_{j\neq k}|\EE_n J_{sj}J_{sk}|\leq \max_{1\leq j\leq M}\{\EE_n A_{6j}\}^{1/2}\max_{1\leq k\leq M}\{\EE_n A_{6k}\}^{1/2},
$$ 
where $A_{6j}=(1+\tilde\zeta_j)^2\{(\hat F_j-F_j^*)X\}^2$. Due to sample splitting, given $X, Z$ and $\hat F$, the r.v. $(1+\tilde\zeta_j)^2$ is $1/2$-sub-exponential and Lemma \ref{GenBern} implies that 
\begin{align*}
\max_{1\leq j\leq M}|\EE_n A_{6j}-\EE_n(A_{6j}|X,Z,\hat F)|&\lesssim \sqrt{\frac{\log M}{n}} \max_{1\leq i\leq n}\max_{j}\{(\hat F_j-F_j^*)X^{(i)}\}^2\\
&\lesssim \sqrt{\frac{\log M}{n}} \max_{j}\|\hat F_j-F_j^*\|_1^2\\
&\lesssim \sqrt{\frac{\log M}{n}} \frac{p^2\log (p\vee M)}{n},
\end{align*}
where the last line follows from Theorem \ref{thm_hatF}. Together with 
$$
\EE_n(A_{6j}|X,Z,\hat F)\lesssim \EE_n \{(\hat F_j-F_j^*)X\}^2\lesssim  \frac{p\log (p\vee M)}{n},
$$
we obtain that 
$$
\max_{1\leq j\leq M}|\EE_n A_{6j}|\lesssim \sqrt{\frac{\log M}{n}} \frac{p^2\log (p\vee M)}{n}+\frac{p\log (p\vee M)}{n}\lesssim \frac{p\log (p\vee M)}{n},
$$
as we assume $p\sqrt{\frac{\log (p\vee M)}{n}}=o(1)$. As a result,
\begin{align}\label{eq_thm_projection_pf_I4_4}
\max_{j\neq k}|\EE_n J_{6j}J_{6k}|\lesssim \frac{p\log (p\vee M)}{n}. 
\end{align}
Next, consider $s=1$ and $t=2$ in (\ref{eq_thm_projection_pf_I4}). Again, given $X$ and $Z$, Lemma \ref{GenBern} (or the Bernstein inequality in Lemma \ref{Bern}) implies
$$
\max_{j\neq k}|\EE_n J_{sj}J_{tk}|=\max_{j\neq k}|\EE_n \epsilon_j B_k(Z-\tilde Z)|\lesssim \sqrt{\frac{\log M}{n}} \max_{1\leq i\leq n}\max_{1\leq k\leq M}|B_k(Z^{(i)}-\tilde Z^{(i)})|.
$$
Since $\|B_m(Z-\tilde{Z})\|_{\psi_2} \leq c/\sqrt{p}$, by the tail bound for the maximum of sub-Gaussian r.v, we can show that 
\begin{align}\label{eq_thm_projection_pf_I4_5}
\max_{j\neq k}|\EE_n J_{1j}J_{2k}|\lesssim \sqrt{\frac{\log M}{n}} \sqrt{\frac{\log (n\vee M)}{p}}.
\end{align}
For $s=1$ and $t=3$, 
\begin{align}
\max_{j\neq k}|\EE_n J_{1j}J_{3k}|&=\max_{j\neq k}|\EE_n \epsilon_j \zeta_k B_k(Z-\tilde Z)|\nonumber\\
&\lesssim \sqrt{\frac{\log M}{n}} \max_{1\leq i\leq n}\max_{1\leq k\leq M}|B_k(Z^{(i)}-\tilde Z^{(i)})|\nonumber\\
&\lesssim\sqrt{\frac{\log M}{n}} \sqrt{\frac{\log (n\vee M)}{p}}.\label{eq_thm_projection_pf_I4_6}
\end{align}
For $s=1$ and $t=4$, since $\epsilon_j (1+\zeta_k)$ is mean 0 and $1/2$-sub-exponential given $X$ and $Z$, invoking Lemma \ref{GenBern} we have
\begin{align}
\max_{j\neq k}|\EE_n J_{1j}J_{4k}|&=\max_{j\neq k}|\EE_n \epsilon_j (1+\zeta_k) \bar\phi_k\Sigma^{-1/2}X|\nonumber\\
&\lesssim \sqrt{\frac{\log M}{n}} \max_{1\leq i\leq n}\max_{1\leq k\leq M}|\bar\phi_k\Sigma^{-1/2}X^{(i)}|\nonumber\\
&\lesssim\sqrt{\frac{\log M}{n}} \frac{\sqrt{\log (n\vee M)}}{p}.\label{eq_thm_projection_pf_I4_7}
\end{align}
For $s=1$ and $t=5$,
\begin{align}
\max_{j\neq k}|\EE_n J_{1j}J_{5k}|&=\max_{j\neq k}|\EE_n \epsilon_j \tilde \eta_k (F_k^*X-\Theta_kX-B_kZ)^2|\nonumber\\
&\lesssim \sqrt{\frac{\log M}{n}} \max_{1\leq i\leq n}\max_{1\leq k\leq M}|(F_k^*X^{(i)}-\Theta_kX^{(i)}-B_kZ^{(i)})^2|\nonumber\\
&\lesssim\sqrt{\frac{\log M}{n}}\frac{\log (n\vee M)}{p}.\label{eq_thm_projection_pf_I4_8}
\end{align}
For $s=1$ and $t=6$, due to sample slitting, conditioned on $X,Z$ and $\hat F$, $\epsilon_j (1+\tilde\zeta_k)$ is mean 0 and $1/2$-sub-exponential. Therefore 
\begin{align}
\max_{j\neq k}|\EE_n J_{1j}J_{6k}|&=\max_{j\neq k}|\EE_n \epsilon_j (1+\tilde\zeta_k)(\hat F_k-F_k^*)X|\nonumber\\
&\lesssim \sqrt{\frac{\log M}{n}} \max_{1\leq i\leq n}\max_{1\leq k\leq M}|(\hat F_k-F_k^*)X^{(i)}|\nonumber\\
&\lesssim\sqrt{\frac{\log M}{n}}\cdot p\sqrt{\frac{\log (p\vee M)}{n}}.\label{eq_thm_projection_pf_I4_9}
\end{align}
where the last line follows from Theorem \ref{thm_hatF}, the assumption that $X$ is bounded, and the fact that $||v||_1 \leq \sqrt{p}||v||_2$ for any $p$-vector $v$.
For $s=2$ and $t=3$, 
\begin{align}
\max_{j\neq k}|\EE_n J_{2j}J_{3k}|&=\max_{j\neq k}|\EE_n \zeta_k B_j(Z-\tilde Z)^{\otimes 2}B_k^T|\nonumber\\
&\lesssim \sqrt{\frac{\log M}{n}} \max_{1\leq i\leq n}\max_{j\neq k}|B_k(Z^{(i)}-\tilde Z^{(i)})^{\otimes 2}B_j^T|\nonumber\\
&\lesssim\sqrt{\frac{\log M}{n}} {\frac{\log (n\vee M)}{p}}.\label{eq_thm_projection_pf_I4_10}
\end{align}
For $s=2$ and $t=4$, 
\begin{align}
\max_{j\neq k}|\EE_n J_{2j}J_{4k}|&=\max_{j\neq k}|\EE_n B_j(Z-\tilde Z)(1+\zeta_k) \bar\phi_k\Sigma^{-1/2}X|\nonumber\\
&\leq \max_{j\neq k}\Big\{|\EE_n B_j(Z-\tilde Z)\bar\phi_k\Sigma^{-1/2}X|+|\EE_n B_j(Z-\tilde Z)\zeta_k \bar\phi_k\Sigma^{-1/2}X|\Big\}.\nonumber
\end{align}
We notice that by the definition of $\tilde Z$, $B_j(Z-\tilde Z)\bar\phi_k\Sigma^{-1/2}X$ is mean 0 and sub-exponential with sub-exponential norm of order $p^{-3/2}$. Given $X$ and $Z$, $\zeta_k$ is sub-exponential with bounded sub-exponential norm. Thus,
$$
\max_{j\neq k}|\EE_n B_j(Z-\tilde Z)\bar\phi_k\Sigma^{-1/2}X|\lesssim \sqrt{\frac{\log M}{n}} \frac{1}{p^{3/2}},
$$
and
\begin{align*}
\max_{j\neq k}|\EE_n B_j(Z-\tilde Z)\zeta_k \bar\phi_k\Sigma^{-1/2}X|&\lesssim \sqrt{\frac{\log M}{n}}  \max_{1\leq i\leq n}\max_{j\neq k}|B_j(Z^{(i)}-\tilde Z^{(i)})\bar\phi_k\Sigma^{-1/2}X^{(i)}|\\
&\lesssim\sqrt{\frac{\log M}{n}} {\frac{\log (n\vee M)}{p^{3/2}}}.
\end{align*}
Combining above two terms, 
\begin{align}
\max_{j\neq k}|\EE_n J_{2j}J_{4k}|&\lesssim\sqrt{\frac{\log M}{n}} {\frac{\log (n\vee M)}{p^{3/2}}}.\label{eq_thm_projection_pf_I4_11}
\end{align}
For $s=2$ and $t=5$, 
\begin{align*}
\max_{j\neq k}|\EE_n J_{2j}J_{5k}|&=\max_{j\neq k}|\EE_n B_j(Z-\tilde Z)\tilde \eta_k (F_k^*X-\Theta_kX-B_kZ)^2|\nonumber\\
&\leq \max_{j\neq k}\sqrt{\EE_n \{B_j(Z-\tilde Z)\}^2} \sqrt{\EE_n\tilde \eta^2_k (F_k^*X-\Theta_kX-B_kZ)^4}.
\end{align*}
Since $\{B_j(Z-\tilde Z)\}^2$ is sub-exponential with norm of order $1/p$, we have
$$
\max_j|(\EE_n-\EE) \{B_j(Z-\tilde Z)\}^2|\lesssim \sqrt{\frac{\log M}{n}} \frac{1}{p},
$$
and $\EE \{B_j(Z-\tilde Z)\}^2\leq C/p$, which implies  $\EE_n \{B_j(Z-\tilde Z)\}^2\lesssim 1/p$. The previous argument (for $s=5$) implies $\EE_n\tilde \eta^2_k (F_k^*X-\Theta_kX-B_kZ)^4\lesssim 1/p^2$. As a result,
\begin{align}
\max_{j\neq k}|\EE_n J_{2j}J_{5k}|&\lesssim{\frac{1}{p^{3/2}}}.\label{eq_thm_projection_pf_I4_12}
\end{align}
Similarly, for $s=2$ and $t=6$, 
\begin{align}
\max_{j\neq k}|\EE_n J_{2j}J_{6k}|&= \max_{j\neq k}|\EE_n B_j(Z-\tilde Z)(1+\tilde\zeta_k)(\hat F_k-F_k^*)X|\nonumber\\
&\leq \max_{j\neq k}\sqrt{\EE_n \{B_j(Z-\tilde Z)\}^2} \sqrt{\EE_n(1+\tilde\zeta_k)^2\{(\hat F_k-F_k^*)X\}^2}\nonumber\\
&\lesssim \frac{1}{p^{1/2}}\sqrt{\frac{p\log (p\vee M)}{n}}.\label{eq_thm_projection_pf_I4_13}
\end{align}
For $s=3$ and $t=4$, 
\begin{align}
\max_{j\neq k}|\EE_n J_{3j}J_{4k}|&= \max_{j\neq k}|\EE_n \zeta_j B_j(Z-\tilde Z)(1+\zeta_k) \bar\phi_k\Sigma^{-1/2}X|\nonumber\\
&\lesssim\sqrt{\frac{\log M}{n}} {\frac{\log (n\vee M)}{p^{3/2}}},\label{eq_thm_projection_pf_I4_14}
\end{align}
where we use the same argument used for $s=2$ and $t=4$ as $\zeta_j(1+\zeta_k)$ is mean 0 and $1/2$-sub-exponential. 

For $s=3$ and $t=5$, $\zeta_j\tilde \eta_k$ is mean 0 and $1/2$-sub-exponential given $X,Z$, and thus
\begin{align}
\max_{j\neq k}|\EE_n J_{3j}J_{5k}|&= \max_{j\neq k}|\EE_n \zeta_j B_j(Z-\tilde Z)\tilde \eta_k (F_k^*X-\Theta_kX-B_kZ)^2|\nonumber\\
&\lesssim\sqrt{\frac{\log M}{n}} {\frac{(\log (n\vee M))^{3/2}}{p^{3/2}}}.\label{eq_thm_projection_pf_I4_15}
\end{align}
For $s=3$ and $t=6$, due to sample splitting, $\zeta_j(1+\tilde \zeta_k)$ is mean 0 and $1/2$-sub-exponential given $X,Z$ and $\hat F$, and thus
\begin{align}
\max_{j\neq k}|\EE_n J_{3j}J_{6k}|&= \max_{j\neq k}|\EE_n \zeta_j B_j(Z-\tilde Z)(1+\tilde \zeta_k) (F_k^*-\hat F_k)X|\nonumber\\
&\lesssim\sqrt{\frac{\log M}{n}}  \max_{1\leq i\leq n}\max_{1\leq j\leq M} |B_j(Z^{(i)}-\tilde Z^{(i)})| \max_{1\leq i\leq n} \max_{1\leq k\leq M}|(\hat F_k-F_k^*)X^{(i)}|\nonumber\\
&\lesssim\sqrt{\frac{\log M}{n}} \sqrt{{\frac{\log (n\vee M)}{p}}} \cdot p\sqrt{\frac{\log (p\vee M)}{n}}.\label{eq_thm_projection_pf_I4_16}
\end{align}
For $s=4$ and $t=5$,
\begin{align}
\max_{j\neq k}|\EE_n J_{4j}J_{5k}|&= \max_{j\neq k}|\EE_n (1+\zeta_j) \bar\phi_j\Sigma^{-1/2}X\tilde \eta_k (F_k^*X-\Theta_kX-B_kZ)^2|\nonumber\\
&\leq \max_{j\neq k}\sqrt{\EE_n (1+\zeta_j)^2 (\bar\phi_j\Sigma^{-1/2}X)^2}\sqrt{\EE_n \tilde \eta_k^2 (F_k^*X-\Theta_kX-B_kZ)^4}\nonumber\\
&\lesssim\frac{1}{p^2},\label{eq_thm_projection_pf_I4_17}
\end{align}
where the last step follows from the analysis for $s=4$ and $s=5$ above. 

For $s=4$ and $t=6$, we apply the same argument to derive
\begin{align}
\max_{j\neq k}|\EE_n J_{4j}J_{5k}|&= \max_{j\neq k}|\EE_n (1+\zeta_j) \bar\phi_j\Sigma^{-1/2}X(1+\tilde\zeta_k)(\hat F_k-F_k^*)X|\nonumber\\
&\leq \max_{j\neq k}\sqrt{\EE_n (1+\zeta_j)^2 (\bar\phi_j\Sigma^{-1/2}X)^2}\sqrt{\EE_n (1+\tilde\zeta_k)^2\{(\hat F_k-F_k^*)X\}^2}\nonumber\\
&\lesssim\frac{1}{p}\sqrt{\frac{p\log (p\vee M)}{n}}.\label{eq_thm_projection_pf_I4_18}
\end{align}
Finally, for $s=5$ and $t=6$
\begin{align}
\max_{j\neq k}|\EE_n J_{4j}J_{5k}|&= \max_{j\neq k}|\EE_n \tilde \eta_j (F_j^*X-\Theta_jX-B_jZ)^2(1+\tilde\zeta_k)(\hat F_k-F_k^*)X|\nonumber\\
&\leq \max_{j\neq k}\sqrt{\EE_n \tilde \eta_j^2 (F_j^*X-\Theta_jX-B_jZ)^4}\sqrt{\EE_n (1+\tilde\zeta_k)^2\{(\hat F_k-F_k^*)X\}^2}\nonumber\\
&\lesssim\frac{1}{p}\sqrt{\frac{p\log (p\vee M)}{n}}.\label{eq_thm_projection_pf_I4_19}
\end{align}
We have considered all the off-diagonal terms on the right hand side of (\ref{eq_thm_projection_pf_I4}). For the diagonal terms, $j=k$, we have the following bounds. For brevity, we skip the intermediate steps.   
\begin{align*}
    (s,t)=(3,3):~&\max_{1\leq j\leq M}\EE_n \zeta_j^2 \{B_j(Z-\tilde Z)\}^2 ~\lesssim~ \frac{1}{p} \\
    (s,t)=(4,4):~&\max_{1\leq j\leq M}\EE_n (1+\zeta_j)^2 \{\bar\phi_j\Sigma^{-1/2}X\}^2 ~\lesssim~ \frac{1}{p^2} \\
    (s,t)=(5,5):~&\max_{1\leq j\leq M}\EE_n \tilde \eta_j^2 (F_j^*X-\Theta_jX-B_jZ)^4\lesssim \frac{1}{p^2} \\
    (s,t)=(6,6):~&\max_{1\leq j\leq M}\EE_n (1+\tilde\zeta_j)^2\{(\hat F_j-F_j^*)X\}^2\lesssim \frac{p\log (p\vee M)}{n}\\
    (s,t)=(1,2):~&\max_{1\leq j\leq M}|\EE_n \epsilon_j B_j(Z-\tilde Z)|\lesssim \sqrt{\frac{\log M}{n}} \sqrt{\frac{\log (n\vee M)}{p}}\\
    (s,t)=(1,3):~&\max_{1\leq j\leq M}|\EE_n \epsilon_j\zeta_j B_j(Z-\tilde Z)|\lesssim \frac{1}{p^{1/2}} \\
    (s,t)=(1,4):~&\max_{1\leq j\leq M}|\EE_n \epsilon_j (1+\zeta_j) \bar\phi_j\Sigma^{-1/2}X|\lesssim \frac{1}{p}\\
    (s,t)=(1,5):~&\max_{1\leq j\leq M}|\EE_n \epsilon_j\tilde \eta_j (F_j^*X-\Theta_jX-B_jZ)^2|\lesssim \frac{1}{p}&\\
    (s,t)=(1,6):~&\max_{1\leq j\leq M}|\EE_n \epsilon_j (1+\tilde\zeta_j)(\hat F_j-F_j^*)X|\lesssim \sqrt{\frac{p\log (p\vee M)}{n}}\\
    (s,t)=(2,3):~&\max_{1\leq j\leq M}|\EE_n \zeta_j (B_j(Z-\tilde Z))^2|\lesssim \sqrt{\frac{\log M}{n}} \frac{\log (n\vee M)}{p} \\
    (s,t)=(2,4):~&\max_{1\leq j\leq M}|\EE_n (1+\zeta_j) B_j(Z-\tilde Z)\bar\phi_j\Sigma^{-1/2}X |\lesssim \frac{1}{p^{3/2}} \\
    (s,t)=(2,5):~&\max_{1\leq j\leq M}|\EE_n B_j(Z-\tilde Z)\tilde \eta_j (F_j^*X-\Theta_jX-B_jZ)^2|\lesssim \frac{1}{p^{3/2}} \\
    (s,t)=(2,6):~&\max_{1\leq j\leq M}|\EE_n (1+\tilde\zeta_j) B_j(Z-\tilde Z) (\hat F_j-F_j^*)X|\lesssim \sqrt{\frac{\log (p\vee M)}{n}}\\
    (s,t)=(3,4):~&\max_{1\leq j\leq M}|\EE_n \zeta_j B_j(Z-\tilde Z)(1+\zeta_j)\bar\phi_j\Sigma^{-1/2}X|\lesssim \frac{1}{p^{3/2}}\\
    (s,t)=(3,5):~&\max_{1\leq j\leq M}|\EE_n  \zeta_j B_j(Z-\tilde Z)\tilde \eta_j (F_j^*X-\Theta_jX-B_jZ)^2|\lesssim \frac{1}{p^{3/2}} \\
    (s,t)=(3,6):~&\max_{1\leq j\leq M}|\EE_n \zeta_j B_j(Z-\tilde Z)(1+\tilde\zeta_j)(\hat F_j-F_j^*)X|\lesssim \sqrt{\frac{\log (p\vee M)}{n}}
\end{align*}
\begin{align*}
     (s,t)=(4,5):~&\max_{1\leq j\leq M}|\EE_n (1+\zeta_j)\bar\phi_j\Sigma^{-1/2}X\tilde \eta_j (F_j^*X-\Theta_jX-B_jZ)^2|\lesssim \frac{1}{p^{2}}\\
     (s,t)=(4,6):~&\max_{1\leq j\leq M}|\EE_n (1+\zeta_j)\bar\phi_j\Sigma^{-1/2}X (1+\tilde\zeta_j)(\hat F_j-F_j^*)X|\lesssim \sqrt{\frac{\log (p\vee M)}{np}} \\
     (s,t)=(5,6):~&\max_{1\leq j\leq M}|\EE_n \tilde \eta_j (F_j^*X-\Theta_jX-B_jZ)^2 (1+\tilde\zeta_j)(\hat F_j-F_j^*)X|\lesssim \sqrt{\frac{\log (p\vee M)}{np}}
\end{align*}
It can be seen that the leading error among all the diagonal terms is 
$$
\frac{1}{p^{1/2}}+\sqrt{\frac{\log (p\vee M)}{n}}\Bigg(\sqrt{p}+\sqrt{\frac{\log (n\vee M)}{p}}+\frac{\log (n\vee M)}{p}\Bigg)=o_p(1). 
$$
By collecting the order of the errors in (\ref{eq_thm_projection_pf_I4_1})--(\ref{eq_thm_projection_pf_I4_19}), the leading error among all the off-diagonal terms is 
$$
\frac{1}{p^{3/2}}+\sqrt{\frac{\log (p\vee M)}{n}}\Bigg[1+\bigg(\frac{\log (n\vee M)}{p}\bigg)^{3/2}\Bigg]. 
$$
Plugging these above results into (\ref{eq_thm_projection_pf_I4}), we derive
\begin{align*}
I_4\lesssim &\sqrt{M}\Bigg\{\frac{1}{p^{1/2}}+\sqrt{\frac{\log (p\vee M)}{n}}\Bigg[\sqrt{p}+\sqrt{\frac{\log (n\vee M)}{p}}+\frac{\log (n\vee M)}{p}\Bigg]\Bigg\}\\
&~~+ M\Bigg\{\frac{1}{p^{3/2}}+\sqrt{\frac{\log (p\vee M)}{n}}\Bigg[1+\bigg(\frac{\log (n\vee M)}{p}\bigg)^{3/2}\Bigg]\Bigg\}.
\end{align*}
since there are $M$ diagonal terms and $O(M^2)$ off diagonal terms and thus we multiply the corresponding square roots when computing the Frobenius norm for $I_4$. From (\ref{eq_thm_projection_pf_1}) and (\ref{eq_thm_projection_pf_2}), after removing redundant terms, we finally obtain under $p\sqrt{p \log (p \vee M)/n}\cdot\log(n \vee M) = o(1)$ (assumed in Theorem \ref{thm_hatF}) and $\log (n \vee M) = O(p^3)$ the following:
$$
\|\hat{P}_{B}^{\perp}-P_{B}^{\perp}\|_{F}\lesssim \frac{p}{M}\|\hat\Sigma-\Sigma\|_F\lesssim \frac{p}{\sqrt{M}}+\frac{1}{\sqrt{p}}+p\sqrt{\frac{\log (p\vee M)}{n}}+\sqrt{\frac{K}{n}}.
$$
This completes the proof. 
\end{proof}

\subsection{Proof for Theorem \ref{thm_final}}
\begin{proof}
    For $\hat{\Theta}:=\hat{P}_B^{\perp}\hat{F}$, we have the following derivation:
    \begin{align*}
        \hat{\Theta}-\Theta~&=~ (\hat{P}_B^{\perp}- P_B^{\perp})\hat{F} + P_B^{\perp}(\hat{F}-F^*) + (P_B^{\perp}F^* - \Theta) \\
        &= (\hat{P}_B^{\perp}- P_B^{\perp})F^*+(\hat{P}_B^{\perp}- P_B^{\perp})(\hat{F}-F^*) + P_B^{\perp}(\hat{F}-F^*) + (P_B^{\perp}F^* - \Theta). 
 \end{align*}
 By Theorems \ref{thm_projection}, \ref{thm_hatF} and \ref{thm_population}, we have
        \begin{align*}
        \frac{1}{\sqrt{M}}\big|\big|\hat{\Theta}-\Theta \big|\big|_F ~&\leq~
        ||\hat{P}_B^{\perp}- P_B^{\perp}||_F\cdot \frac{||F^*||_{\text{op}}}{\sqrt{M}} ~+~ (\|\hat{P}_B^{\perp}- P_B^{\perp}\|_F+||P_B^{\perp}||_{\text{op}})\cdot\frac{1}{\sqrt{M}}||\hat{F}-F^*||_F \\
        &~~~~+~ \frac{1}{\sqrt{M}}||P_B^{\perp}F^* - \Theta||_F
        \\
&\lesssim \Bigg(\frac{p}{\sqrt{M}}+\frac{1}{\sqrt{p}}+p\sqrt{\frac{\log (p\vee M)}{n}}+\sqrt{\frac{K}{n}}\Bigg) \frac{||F^*||_{\text{op}}}{\sqrt{M}}+\bigg(1+\frac{p}{\sqrt{M}}\bigg)\sqrt{\frac{p\log (p\vee M)}{n}}+\frac{1}{p}.
    \end{align*}
\end{proof}

\subsection{Proof for Theorem \ref{thm_inference}}
\begin{proof}
For any $u\in\RR^{M}$ and $v\in\RR^p$ with $\|u\|_2=1$ and $\|v\|_2=1$, we can write
\begin{align}\label{eq_thm_inference_pf_0}
u^T(\hat\Theta-P_B^\perp F^*)v&=u^TP_B^\perp(\hat F-F^*)v+u^T(\hat P_B^\perp-P_B^\perp)\hat F v\nonumber\\
&=\underbrace{u^TP_B^\perp(\hat F-F^*)v}_{I_1}+\underbrace{u^T(\hat P_B^\perp-P_B^\perp) F^* v}_{I_2}+\underbrace{u^T(\hat P_B^\perp-P_B^\perp)(\hat F-F^*) v}_{I_3}.
\end{align}
For the first term $I_1$, by (\ref{eq_linear_2}) we can rewrite $I_1$ as
\begin{align}\label{eq_thm_inference_pf_1}
I_1&=\frac{1}{n}\sum_{i=1}^n u^TP_B^\perp h^{(i)}+\frac{1}{n}\sum_{i=1}^n u^TP_B^\perp \xi_1^{(i)}+\frac{1}{n}\sum_{i=1}^n u^TP_B^\perp \xi_2^{(i)},
\end{align}
where $\xi_1^{(i)}, \xi_2^{(i)}\in \RR^M$ with the $m$th entry being $\bar\epsilon^{(i)}_m v^T\Delta_{m1}X^{(i)}$ and $\bar\epsilon^{(i)}_m v^T\Delta_{m2}X^{(i)}$. Here,
\begin{align*}
   \Delta_{m1}&=(\frac{1}{n}\sum_{i=1}^n (1+\zeta^{(i)}_m)X^{(i)}X^{(i)T})^{-1}-G_m^{-1} \\
   \Delta_{m2}&=(\frac{1}{n}\sum_{i=1}^n (1+\tilde\zeta^{(i)}_m)X^{(i)}X^{(i)T})^{-1}-(\frac{1}{n}\sum_{i=1}^n (1+\zeta^{(i)}_m)X^{(i)}X^{(i)T})^{-1}
\end{align*}
where the same notation in (\ref{eq_linear_2}) is used here. In the following, we will first apply the Berry–Esseen theorem to the first term in (\ref{eq_thm_inference_pf_1}). Let $O_i=u^TP_B^\perp h^{(i)}$. The Berry–Esseen theorem yields that 
\begin{align}\label{eq_thm_inference_pf_2}
\sup_t\Bigg|\PP\bigg(\frac{\sum_{i=1}^n O_i}{s_n}\leq t\bigg)-\Phi(t)\Bigg|&\leq C\frac{\sum_{i=1}^n \rho_i}{s_n^3}, 
\end{align}
where $s_n^2=\sum_{i=1}^n \EE O_i^2$ and $\rho_i=\EE|O_i|^3$. Since $\EE O_i^2\geq C$, we have $s_n\geq \sqrt{Cn}$. By Hölder's inequality, $\rho_i\leq \{\EE O_i^4\}^{3/4}$. For $\EE O_i^4$, we first note that
\begin{align}\label{eq_thm_inference_pf_uP}
\|u^T P_B\|_2\leq \|u^T B\|_2 \|(B^TB)^{-1}B^T\|_{\text{op}}\leq \|u\|_1 \|B\|_{\infty, 2}\sqrt{\lambda_{\max}(B^TB)^{-1}}\lesssim R_n \sqrt{\frac{K}{M}},
\end{align}
where the last step follows from the factor model assumption, $\|u\|_1\leq R_n$ and the fact that each entry of $B$ is bounded by $C_4$ and therefore $\|B\|_{\infty, 2}\lesssim \sqrt{K}$. Recall that $h^{(i)}_m=\bar \epsilon^{(i)}_mv^T G^{-1}_mX^{(i)}$, where $\bar \epsilon^{(i)}_m$ is sub-exponential, and that $v^T G^{-1}_mX^{(i)}=v^T \Sigma_X^{-1/2}\tilde G^{-1}_m\tilde X^{(i)}$ is sub-Gaussian with bounded sub-Gaussian norm. Also, $\tilde X^{(i)}=\Sigma_X^{-1/2}X^{(i)}$ is sub-Gaussian with bounded sub-Gaussian norm, $\tilde G_m=\EE(1+\zeta_m^{(i)})\tilde X^{(i)}\tilde X^{(i)T}$ satisfies (\ref{eq_min_eigen_G}), and $\lambda_{\min}(\Sigma_X)\geq 1$ (WLOG). Writing $u^TP_B^\perp h^{(i)}=u^Th^{(i)}-u^TP_B h^{(i)}$,  we can show that   
\begin{align*}
\EE O_i^4&\lesssim \EE(u^T h^{(i)})^4+\EE(u^TP_B h^{(i)})^4\\
&\lesssim R_n^4 \EE(\|h^{(i)}\|^4_\infty)+ \EE[\|u^TP_B\|_2^4\|h^{(i)}\|_2^4]\\
&\lesssim R_n^4 ((\log M)^6\vee K^2).
\end{align*}
Therefore, (\ref{eq_thm_inference_pf_2}) implies
\begin{align}\label{eq_thm_inference_pf_3}
\sup_t\Bigg|\PP\bigg(\frac{\sum_{i=1}^n O_i}{s_n}\leq t\bigg)-\Phi(t)\Bigg|&\leq C\frac{R_n^3((\log M)^{9/2}\vee K^{3/2})}{\sqrt{n}}. 
\end{align}
For the second term in (\ref{eq_thm_inference_pf_1}), we have
\begin{align*}
\Big|\frac{1}{n}\sum_{i=1}^n u^TP_B^\perp \xi_1^{(i)}\Big|&\leq \Big|\frac{1}{n}\sum_{i=1}^n u^T \xi_1^{(i)}\Big|+\Big|\frac{1}{n}\sum_{i=1}^n u^TP_B \xi_1^{(i)}\Big|\\
&\leq \|u\|_1\max_m \Big|v^T\Delta_{m1}\frac{1}{n}\sum_{i=1}^n \bar\epsilon_m^{(i)} X^{(i)}\Big|+\|u^TP_B\|_2 \bigg\{\sum_{m=1}^M |v^T\Delta_{m1}\frac{1}{n}\sum_{i=1}^n \bar\epsilon_m^{(i)} X^{(i)}|^2\bigg\}^{1/2}. 
\end{align*}
By the proof of Theorem \ref{thm_hatF}, 
$$
\max_m \Big\|\frac{1}{n}\sum_{i=1}^n \bar\epsilon_m^{(i)} X^{(i)}\Big\|_\infty \lesssim \sqrt{\frac{\log (p\vee M)}{n}},
$$
and
\begin{align*}
\max_{m}\|\Delta_{m1}\|_{\text{op}}&=\max_{m}\Big\|(\frac{1}{n}\sum_{i=1}^n (1+\zeta^{(i)}_m)X^{(i)}X^{(i)T})^{-1}\Big\|_{\text{op}}\Big\|\frac{1}{n}\sum_{i=1}^n (1+\zeta^{(i)}_m)X^{(i)}X^{(i)T}-G_m\Big\|_{\text{op}}\|G_m^{-1}\|_{\text{op}}\\
&\lesssim p\sqrt{\frac{\log (p\vee M)}{n}}.
\end{align*}
As a result, we can show that
\begin{align*}
\|u\|_1\max_m \Big|v^T\Delta_{m1}\frac{1}{n}\sum_{i=1}^n \bar\epsilon_m^{(i)} X^{(i)}\Big|&\leq \|u\|_1\max_m \|v\|_2\|\Delta_{m1}\|_{\text{op}} \Big\|\frac{1}{n}\sum_{i=1}^n \bar\epsilon_m^{(i)} X^{(i)}\Big\|_2\\
&\leq \|u\|_1\max_m \|v\|_2\|\Delta_{m1}\|_{\text{op}} \Big\|\frac{1}{n}\sum_{i=1}^n \bar\epsilon_m^{(i)} X^{(i)}\Big\|_\infty p^{1/2}\\
&\lesssim R_n p^{3/2}\frac{\log (p\vee M)}{n}.
\end{align*}
Following a similar argument, 
\begin{align*}
\|u^TP_B\|_2 \bigg\{\sum_{m=1}^M |v^T\Delta_{m1}\frac{1}{n}\sum_{i=1}^n \bar\epsilon_m^{(i)} X^{(i)}|^2\bigg\}^{1/2}&\lesssim R_n \sqrt{\frac{K}{M}}\sqrt{M}\max_m \Big|v^T\Delta_{m1}\frac{1}{n}\sum_{i=1}^n \bar\epsilon_m^{(i)} X^{(i)}\Big|\\
&\lesssim R_n\sqrt{K} p^{3/2}\frac{\log (p\vee M)}{n},
\end{align*}
which implies that
\begin{align*}
\Big|\frac{1}{n}\sum_{i=1}^n u^TP_B^\perp \xi_1^{(i)}\Big|&\lesssim  R_n\sqrt{K} p^{3/2}\frac{\log (p\vee M)}{n}.
\end{align*}
For the third term in (\ref{eq_thm_inference_pf_1}), we have
\begin{align*}
\Big|\frac{1}{n}\sum_{i=1}^n u^TP_B^\perp \xi_2^{(i)}\Big|&\leq \Big|\frac{1}{n}\sum_{i=1}^n u^T \xi_2^{(i)}\Big|+\Big|\frac{1}{n}\sum_{i=1}^n u^TP_B \xi_2^{(i)}\Big|\\
&\leq \|u\|_1\max_m \Big|v^T\Delta_{m2}\frac{1}{n}\sum_{i=1}^n \bar\epsilon_m^{(i)} X^{(i)}\Big|+\|u^TP_B\|_2 \bigg\{\sum_{m=1}^M |v^T\Delta_{m2}\frac{1}{n}\sum_{i=1}^n \bar\epsilon_m^{(i)} X^{(i)}|^2\bigg\}^{1/2}\\
&\lesssim R_n\sqrt{K} p^{3/2}\frac{\log (p\vee M)\log (n\vee M)}{n}.
\end{align*}
Combined with (\ref{eq_thm_inference_pf_3}), we obtain the Berry–Esseen bound for $I_1$, 
\begin{align}\label{eq_thm_inference_pf_4}
\sup_t\Bigg|\PP\bigg(\frac{I_1}{s_n/n}\leq t\bigg)-\Phi(t)\Bigg|&\leq C\frac{R_n^3((\log M)^{9/2}\vee K^{3/2})}{\sqrt{n}}+R_n\sqrt{K} p^{3/2}\frac{\log (p\vee M)\log (n\vee M)}{\sqrt{n}}. 
\end{align}
It remains to bound $I_2$ and $I_3$. We note that by Lemma \ref{lem_proj_L2}
\begin{align*}
|I_2|\leq \|u\|_1 \|(\hat P_B-P_B) F^* v\|_\infty&\leq \|u\|_1 \max_j\|(\hat P_B-P_B)e_j\|_2 \|F^* v\|_2\\
&\lesssim R_n \|F^* v\|_2 \eta \sqrt{\frac{p}{M}},
\end{align*}
where $\eta$ is defined in (\ref{eq_eta}) and derived in Remark \ref{rem.sharper}. We then have
\begin{align*}
|I_3|\leq \|u\|_1 \|(\hat P_B-P_B) (\hat F-F^*) v\|_\infty&\leq \|u\|_1 \max_j\|(\hat P_B-P_B)e_j\|_2 \|(\hat F-F^*) v\|_2\\
&\lesssim R_n\eta \sqrt{\frac{p}{M}} \sqrt{M}\Big[\sqrt{\frac{1}{n}}+p^{3/2}\frac{\log (p\vee M)\log (n\vee M)}{n}\Big].
\end{align*}
As a result, the contribution to the  Berry–Esseen bound from $I_2$ and $I_3$ is given by
$$
\sqrt{n} I_2\lesssim R_n \|F^* v\|_2 \bigg(\sqrt{\frac{n}{pM}}+p\sqrt{\frac{\log (p\vee M)}{M}}+\frac{p\sqrt{n}}{M}\bigg),
$$
and 
$$
\sqrt{n} I_3\lesssim R_n \bigg(\frac{1}{\sqrt{p}}+\frac{p}{\sqrt{M}}+\frac{p\log (p\vee M)\log (n\vee M)}{\sqrt{n}}+\frac{p^{5/2}\log (p\vee M)\log (n\vee M)}{\sqrt{Mn}}\bigg).
$$
To show (\ref{thm_inference_2}), we decompose
\begin{align*}
\Bigg|\frac{s_n^2}{n}-\frac{\hat s_n^2}{n}\Bigg|&\leq \Bigg|\frac{1}{n}\sum_{i=1}^n \{(u^TP_B^\perp h^{(i)})^2-\EE(u^TP_B^\perp h^{(i)})^2\}\Bigg|+\Bigg|\frac{1}{n}\sum_{i=1}^n \{(u^TP_B^\perp h^{(i)})^2-(u^T\hat P_B^\perp\hat h^{(i)})^2\}\Bigg|,
\end{align*}
where we denote the above two terms as $J_1$ and $J_2$. Let 
$$
Q_i=(u^TP_B^\perp h^{(i)})^2-\EE(u^TP_B^\perp h^{(i)})^2.
$$ 
By the Markov inequality we have 
$$
J_1\lesssim \Big\{\EE\Big(\frac{1}{n}\sum_{i=1}^n Q_i\Big)^2\Big\}^{1/2}=\bigg\{\frac{\EE Q_i^2}{n}\bigg\}^{1/2}\lesssim \frac{R_n^2 ((\log M)^3\vee K)}{\sqrt{n}}.
$$
In addition, we can further decompose $J_2$ as
\begin{align*}
J_2&\leq \Bigg|\frac{1}{n}\sum_{i=1}^n \{(u^TP_B^\perp h^{(i)})^2-(u^TP_B^\perp\tilde h^{(i)})^2\}\Bigg|+\Bigg|\frac{1}{n}\sum_{i=1}^n \{(u^TP_B^\perp\tilde h^{(i)})^2-(u^T\hat P_B^\perp\hat h^{(i)})^2\}\Bigg|,
\end{align*}
where $\tilde h^{(i)}=(\tilde h^{(i)}_1,...,\tilde h^{(i)}_M)^T$ with $\tilde h^{(i)}_m=\bar \epsilon^{(i)}_m v^T \hat G^{-1}_mX^{(i)}$. We denote the above two terms as $J_{21}$ and $J_{22}$, respectively. We first note that
\begin{align}
\|u^TP_B^\perp\|_1&\leq \|u\|_1+ \|u^TP_B\|_1 \lesssim R_n + \sqrt{M}\|u^TP_B\|_2
\lesssim R_n \sqrt{K},\label{eq_thm_inference_pf_uP_max}
\end{align}
where we use (\ref{eq_thm_inference_pf_uP}). With a slight abuse of notation, let $\hat H_m=\frac{1}{n}\sum_{i=1}^n (1+\hat\zeta^{(i)}_m)\tilde X^{(i)}\tilde X^{(i)T}$, $\tilde H_m=\frac{1}{n}\sum_{i=1}^n (1+\zeta^{(i)}_m)\tilde X^{(i)}\tilde X^{(i)T}$ and $H_m=\EE (1+\zeta^{(i)}_m)\tilde X^{(i)}\tilde X^{(i)T}$, where $\tilde X^{(i)}=\Sigma_X^{-1/2} X^{(i)}$. Recall from the proof of Theorem \ref{thm_hatF} that 
\begin{align}\label{eq_hat_H}
\max_m \|\hat H_m-\tilde H_m\|_{\text{op}}\lesssim p\sqrt{\frac{\log (p\vee M)}{n}} \log (M\vee n)=o_p(1),
\end{align}
and 
\begin{align}\label{eq_tilde_H}
\max_m \|\tilde H_m- H_m\|_{\text{op}}\leq \|\tilde H_m-H_m\|_{F}\lesssim p\sqrt{\frac{\log (p\vee M)}{n}}=o_p(1),
\end{align}
and recall that the smallest and largest eigenvalues of $H_m$ are lower and upper bounded by constants as shown in  (\ref{eq_min_eigen_EH}). As a result, 
\begin{align*}
J_{21}&= \Bigg|\frac{1}{n}\sum_{i=1}^n u^TP_B^\perp (h^{(i)}-\tilde h^{(i)})\cdot u^TP_B^\perp (h^{(i)}+\tilde h^{(i)})\Bigg|\\
&=\Bigg|\frac{1}{n}\sum_{i=1}^n \sum_{m=1}^M\sum_{s=1}^M (u^TP_B^\perp)_m (u^TP_B^\perp)_s \bar \epsilon^{(i)}_m \bar \epsilon^{(i)}_s v^T(G_m^{-1}-\hat G_m^{-1})X^{(i)}X^{(i)T}(G_s^{-1}+\hat G_s^{-1})v\Bigg|\\
&=\Bigg|\frac{1}{n}\sum_{i=1}^n \sum_{m=1}^M\sum_{s=1}^M (u^TP_B^\perp)_m (u^TP_B^\perp)_s \bar \epsilon^{(i)}_m \bar \epsilon^{(i)}_s v^T\Sigma_X^{-1/2}\hat H_m^{-1}(\hat H_m-H_m)H_m^{-1}\tilde X^{(i)}\tilde X^{(i)T}\\
&~~~~~~~~~~~~~~~~~~~~~~~~~~~~~~~~~~~~~~~~~~~~~~~~~~~~~~~~~~~~~~~~\quad\quad\quad H_s^{-1}(\hat H_s+H_s)\hat H_s^{-1}\Sigma_X^{-1/2}v\Bigg|.
\end{align*}
To bound the above term, we first consider
\begin{align*}
&\max_{1\leq m,s\leq M}\Big|\frac{1}{n}\sum_{i=1}^n \bar \epsilon^{(i)}_m \bar \epsilon^{(i)}_s v^T\Sigma_X^{-1/2}\hat H_m^{-1}(\hat H_m-H_m)H_m^{-1}\tilde X^{(i)}\tilde X^{(i)T} H_s^{-1}(\hat H_s+H_s)\hat H_s^{-1}\Sigma_X^{-1/2}v\Big|\\
&\leq\max_{1\leq m,s\leq M} \|v\|_2^2\|\Sigma_X^{-1}\|_{\text{op}}\|\hat H_m^{-1}\|_{\text{op}}^2\|H_m^{-1}\|_{\text{op}}^2 \|\hat H_m-H_m\|_{\text{op}}\|\hat H_m+H_m\|_{\text{op}}\Big\|\frac{1}{n}\sum_{i=1}^n \bar \epsilon^{(i)}_m \bar \epsilon^{(i)}_s\tilde X^{(i)}\tilde X^{(i)T} \Big\|_{\text{op}}\\
&\lesssim p\sqrt{\frac{\log (p\vee M)}{n}} \log (M\vee n) \{\log (M\vee n)\}^2\Big\|\frac{1}{n}\sum_{i=1}^n \tilde X^{(i)}\tilde X^{(i)T} \Big\|_{\text{op}}\\
&\lesssim p\sqrt{\frac{\log (p\vee M)}{n}}\{\log (M\vee n)\}^3,
\end{align*}
where the third line follows from $\max_i\|\bar\epsilon^{(i)}\|_\infty\lesssim \log (M\vee n)$, (\ref{eq_hat_H}) and (\ref{eq_tilde_H}), and the last line is from 
$$
\Big\|\frac{1}{n}\sum_{i=1}^n \tilde X^{(i)}\tilde X^{(i)T} -I_p\Big\|_{\text{op}}\lesssim \|I_p\|_{\text{op}}\sqrt{\frac{p}{n}}=o_p(1).
$$
Therefore, 
$$
J_{21}\lesssim \|u^TP_B^\perp\|_1^2 \cdot p\sqrt{\frac{\log (p\vee M)}{n}}\{\log (M\vee n)\}^3\lesssim R^2_n K p\sqrt{\frac{\log (p\vee M)}{n}}\{\log (M\vee n)\}^3.
$$
To bound $J_{22}$, we can decompose $J_{22}$ as 
$$
J_{22}\leq \Bigg|\frac{1}{n}\sum_{i=1}^n \{(u^TP_B^\perp\tilde h^{(i)})^2-(u^T P_B^\perp\hat h^{(i)})^2\}\Bigg|+\Bigg|\frac{1}{n}\sum_{i=1}^n \{(u^TP_B^\perp\hat h^{(i)})^2-(u^T\hat P_B^\perp\hat h^{(i)})^2\}\Bigg|.
$$
We first note that $\max_i\|\bar\epsilon^{(i)}\|_\infty\lesssim \log (M\vee n)$, and $\max_i\|\hat\epsilon^{(i)}\|_\infty\lesssim \log (M\vee n)$ since we have 
\begin{align}\label{eq_epsilon_dif}
|\hat\epsilon^{(i)}_m-\bar\epsilon^{(i)}_m|\lesssim \log (M\vee n) \|\hat F_m-F^*_m\|_1\lesssim \log (M\vee n) p\sqrt{\frac{\log (p\vee M)}{n}}.
\end{align}
Following a similar argument for $J_{21}$, we have
\begin{align*}
&\Bigg|\frac{1}{n}\sum_{i=1}^n \{(u^TP_B^\perp\tilde h^{(i)})^2-(u^T P_B^\perp\hat h^{(i)})^2\}\Bigg|\\
&=\Bigg|\frac{1}{n}\sum_{i=1}^n \sum_{m=1}^M\sum_{s=1}^M (u^TP_B^\perp)_m (u^TP_B^\perp)_s (\bar \epsilon^{(i)}_m-\hat \epsilon^{(i)}_m) (\bar \epsilon^{(i)}_s+\hat \epsilon^{(i)}_s) v^T\hat G_m^{-1}X^{(i)}X^{(i)T}\hat G_s^{-1}v\Bigg|\\
&\lesssim \|u^TP_B^\perp\|_1^2\max_i\|\hat\epsilon^{(i)}+\bar\epsilon^{(i)}\|_\infty\|\hat\epsilon^{(i)}-\bar\epsilon^{(i)}\|_\infty\\
&\lesssim R_n^2 K p\{\log (M\vee n)\}^2\sqrt{\frac{\log (p\vee M)}{n}},
\end{align*}
where we use (\ref{eq_thm_inference_pf_uP_max}) and (\ref{eq_epsilon_dif}) in the final step. Finally, we further note that by Lemma \ref{lem_proj_L2}, with $\eta$ defined therein, we have
\begin{align}\label{eq_P_dif}
\|u^T(P_B^\perp-\hat P_B^\perp)\|_1\leq \sqrt{M}\|u^T(P_B^\perp-\hat P_B^\perp)\|_2\leq \sqrt{M}\|u\|_1\max_j\|e_j^T(P_B^\perp-\hat P_B^\perp)\|_2\lesssim R_n\eta\sqrt{p}.
\end{align}
Combined with  (\ref{eq_thm_inference_pf_uP_max}), $\|u^T\hat P_B^\perp\|_1\lesssim R_n\sqrt{K}+R_n\eta\sqrt{p}\lesssim R_n\sqrt{K}$, where indeed we have $\eta \sqrt{p}=o(1)$ under the assumption $\frac{p}{\sqrt{M}}=o(1)$ and $p\sqrt{\frac{\log (p\vee M)}{n}}=o(1)$. For the second term in the decomposition of $J_{22}$, we can obtain that
\begin{align*}
&\Bigg|\frac{1}{n}\sum_{i=1}^n \{(u^TP_B^\perp\hat h^{(i)})^2-(u^T\hat P_B^\perp\hat h^{(i)})^2\}\Bigg|\\
&=\Bigg|\frac{1}{n}\sum_{i=1}^n \sum_{m=1}^M\sum_{s=1}^M (u^TP_B^\perp-u^T\hat P_B^\perp)_m (u^TP_B^\perp+u^T\hat P_B^\perp)_s \hat \epsilon^{(i)}_m \hat \epsilon^{(i)}_s v^T\hat G_m^{-1}X^{(i)}X^{(i)T}\hat G_s^{-1}v\Bigg|\\
&\lesssim \|u^T(P_B^\perp-\hat P_B^\perp)\|_1 \|u^T(P_B^\perp+\hat P_B^\perp)\|_1\max_i\|\hat\epsilon^{(i)}\|^2_\infty\\
&\lesssim R_n\eta\sqrt{p} R_n\sqrt{K}\{\log (M\vee n)\}^2\\
&\lesssim R_n^2\sqrt{K}\{\log (M\vee n)\}^2 \bigg(\frac{p}{\sqrt{M}}+\frac{1}{\sqrt{p}}+p\sqrt{\frac{\log (p\vee M)}{n}}\bigg). 
\end{align*}
By the decomposition of $J_{22}$, we finally derive the bound
\begin{align*}
J_{22}&\lesssim R_n^2\{\log (M\vee n)\}^2\sqrt{K} \bigg(\frac{p}{\sqrt{M}}+\frac{1}{\sqrt{p}}+p\sqrt{\frac{K\log (p\vee M)}{n}}\bigg).
\end{align*}
Combining the bound for $J_1$, $J_{21}$ and $J_{22}$, we prove that
\begin{align*}
\Bigg|\frac{s_n^2}{n}-\frac{\hat s_n^2}{n}\Bigg|&\lesssim \frac{R_n^2 ((\log M)^3\vee K)}{\sqrt{n}}+R^2_n K p\sqrt{\frac{\log (p\vee M)}{n}}\{\log (M\vee n)\}^3\\
&~~~~+R_n^2\{\log (M\vee n)\}^2\sqrt{K} \bigg(\frac{p}{\sqrt{M}}+\frac{1}{\sqrt{p}}+p\sqrt{\frac{K\log (p\vee M)}{n}}\bigg)\\
&\lesssim R_n^2\{\log (M\vee n)\}^2\sqrt{K} \bigg(\frac{p}{\sqrt{M}}+\frac{1}{\sqrt{p}}+p\sqrt{\frac{K\log (p\vee M)}{n}}\log (M\vee n)\bigg).
\end{align*}
This completes the proof.
\end{proof}

\subsection{Supplementary Lemmas}

\begin{lemma}\label{thm_first_term}
    Under Assumptions \ref{assum_ident} - \ref{assum_pervasiveness}, we have
    \begin{align*}
        \lambda_{K}\Big[B(\Sigma_Z^{-1}+A^TA)^{-1}B^T\Big] ~\geq~ C\cdot\frac{M}{p}
    \end{align*}
    for some fixed constant $C>0$.
\end{lemma}
\begin{proof}
    \begin{align*}
        \lambda_{K}\Big[B(\Sigma_Z^{-1}+A^TA)^{-1}B^T\Big] ~&=~\lambda_{K}\Big[\big\{B(\Sigma_Z^{-1}+A^TA)^{-1/2}\big\}\big\{B(\Sigma_Z^{-1}+A^TA)^{-1/2}\big\}^T\Big],
    \end{align*}
    where $(\Sigma_Z^{-1}+A^TA)^{-1/2}$ is defined since $\Sigma_Z^{-1}+A^TA$ is symmetric positive definite. Since $\Sigma_Z^{-1}+A^TA$ has rank $K$, the above is equal to the smallest eigenvalue of its transpose as $XY$ shares the same non-zero eigenvalues as $YX$.
    \begin{align*}
    \lambda_{K}\Big[B(\Sigma_Z^{-1}+A^TA)^{-1}B^T\Big] ~&=~  \lambda_{\min}\Big[\big\{B(\Sigma_Z^{-1}+A^TA)^{-1/2}\big\}^T\big\{B(\Sigma_Z^{-1}+A^TA)^{-1/2}\big\}\Big] \\
    ~&=~ \lambda_{\min}\Big[(\Sigma_Z^{-1}+A^TA)^{-1/2}B^TB(\Sigma_Z^{-1}+A^TA)^{-1/2}\Big] \\
    ~&\geq~\lambda_{\min}\Big[(\Sigma_Z^{-1}+A^TA)^{-1}\Big]\cdot\lambda_{\min}(B^TB)\\
    ~&=~ \frac{\lambda_{\min}(B^TB)}{\lambda_{\max}\Big[(\Sigma_Z^{-1}+A^TA)\Big]} \\
    ~&\geq~ C\cdot \frac{M}{p},
    \end{align*}
where the last line follows from Assumption \ref{assum_pervasiveness}.
\end{proof}

\begin{lemma}\label{thm_sec_term}
    Under Assumption \ref{assum_bound} and \ref{assum_glm}, in the regime of $\log M < n$, we have the following bound:
    \begin{align*}
        \bigg|\bigg|\frac{1}{n}\bm\epsilon\bm\epsilon^T - \EE\Big[\frac{1}{n}\bm\epsilon\bm\epsilon^T\Big]\bigg|\bigg|_F ~~&\leq~~ C\cdot\sigma^4_{\epsilon, \max}\cdot M \cdot\sqrt{\frac{t}{n}}~~~\text{w.p. }\geq 1-\frac{2M^2}{e^t}
    \end{align*}
    for some fixed constant $C>0$. Thus, we have 
$$
\bigg|\bigg|\frac{1}{n}\bm\epsilon\bm\epsilon^T - \EE\Big[\frac{1}{n}\bm\epsilon\bm\epsilon^T\Big]\bigg|\bigg|_F=O_p\bigg(M\sqrt{\frac{\log M}{n}}\bigg).
$$
\end{lemma}
\begin{proof}
    The $(m,m')$-th element of $\frac{1}{n}\bm\epsilon\bm\epsilon^T - \EE\big[\frac{1}{n}\bm\epsilon\bm\epsilon^T\big]$ is $\sum_{i=1}^n\big\{ \epsilon_{m}^{(i)}\epsilon_{m'}^{(i)}-E[\epsilon_{m}^{(i)}\epsilon_{m'}^{(i)}]\big\}/n$. When $m=m'$, $\epsilon_{m}^{(i)2}-\EE[\epsilon_{m}^{(i)2}]$ is a centered $\frac{1}{2}$-sub-exponential random variable (defined in Definition \ref{alphasub}) with a norm bounded by $\sigma^4_{\epsilon, \max}$. By the concentration inequality in Lemma \ref{GenBern}, we have
    \begin{align*}
        \bigg|\frac{1}{n}\sum_{i=1}^n \Big\{\epsilon_{m}^{(i)}\epsilon_{m'}^{(i)}-E\big[\epsilon_{m}^{(i)}\epsilon_{m'}^{(i)}\big]\Big\}\bigg| ~&\leq~ \sigma^4_{\epsilon, \max} \cdot \max \Bigg\{ \sqrt{\frac{T}{n}},~\frac{T^2}{n}\Bigg\}~~~\text{w.p. }\geq 1-\frac{2}{e^{c\cdot T}} \\&\leq~ \sigma^4_{\epsilon, \max} \cdot \sqrt{\frac{T}{n}} ~~~\text{w.p. }\geq 1-\frac{2}{e^{c\cdot T}}.
    \end{align*}
    The same logic holds when $m\neq m'$, except now $\epsilon_m^{(i)}\epsilon_{m'}^{(i)}$ is a centered random variable. The upper bound for the average is exactly the same. Using the union bound over all $M^2$ elements in the matrix, we get the conclusion.
\end{proof}
\begin{lemma}\label{thm_third_term}
Under Assumptions \ref{assum_ident} - \ref{assum_pervasiveness} and the regime of $K<n$, we have the following bound:
        \begin{align*}
        \EE \bigg|\bigg|\frac{1}{n}B(\ZZ-\tilde{\ZZ})(\ZZ-\tilde{\ZZ})^TB^T - B~\EE\Big[\frac{1}{n}(\ZZ-\tilde{\ZZ})(\ZZ-\tilde{\ZZ})^T\Big]B^T\bigg|\bigg|_F ~~&\leq~~ C\cdot\frac{M}{p}\cdot\sqrt{\frac{K}{n}}
    \end{align*}
    for some fixed constant $C>0$. Thus, we have
$$
\bigg|\bigg|\frac{1}{n}B(\ZZ-\tilde{\ZZ})(\ZZ-\tilde{\ZZ})^TB^T - B~\EE\Big[\frac{1}{n}(\ZZ-\tilde{\ZZ})(\ZZ-\tilde{\ZZ})^T\Big]B^T\bigg|\bigg|_F=O_p\bigg(\frac{M}{p}\sqrt{\frac{K}{n}}\bigg). 
$$
\end{lemma}
\begin{proof}
    Note that for a $K \times 1$ centered sub-Gaussian random vector $V$, if we define $\hat{\Sigma}_V:=\sum_{i=1}^n V^{(i)}V^{(i)T}/n$ and $\Sigma_V:=\EE[VV^T]$, then we have the following bound from \cite{koltchinskii2017concentration}:
    \begin{align*}
        \EE ||\hat{\Sigma}_V-\Sigma_V||_{\text{op}}~&\leq~ C \cdot||\Sigma_V||_{\text{op}}\cdot \bigg(\sqrt{\frac{K}{n}} + \frac{K}{n}\bigg).
    \end{align*}
%Since $||\hat{\Sigma}_V-\Sigma_V||_F \leq \sqrt{K}\cdot||\hat{\Sigma}_V-\Sigma_V||_{\text{op}}$, 
If we consider $V=Z-\tilde{Z}$, we have 
\begin{align*}
    \EE \bigg|\bigg|\frac{1}{n}(\ZZ-\tilde{\ZZ})(\ZZ-\tilde{\ZZ})^T &- \EE\Big[(Z-\tilde{Z})(Z-\tilde{Z})^T\Big]\bigg|\bigg|_{\text{op}} \\ ~&\leq~ 
    C\cdot\bigg|\bigg|\EE\Big[(Z-\tilde{Z})(Z-\tilde{Z})^T\Big]\bigg|\bigg|_{\text{op}}\cdot\bigg(\sqrt{\frac{K}{n}} + \frac{K}{n}\bigg)\\
    ~&\leq~ \frac{C'}{p}\cdot\sqrt{\frac{K}{n}}.
\end{align*}
For the last step, we follow the proof of Lemma \ref{thm_first_term} to obtain
$$
\bigg|\bigg|\EE\Big[(Z-\tilde{Z})(Z-\tilde{Z})^T\Big]\bigg|\bigg|_{\text{op}}\leq C/p.
$$
So, we have the conclusion:
\begin{align*}
  \EE  \bigg|\bigg|\frac{1}{n}B(\ZZ-\tilde{\ZZ})(\ZZ-\tilde{\ZZ})^TB^T &- B~\EE\Big[\frac{1}{n}(\ZZ-\tilde{\ZZ})(\ZZ-\tilde{\ZZ})^T\Big]B^T\bigg|\bigg|_F \\
    ~&\leq~ ||B||_F^2\cdot \bigg|\bigg|\frac{1}{n}(\ZZ-\tilde{\ZZ})(\ZZ-\tilde{\ZZ})^T - \EE\Big[(Z-\tilde{Z})(Z-\tilde{Z})^T\Big]\bigg|\bigg|_{\text{op}} \\
    ~&\leq~ C\cdot\frac{M}{p}\cdot\sqrt{\frac{K}{n}}.
\end{align*}
This completes the proof. 
\end{proof}
\begin{remark}\label{rem.sharper}
    While Theorem \ref{thm_projection} controls the Frobenius norm of $\hat P_{B}-P_{B}$, it is not sharp enough to bound $ \max_j\|(\hat P_B-P_B)e_j\|_2$ required in the proof of Theorem \ref{thm_inference}. To this end, we provide the necessary bound in Lemma \ref{lem_proj_L2} and provide the derivation for the relevant rate, $\eta$, in this remark. We first apply spectral decomposition 
$$
\frac{1}{M}\hat\Sigma=\frac{1}{nM}\sum_{i=1}^n (\hat\epsilon^{(i)})^{\otimes 2}=\hat V\hat D^2\hat V^T,
$$
where $\hat V\in\RR^{M\times M}$ and $\hat D^2\in\RR^{M\times M}$ are the corresponding eigenvectors and eigenvalues (in non-increasing order). We define the estimator 
$$
\hat B=\hat V_K \hat D_K \sqrt{M},
$$
where $\hat D_K\in\RR^{K\times K}$ is the square root of the top $K$ eigenvalues and $\hat V_K\in\RR^{M\times K}$ is the matrix of corresponding eigenvectors. Then $\frac{1}{M}\hat\Sigma \hat V_K=\hat V_K\hat D_K^2$ and $\hat V_K=\frac{1}{M}\hat\Sigma \hat V_K \hat D_K^{-2}$, which implies 
$$
\hat B=\hat V_K \hat D_K \sqrt{M}=\frac{1}{M}\hat\Sigma \hat V_K \hat D_K^{-1}\sqrt{M}.
$$
Note that
$$
\hat\Sigma=\EE_n \{B(Z-\tilde Z)\}^{\otimes 2}+ \EE_n \epsilon^{\otimes 2}+ \EE_n[\hat \epsilon^{\otimes 2}-\epsilon^{\otimes 2}-\{B(Z-\tilde Z)\}^{\otimes 2}].
$$
Define $H=\frac{1}{M}\EE_n (Z-\tilde Z)^{\otimes 2}B^T\hat V_K \hat D_K^{-1}\sqrt{M}$. We have
$$
\hat B=BH+\frac{1}{\sqrt{M}}\Big\{ \EE_n \epsilon^{\otimes 2}+\EE_n[\hat \epsilon^{\otimes 2}-\epsilon^{\otimes 2}-\{B(Z-\tilde Z)\}^{\otimes 2}]\Big\}\hat V_K \hat D_K^{-1}. 
$$
Recall that $\Sigma=B(\Sigma_Z^{-1}+A^TA)^{-1}B^T$. By the proof of Lemma \ref{thm_first_term}, we know $C_1M/p\leq\lambda_k(\Sigma)\leq C_2M/p$. Since $\lambda_k(\hat\Sigma/M)=\lambda_k^2(\hat D)$, by Weyl's inequality and the proof of Theorem \ref{thm_projection}, under the assumption $K=o(n)$, $p=o(\sqrt{M})$, we have
$$
\Big|\lambda^2_k(\hat D)-\lambda_k(\frac{1}{M}\Sigma)\Big|\leq \frac{1}{M} \|\hat\Sigma-\Sigma\|_F=o_p(p^{-1})
$$
which implies $\lambda_k(\hat D) \asymp p^{-1/2}$. We further have
\begin{align}
\|e_j^T(\hat B-BH)\|_2&=\Big\|\frac{1}{\sqrt{M}}e_j^T\Big\{ \EE_n \epsilon^{\otimes 2}+\EE_n[\hat \epsilon^{\otimes 2}-\epsilon^{\otimes 2}-\{B(Z-\tilde Z)\}^{\otimes 2}]\Big\}\hat V_K \hat D_K^{-1}\Big\|_2\nonumber\\
&\lesssim \sqrt{\frac{p}{M}} \bigg\|\frac{1}{n}\sum_{i=1}^n \epsilon_j^{(i)}\epsilon^{(i)}+e_j^T\EE_n[\hat \epsilon^{\otimes 2}-\epsilon^{\otimes 2}-\{B(Z-\tilde Z)\}^{\otimes 2}]\bigg\|_2\nonumber\\
&\lesssim \sqrt{\frac{p}{M}}\bigg\{1+\sqrt{\frac{M\log M}{n}}+\sqrt{M}\Big(\frac{1}{p^{3/2}}+\sqrt{\frac{\log (p\vee M)}{n}}\Big[1+\Big(\frac{\log (n\vee M)}{p}\Big)^{3/2}\Big]\Big)\bigg\}\nonumber\\
&\lesssim \sqrt{\frac{p}{M}}+\frac{1}{p}+ \sqrt{\frac{p\log (M\vee p)}{n}}+\frac{\log^{3/2}(n\vee M)}{p}\sqrt{\frac{\log (M\vee p)}{n}}\nonumber\\
&\lesssim \sqrt{\frac{p}{M}}+\frac{1}{p}+ \sqrt{\frac{p\log (M\vee p)}{n}},\label{eq_rate_b_j}
\end{align}
where the third step follows from the proof of Theorem \ref{thm_projection} by just taking the sum of $M$ elements in one column of the $M \times M$ matrix in consideration rather than all $M^2$ elements. The last step is due to $\{\log (M\vee p)\}^5=O(n)$. Indeed, the above bound holds uniformly over $1\leq j\leq M$. For simplicity, we denote 
\begin{align}\label{eq_eta}
\eta=\sqrt{\frac{p}{M}}+\frac{1}{p}+ \sqrt{\frac{p\log (M\vee p)}{n}}.
\end{align}
\end{remark}

\begin{lemma}\label{lem_proj_L2}
Under the same assumptions as in Theorem \ref{thm_projection} and $K=o(n)$, $p=o(\sqrt{M})$, we have
\begin{align*}
\max_{1\leq j\leq M}\|(\hat P_B-P_B)e_j\|_2\lesssim \eta \sqrt{\frac{p}{M}}.
\end{align*}
\end{lemma}

\begin{proof}
Let $\tilde B=BH$. Since we can also write $P_B=\tilde B(\tilde B^T\tilde B)^{-1}\tilde B^T$, it is apparent from the identity $A^{-1} - B^{-1} = -A^{-1}(A-B)B^{-1}$ that
\begin{align*}
\|(\hat P_B-P_B)e_j\|_2&\lesssim \|(\tilde B-\hat B)(\tilde B^T\tilde B)^{-1}\tilde B^T e_j\|_2\\
&~~~+\|\hat B\{(\hat B^T\hat B)^{-1}-(\tilde B^T\tilde B)^{-1}\}\tilde B^T e_j\|_2+\|\hat B(\hat B^T\hat B)^{-1}(\tilde B-\hat B)^T e_j\|_2.
\end{align*}
We denote the above three terms on the right hand side as $I_1, I_2$, and $I_3$, respectively. For $I_1$, 
$$
I_1\leq \|\tilde B-\hat B\|_F \|H^{-1}\|_{\text{op}}\|(B^TB)^{-1}\|_{\text{op}}\|B^T e_j\|_2.
$$
We note that $\|B^T e_j\|_2\leq C_4$ by Assumption \ref{assum_pervasiveness}, $\|(B^T B)^{-1}\|_{op}\lesssim M^{-1}$ and $\|\tilde B-\hat B\|_F\lesssim \eta \sqrt{M}$ by (\ref{eq_rate_b_j}).  In addition, by Weyl's inequality and Lemma \ref{thm_third_term},
\begin{align*}
\lambda_{\min}(HH^T)&\geq C\frac{p}{M}\lambda_{\min} \Big(\EE_n (Z-\tilde Z)^{\otimes 2}B^T\Big)^{\otimes 2}\\
&\geq C'\frac{p}{M}\lambda_{\min} \Big((\Sigma_Z^{-1}+A^TA)^{-1}B^TB(\Sigma_Z^{-1}+A^TA)^{-1}\Big)\\
&\geq C''/p,
\end{align*}
which implies $\|H^{-1}\|_{\text{op}}\lesssim p^{1/2}$. Combining all these results, we obtain that 
$$
I_1\lesssim \eta \sqrt{M}p^{1/2} M^{-1}\lesssim \eta \sqrt{\frac{p}{M}}. 
$$
For $I_2$, note that the smallest eigenvalue of $\hat B^T\hat B$ is lower bounded by $M/p$ since for $\tilde{B}$ we have $\lambda_{\min}(\tilde{B}^T\tilde{B}) = \lambda_{\min}(H^TB^TBH) \geq \lambda_{\min}(B^TB)\lambda_{\min}(H^TH)\gtrsim \sqrt{M/p}\sqrt{M/p} = M/p$, and by Weyl's inequality we have $\lambda_{\min}(\hat{B}^T\hat{B}) \geq \lambda_{\min}(\tilde{B}^T\tilde{B}) - ||\hat{B}^T\hat{B} - \tilde{B}^T\tilde{B}||_{\text{op}}$, where the last term is upper bounded by $\big(||\hat{B}||_{\text{op}} + ||\tilde{B}||_{\text{op}}\big)||\hat{B}-\tilde{B}||_{F} = O(\sqrt{M/p})\cdot o(\sqrt{M/p}) = o(M/p)$ in the assumed regime of $p=o(\sqrt{M})$. Combined with the following:
$$
\|P_B e_j\|_2=\|B( B^T B)^{-1} B^T e_j\|_2\leq \|B( B^T B)^{-1}\|_{\text{op}} \|B^T e_j\|_2\lesssim M^{-1/2},
$$ 
we derive
\begin{align*}
I_2&=\|\hat B(\hat B^T\hat B)^{-1}\{\hat B^T\hat B-\tilde B^T\tilde B\}(\tilde B^T\tilde B)^{-1}\tilde B^T e_j\|_2\\
&\leq \|\hat B(\hat B^T\hat B)^{-1}\|_{\text{op}}\Big[\|(\hat B-\tilde B)^TP_B e_j\|_2+ \|\hat B^T(\hat B-\tilde B)^T(\tilde B^T\tilde B)^{-1}\tilde B^T e_j\|_2\Big]\\
&\lesssim \sqrt{\frac{p}{M}}\Bigg\{\sqrt{M}\eta \frac{1}{\sqrt{M}}+\sqrt{\frac{M}{p}}I_1\Bigg\}\\
&\lesssim \eta \sqrt{\frac{p}{M}}. 
\end{align*}
Finally, we can show that
$$
I_3\leq \|\hat B(\hat B^T\hat B)^{-1}\|_{\text{op}}\|(\tilde B-\hat B)^T e_j\|_2\lesssim \eta \sqrt{\frac{p}{M}}. 
$$
This completes the proof.
\end{proof}

\section{Additional Technical Results for Remark \ref{rem_1}}\label{sec_rem_1}

In this section, we present some results on the eigenvalue ratio used in Remark \ref{rem_1}. Recall that $\Sigma=B\big(\Sigma_Z^{-1}+A^TA\big)^{-1}B^T$. We know from Assumption \ref{assum_pervasiveness} that $\lambda_j\big(\Sigma\big)\asymp M/p$  for $1 \leq j \leq K$, and $\lambda_j\big(\Sigma\big)=0$ for $K+1\leq j\leq M$. Similar to the derivation in the proof for Theorem \ref{thm_projection}, we know by Weyl's inequality that for all $1 \leq j \leq M$, 
    \begin{align*}
        \big|\lambda_j(\hat{\Sigma}) - \lambda_j (\Sigma)\big|~&\leq~ \Big|\Big|\hat{\Sigma} - \Sigma\Big|\Big|_{\text{op}} \\
        ~&\lesssim \frac{M}{p}\Bigg(\frac{p}{\sqrt{M}}+\frac{1}{\sqrt{p}}+p\sqrt{\frac{\log (p\vee M)}{n}}+\sqrt{\frac{K}{n}}\Bigg).
    \end{align*}
Assume that 
$$
\frac{p}{\sqrt{M}}+\frac{1}{\sqrt{p}}+p\sqrt{\frac{\log (p\vee M)}{n}}+\sqrt{\frac{K}{n}}=o(1).
$$
Then $M/p$ dominates all of the terms in the above rate and we conclude that $\lambda_j(\hat{\Sigma})\asymp M/p$ for $1 \leq j \leq K$, and  $\lambda_j(\hat{\Sigma})= o(M/p)$ for $K+1 \leq j \leq M$. This implies that $\lambda_{j}(\hat{\Sigma})/\lambda_{j+1}(\hat{\Sigma}) \asymp 1$ for $1 \leq j \leq K-1$, and $\lambda_{K}(\hat{\Sigma})/\lambda_{K+1}(\hat{\Sigma}) \rightarrow \infty$. This implies $\hat K\geq K$. While the eigenvalue ratio approach in general cannot imply $\hat K=K$ with high probability, the numerical results in \cite{bing2022adaptive} show that the performance of the PCA based estimator is robust even if $\hat K$ is above $K$, as long as it's in a reasonable range.

\section{Other Useful Definitions and Inequalities}

\subsection{Davis-Kahan Theorem for Statisticians}
\begin{lemma} \label{DavisKahan}
 Let $\Sigma, \hat{\Sigma} \in \RR^{p \times p}$ be symmetric with eigenvalues $\lambda_1 \geq ... \geq \lambda_p$ and $\hat{\lambda}_1 \geq ... \geq \hat{\lambda}_p$. Fix $1 \leq r \leq s \leq p$ and assume that $\min(\lambda_{r-1} - \lambda_r,~ \lambda_s - \lambda_{s+1}) > 0$ where we define $\lambda_0 = \infty$ and $\lambda_{p+1} = -\infty$. Let $d=s-r+1$ and let $V = [v_r, v_{r+1}, ..., v_s] \in \RR^{p \times d}$ and $\hat{V} = [\hat{v}_r, \hat{v}_{r+1}, ..., \hat{v}_s] \in \RR^{p \times d}$ have orthonormal columns satisfying $\Sigma v_j = \lambda_j v_j$ and $\hat{\Sigma} \hat{v}_j = \hat{\lambda}_j \hat{v}_j$ for $j=r, r+1,...,s$. Then there exists an orthogonal matrix $\hat{O} \in \RR^{d \times d}$ such that 
\[ ||\hat{V}\hat{O} - V||_{F} ~\leq~ \frac{2^{3/2} \min\big(\sqrt{d}||\hat{\Sigma}-\Sigma||_{\text{op}}, ~||\hat{\Sigma}-\Sigma||_{F} \big)}{\min\big(\lambda_{r-1} - \lambda_r,~ \lambda_{s} - \lambda_{s+1}\big)}.\]
\end{lemma}
The full presentation of the theorem and its proof can be found in \cite{yu2015useful}.

\subsection{Maximal Inequality for Sub-Gaussian Random Variables}
\begin{lemma}\label{maxmax}
    Let $X_1,...,X_n$ be $n$ random variables such that $X_i \sim $subGaussian$(||X||^2_{\psi_2})$. Then,
    \begin{align*}
        \EE\Big[\underset{1\leq i \leq n}{\max}X_i\Big] ~\leq~ ||X||_{\psi_2}\cdot \sqrt{2\log n}&~,~~~~\EE\Big[\underset{1\leq i \leq n}{\max}|X_i|\Big] ~\leq~ ||X||_{\psi_2}\cdot \sqrt{2\log (2n)}\\
        \PP\Big[\underset{1\leq i \leq n}{\max}X_i > t\Big] ~\leq~ n\cdot e^{-\frac{t^2}{2\cdot||X||^2_{\psi_2}}}&~,~~~~\PP\Big[\underset{1\leq i \leq n}{\max}|X_i| > t\Big] ~\leq~ 2n\cdot e^{-\frac{t^2}{2\cdot||X||^2_{\psi_2}}}.
    \end{align*}
\end{lemma}

\subsection{Other Concentration Inequalities}
\begin{lemma} \label{Hoef}
    $\mathrm{(Hoeffding ~Inequality)}$ Let $X_1, ..., X_n$ be independent, mean-zero sub-Gaussian random variables, and let $a=(a_1,...,a_n) \in \mathbf{R}^n$. Then, for every $t \geq 0$, we have
    \[P\Bigg\{ \bigg|\sum_{i=1}^n a_iX_i \bigg| ~\geq~ t \Bigg\} \leq 2 \exp\bigg(-c_H \cdot \frac{t^2}{V^2 \cdot||a||_2^2} \bigg)\]
    where $V = \underset{1\leq i\leq n}{\max} ||X_i||_{\psi_2}$.
\end{lemma}
\begin{lemma} \label{Bern}
    $\mathrm{(Bernstein ~Inequality)}$ Let $X_1, ..., X_n$ be independent, mean-zero sub-exponential random variables, and let $a=(a_1,...,a_n) \in \mathbf{R}^n$. Then, for every $t \geq 0$, we have
    \[P\Bigg\{ \bigg|\sum_{i=1}^n a_iX_i \bigg| ~\geq~ t \Bigg\} \leq 2 \exp\bigg(-c_B \cdot \min \Big(\frac{t^2}{V^2 \cdot||a||_2^2}, \frac{t}{V\cdot||a||_{\max}} \Big) \bigg)\]
    where $V = \underset{1\leq i\leq n}{\max} ||X_i||_{\psi_1}$.
\end{lemma}
The proof for this lemma can be found in \cite{vershynin2018high}. Note that if you set $T=\min \Big(\frac{t^2}{V^2 \cdot||a||_2^2}, \frac{t}{V\cdot||a||_{\max}} \Big)$, this can be rewritten as:
\[\bigg|\sum_{i=1}^n a_iX_i \bigg| ~\leq~ V\cdot \max \Big( \sqrt{T}\cdot||a||_2,~ T\cdot ||a||_{\max}\Big) ~~~\text{w.p.} ~ \geq 1-2\exp\big(-c_B\cdot T\big).\]
\begin{defn} \label{alphasub}
$(\alpha-\mathrm{Sub}$-$\mathrm{exponential~Random~Variables)}$ We say that a random variable $X$ is $\alpha$-sub-exponential if there exists $K_{\alpha}>0$ such that $P(|X| \geq t) \leq 2 \exp (-\frac{t^{\alpha}}{K_{\alpha}^{\alpha}})$ for all $t \geq 0$. We define the $\alpha$-sub-exponential norm as $||X||_{\psi_{\alpha}} := \underset{p \geq 1}{\sup} ~\frac{1}{p^{1/\alpha}}\cdot \big\{E(|X|^p)\big\}^{1/p}$. 
\end{defn}
Note that it follows that $E(|X|^2) \leq 2^{\frac{2}{\alpha}}\cdot ||X||_{\psi_{\alpha}}^2$. Thus, for a fixed $\alpha \in (0,1]$, a centered $\alpha$-sub-exponential random variable with a finite $\alpha$-sub-exponential norm has a bounded variance. 
\begin{lemma} \label{GenBern}
$(\mathrm{Concentration~for~}\alpha-\mathrm{Sub}$-$\mathrm{exponential~Random~Variables)}$ Let $X_1, ..., X_n$ be independent, mean-zero $\alpha$-sub-exponential random variables satisfying $||X_i||_{\psi_{\alpha}} \leq V$, for some $\alpha \in (0,1]$. Let $a \in \mathbf{R}^n$. For any $t > 0$, we have
\[ P\Bigg\{ \bigg|\sum_{i=1}^n a_iX_i \bigg| ~\geq~ t \Bigg\} \leq 2 \exp\bigg(-c_{\alpha} \cdot \min \Big(\frac{t^2}{V^2 \cdot||a||_2^2}, \frac{t^{\alpha}}{V^{\alpha}\cdot||a||_{\max}^{\alpha}} \Big) \bigg). \]
\end{lemma}
The proof for Lemma \ref{GenBern} can be found in \cite{gotze2021concentration}. In a similar fashion to the lemma above, this can be rewritten as: 
\[\bigg|\sum_{i=1}^n a_iX_i \bigg| ~\leq~ V\cdot \max \Big( \sqrt{T}\cdot||a||_2,~ T^{\frac{1}{\alpha}}\cdot ||a||_{\max}\Big) ~~~\text{w.p.} ~ \geq 1-2\exp\big(-c_{\alpha}\cdot T\big).\]

In particular, by taking $T=\log M$ and plugging in $a = [1/n, ..., 1/n]$, we obtain 
$$
\bigg|\frac{1}{n}\sum_{i=1}^n X_i \bigg|\lesssim V\sqrt{\frac{\log M}{n}},
$$
provided $(\log M)^{(2-\alpha)/\alpha}/n=O(1)$. 

\begin{lemma} \label{GenBern2}
$(\mathrm{Concentration~of~the~Euclidean~Norm~of~}\alpha-\mathrm{Sub}$-$\mathrm{exponential~Random~Variables)}$ Let $X_1, ..., X_n$ be independent, mean-zero $\alpha$-sub-exponential random variables satisfying $||X_i||_{\psi_{\alpha}} \leq V$, for some $\alpha \in (0,1]$ and $E(X^2_{i}) \leq w_i^2$. For an $n \times n$ matrix $Q$, let $A=Q^TQ=(a_{ij})$. Then for any $t>0$, we have
\begin{align*} \Bigg| \Big|\Big| QX \Big|\Big|_2^2 - \sum_{i=1 }\bigg(w_i^2\sum_{j=1}^n q_{ji}^2\bigg)\Bigg| &~\leq~ V^2 \max \Big( \sqrt{t} \cdot ||A||_F,~~ t \cdot||A||_{\text{op}},~~ t^{\frac{2 + \alpha}{2\alpha}}\cdot\underset{1 \leq i \leq n}{\max}||(a_{ij})_j||_2,~~ t^{\frac{2}{\alpha}}\cdot||A||_{\max} \Big) \\
& \hspace{3cm} \text{w.p.~} \geq 1-2\exp\Big(-\frac{t}{C_{\alpha}}\Big).
\end{align*}
\end{lemma}
The proof for Lemma \ref{GenBern2} can also be found in \cite{gotze2021concentration}.

\section{NHANES dataset}\label{sec_NHANES}
We now apply our \textsc{G-hive} procedure to the 2017-2018 NHANES dataset through the \texttt{R} package \texttt{nhanesA} \citep{nhanesA}. This is a public dataset that consists of a combination of demographic data (age, gender, etc.), examination data (height, weight, etc.), questionnaire responses (\say{Do you now smoke cigarettes every day, some days, or not at all?}), and laboratory results (quantity of biomarkers taken from blood or urine samples). We take a subset of this data such that each response variable is binary and each covariate is continuous, and apply our \textsc{data driven G-hive} method on it. We focused on five binary responses corresponding to whether the individuals were diagnosed with depression, hypertension, arthritis, diabetes, and osteoporosis. We included age, income-to-poverty-ratio, systolic blood pressure (mmHg), body mass index (BMI, kg/$\text{m}^2$), and fasting glucose levels (mg/dL) as the observed covariates. After removing observations that have missing values for any of the response variables or covariates, we were left with $n=320$ valid observations. All of the observed covariates were standardized prior to the data analysis to prevent variables that are on a larger scale from dominating the model. The results with \textsc{data driven G-hive} for the coefficients $\widehat{\Theta}$ are presented in Table \ref{table_2}.

\begin{table}[ht]
  \centering
  \begin{tabular}{lrrrrr}
    \toprule
                 &    Age   &  Income  &   SysBP   &    BMI    & Glucose  \\
    \midrule
    Depression   &  0.0240  & -0.4243  & -0.2453   &  0.5891   & -0.7752  \\
    Diabetes     &  0.3589  &  0.1676  &  0.1695   &  0.4284   &  0.9332  \\
    Osteoporosis & -0.0595  &  0.0475  &  0.0286   & -0.1754   &  0.0643  \\
    Hypertension &  0.3733  &  0.1202  &  0.4070   &  0.5687   &  0.6886  \\
    Arthritis    &  0.3566  & -0.0667  & -0.2127   &  0.7989   & -0.3615  \\
    \bottomrule
  \end{tabular}
  \caption{\textsc{G-hive} results for the coefficient values relating the covariates to the response variables in the NHANES dataset with $n=320$ observations, $M=5$ binary response variables and $p=5$ observed covariates.}
  \label{table_2}
\end{table}

While it is infeasible in general to discern the ground truth coefficient values for real data analyses, it is apparent that many coefficients are well aligned with basic knowledge of the relationships between the covariates and the response variables. For instance, fasting glucose levels are known to be higher in individuals with diabetes (\cite{american20252}), and this is reflected in the large positive coefficient value of $0.9332$. Also, it is well known that BMI is a good predictor of diabetes (\cite{hu2001diet}), and this is consistent with our relatively large positive value of $0.4284$. Another example of well-alignment with the medical literature is the large positive coefficient value of $0.4070$ that relates systolic blood pressure to hypertension as hypertension is normally diagnosed with a combination of systolic blood pressure values and diastolic blood pressure values (\cite{parikh2008risk}). Lastly, age is shown to be positively correlated with arthritis, which is consistent with common knowledge in medicine as well (\cite{elgaddal2024arthritis}).

% ---------- Tables 1 & 2 side by side ----------
\begin{table}[ht]
\centering
\caption{Results of \textsc{naive mle} and \textsc{G-hive} on the reduced model. The estimates in the second column represent the effect of ``Income" on the response variables \textbf{without} taking into account the effect of ``Age".}
\begin{subtable}[t]{0.48\textwidth}
\centering
\caption{Results of \textsc{naive mle} on the reduced model.}
\label{table_reduced_1}
\begin{tabular}{l S[table-format=1.2] S[table-format=1.2]}
\toprule
Outcome & {Intercept} & {Income} \\
\midrule
Smoke        & -1.45 & -0.63 \\
Diabetes     & -1.60 & -0.17 \\
Hypertension & -0.65 &  0.16 \\
Arthritis    & -1.14 &  0.19 \\
Depression   & -1.31 & -0.33 \\
\bottomrule
\end{tabular}
\end{subtable}\hfill
\begin{subtable}[t]{0.48\textwidth}
\centering
\caption{Results of \textsc{G-hive} on the reduced model.}
\label{table_reduced_2}
\begin{tabular}{l S[table-format=1.2] S[table-format=1.2]}
\toprule
Outcome & {Intercept} & {Income} \\
\midrule
Smoke        & -1.90 & -0.47 \\
Diabetes     & -1.02 & -0.38 \\
Hypertension & -0.45 &  0.09 \\
Arthritis    & -1.05 &  0.14 \\
Depression   & -1.23 & -0.35 \\
\bottomrule
\end{tabular}
\end{subtable}\label{table_reduced}
\end{table}

% ---------- Table 3: Full model (Age + Income) ----------
\begin{table}[ht]
\centering
\caption{Results of \textsc{naive mle} on the full model. The estimates in the third column represent the effect of ``Income" on the response variables \textbf{while} taking into account the effect of ``Age".}
\begin{tabular}{l S[table-format=1.2] S[table-format=1.2] S[table-format=1.2]}
\toprule
Outcome & {Intercept} & {Age} & {Income} \\
\midrule
Smoke        & -1.52 & -0.51 & -0.54 \\
Diabetes     & -1.91 &  1.03 & -0.36 \\
Hypertension & -0.73 &  0.73 &  0.04 \\
Arthritis    & -1.32 &  0.87 &  0.07 \\
Depression   & -1.32 & -0.23 & -0.29 \\
\bottomrule
\end{tabular}\label{table_full}
\end{table}

\subsection{Deconfounding benefits of G-HIVE on the NHANES dataset}
We also present real data analysis results that highlight the ability of \textsc{G-hive} to account for model misspecification bias in the context of confounding. Recall that our method assumes that $Y$ depends on $X,Z$ through the GLM in (\ref{glm}) and that $X$ depends on the hidden variable $Z$ through the factor model $X=AZ+W$ in (\ref{factor_model}). This has a very natural connection to the basic confounded model depicted in Figure \ref{figure_confound_1}. Because $Z$ affects $X$ and $Y$, the observational association $P(Y|X)$ mixes the effect of $X \rightarrow Y$ and the spurious flow through $Z$ (\cite{pearl2009causality}). We utilize the same NHANES dataset from the previous section with slightly different variables that yield $n=230$ viable observations (with no missing values, etc.). We focus on two explanatory variables, \say{Age} and \say{Income} to clearly see the effects of confounding from a hidden variable (\say{Age}) and deconfounding with \textsc{G-hive}. Figure \ref{figure_confound_2} shows the basic confounded model with \say{Hypertension} as an example response variable. Figure \ref{figure_confound_2} is reasonable as it is well established that \say{Age} has a positive effect on \say{Income} (\cite{mincer1974schooling}) and that \say{Age} has a positive effect on \say{Hypertension} (\cite{parikh2008risk}). The same can be said about the effect of \say{Age} on \say{Diabetes} (\cite{wilson2007prediction}) and \say{Age} on \say{Arthritis} (\cite{elgaddal2024arthritis}), hence justifying similar figures with these response variables included instead. We run both \textsc{naive mle} and \textsc{G-hive} on the reduced model that just includes \say{Income} as the covariate, and we run \textsc{naive mle} on the full model that includes both \say{Age} and \say{Income} as covariates. It is expected that the confounding will cause the coefficient corresponding to \say{Income} in the reduced model to appear more positive than it truly is in the full model. The coefficient values for each of these models are shown in Tables \ref{table_reduced} and \ref{table_full}. We assume the latter model shows the \say{true} effect of \say{Income} on \say{Hypertension}, \say{Diabetes}, and \say{Arthritis} after accounting for the effect of \say{Age}. Comparing the coefficient values corresponding to \say{Income} on \say{Arthritis} between Table \ref{table_reduced_1} and Table \ref{table_reduced_2}, it is apparent that \textsc{G-hive}'s $0.14$ is closer to the \say{true} value of $0.07$ compared to the more positively pushed \textsc{naive mle} value of $0.19$. Similarly, for \say{Hypertension}, \textsc{G-hive}'s $0.09$ is closer to the \say{true} value of $0.04$ compared to the positively shifted \textsc{naive mle} value of $0.16$. Lastly, for \say{Diabetes}, \textsc{G-hive}'s $-0.38$ is closer to the \say{true} value of $-0.36$ compared to the positively shifted \textsc{naive mle} value of $-0.17$. This demonstrates that unlike \textsc{naive mle}, \textsc{G-hive} is able to account for confounding effects from the hidden variables and obtain estimates that are closer to the unconfounded effects even in real datasets. 

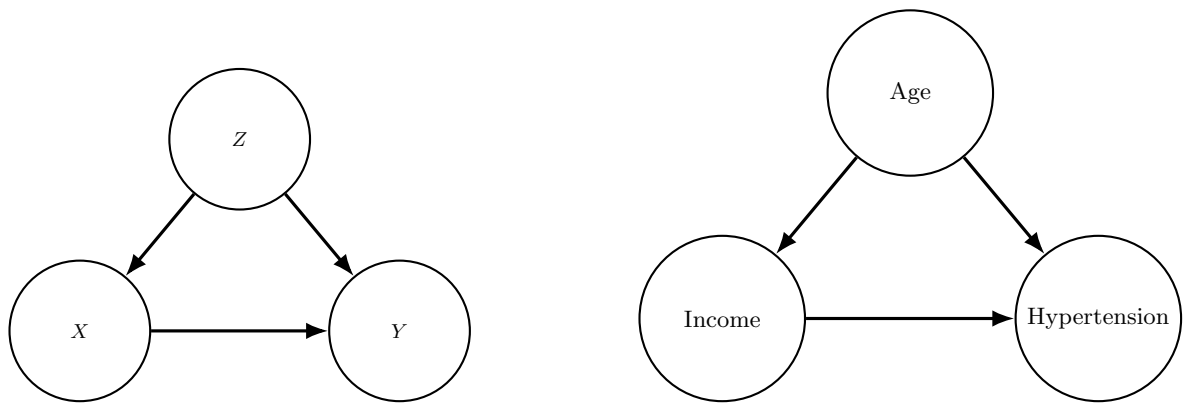
\begin{figure}[ht]
\centering

% ---- Shared (larger) styles ----
\tikzset{
  >=Latex,
  dagnode/.style={
    circle, draw, thick,
    minimum size=22mm,   % larger circle
    text width=20mm,     % helps center longer labels
    align=center,
    inner sep=1pt,
    font=\small
  },
  solidedge/.style={-Latex, very thick},
  dashededge/.style={-Latex, very thick, dashed}
}

% ---------- (a) X–Y–Z DAG (left, slightly smaller via scaling) ----------
\begin{subfigure}[t]{0.46\textwidth}
\centering
\begin{tikzpicture}[scale=0.85, every node/.style={transform shape}]
  % Nodes at triangle vertices
  \node[dagnode] (Z) at (2.5,0)   {$Z$};
  \node[dagnode] (X) at (0,-3)    {$X$};
  \node[dagnode] (Y) at (5,-3)    {$Y$};

  % Edges
  \draw[solidedge]  (Z) -- (X);   % Z -> X
  \draw[solidedge]  (Z) -- (Y);   % Z -> Y
  \draw[solidedge] (X) -- (Y);   % X -> Y (dashed)
\end{tikzpicture}
\caption{Models (\ref{glm}) and (\ref{factor_model}) represented as a basic confounded model.}
\label{figure_confound_1}
\end{subfigure}\hfill
% ---------- (b) Age–Income–Hypertension DAG (right, original size) ----------
\begin{subfigure}[t]{0.46\textwidth}
\centering
\begin{tikzpicture}
  % Nodes at triangle vertices
  \node[dagnode] (age) at (2.5,0)   {Age};
  \node[dagnode] (inc) at (0,-3)    {Income};
  \node[dagnode] (htn) at (5,-3)    {Hy\-per\-ten\-sion}; % hyphenation for perfect centering

  % Edges
  \draw[solidedge]  (age) -- (inc); % Age -> Income
  \draw[solidedge]  (age) -- (htn); % Age -> Hypertension
  \draw[solidedge] (inc) -- (htn); % Income -> Hypertension (dashed)
\end{tikzpicture}
\caption{The basic confounded model applied to variables in the NHANES dataset from 2017-2018.}
\label{figure_confound_2}
\end{subfigure}

\caption{Models (\ref{glm}) and (\ref{factor_model}) and variables in the real dataset NHANES (2017-2018) in the context of confounding (\cite{pearl2009causality}).}
\end{figure}

\end{document}